\documentclass[12pt, full]{article}

\usepackage{amsthm}
\usepackage{graphicx}
\usepackage[utf8]{inputenc} 
\usepackage[T1]{fontenc}    
\usepackage{hyperref}       
\usepackage{url}            
\usepackage{standalone}     
\usepackage[round]{natbib}  
\usepackage{booktabs}       
\usepackage{amsfonts}       
\usepackage{nicefrac}       
\usepackage{microtype}      
\usepackage[shortlabels]{enumitem} 
\usepackage{graphicx}       
\usepackage{multirow}
\usepackage{amsmath}
\usepackage{amssymb}
\usepackage{amsthm}
\usepackage{amssymb}        
\usepackage{amsmath,amsthm} 
\usepackage{bbm}            
\usepackage[titletoc, title, page]{appendix} 
\newtheorem{theorem}{Theorem}
\newtheorem{lemma}{Lemma}
\newtheorem{corollary}{Corollary}
\newtheorem{assumption}{Assumption}
\newtheorem{innercustomassump}{Assumption}
\newenvironment{customassump}[1]
  {\renewcommand\theinnercustomassump{#1}\innercustomassump}
  {\endinnercustomassump}

\theoremstyle{definition}

\theoremstyle{remark}

\usepackage{graphicx}
\newcommand{\Lim}[1]{\raisebox{0.5ex}{\scalebox{0.8}{$\displaystyle \lim_{#1}\;$}}}
\newcommand{\E}{\mathbb{E}}

\title{Difference-in-Discontinuities: Estimation, Inference and Validity Tests}
\author{Pedro Picchetti\thanks{Pontificia Universidad Católica de Chile.} \\ \small \href{mailto:pedrop3@al.insper.edu.br}{pedro.picchetti@uc.cl} \normalsize
    \and Cristine C. X. Pinto\thanks{INSPER Institute of Research and Education.}  \\ \small \href{mailto:cristinecxp@insper.edu.br}{cristinecxp@insper.edu.br} \normalsize
    \and Stéphanie T. Shinoki\footnotemark[2] \\ \small \href{mailto:stephaniets@al.insper.edu.br}{stephaniets@al.insper.edu.br} \normalsize}

\begin{document}
\maketitle
\begin{abstract}
    This paper provides a formal econometric framework behind the newly developed difference-in-discontinuities design (DiDC). Despite its increasing use in applied research, there are currently limited studies of its properties. We formalize the theory behind the difference-in-discontinuity approach by stating the identification assumptions, proposing a nonparametric estimator, and deriving its asymptotic properties. We also provide comprehensive tests for one of the identification assumption of the DiDC and sensitivity analysis methods that allow researchers to evaluate the robustness of DiDC estimates under violations of the identifying assumptions. Monte Carlo simulation studies show that the estimators have desirable finite-sample properties. Finally, we revisit \cite{grembi2016fiscal}, which studies the effects of relaxing fiscal rules on public finance outcomes. Our results show that most of the qualitative takeaways of the original work are robust to time-varying confounding effects.
\end{abstract}
\normalsize

\section{Introduction}

The difference-in-discontinuities design (DiDC) aims to address the limitations of both regression discontinuity designs (RDD) and difference-in-differences designs (DiD) by combining temporal and discontinuity-based sources of variation from the data-generating process (DGP). \cite{grembi2016fiscal} and \cite{eggers2018regression} proposed this quasi-experimental approach, using differences between the pre- and post-treatment periods around a threshold for both the treated and the untreated groups. Despite its increasing use in applied microeconomics \citep*{azuaga2017violencia, chicoine2017homicides, garcia2022plant, albright2024hidden}, the econometric theory of DiDC (identification, estimation, and asymptotic behavior) remains limited. \cite{leventer2025correcting} is one of the few papers that study the identification of regression discontinuity in the case of violation of the continuity of potential outcomes assumption and multiple periods. Their estimator can be an alternative to the DiDC. 

The method offers more flexibility over the standard RDD in some settings and can often be used in contexts where neither the RDD nor the DiD is applicable. In particular, the DiDC can handle cases where the control and treatment groups differ significantly and do not satisfy the parallel trends assumption of the DiD, or when multiple time-invariant confounders are at or near the threshold in an RDD setting. Additionally, by incorporating more information into the estimation, the DiDC eliminates bias in RDD estimates, providing more accurate and reliable estimates of treatment effects under certain assumptions.

The advantages of the DiDC have been explored in several applied microeconomic studies. For example, \cite{azuaga2017violencia} use the DiDC to analyze the impact of introducing legislation on domestic violence in Brazil. \cite{butts2023geographic} explore the use of Diff-in-Disc in geographic settings, \cite{chicoine2017homicides} investigates the expiration of the Assault Weapon Ban (AWB) by comparing municipalities in which the incumbent mayoral party wins a close election with those where the incumbent is defeated, and \cite{albright2024hidden} isolates the causal effects of algorithmic recommendations on decision-makers using DiDC. These real-life examples demonstrate how effective this approach can be in various research settings, highlighting its value as a practical analysis method.

In this paper, we develop the econometric theory for the Difference-in-Discontinuity design. While \cite{galindo2021fuzzy} develop an identification theory for the \textit{fuzzy} difference-in-discontinuity design based on the difference of RDD estimations, this work will focus on the sharp design, developing identification, inference and validity tests for the DiDC design. \cite{leventer2025correcting} also provide an alternative method to the DiDC. In fact, they use the time to address violations of the RDD's continuity assumption (e.g., the confounding effect at the threshold). However, they do not specifically study the DiDC estimator and compare it to the traditional RDD.   

Drawing from standard assumptions in cross-sectional RDD and assumptions specific to the DiDC framework, we establish conditions for reliable identification. These new assumptions, specific to the DiDC design, constrain how confounding effects behave around the threshold and over time. 

We show that the parameter of interest can be recovered using local polynomial estimation of the differences in the outcomes employing the methodology proposed by \cite{calonico2014robust}. We derive the asymptotic properties of the DiDC estimator and highlight scenarios in which its asymptotic bias can be smaller than that of the RDD. In addition, we introduce simple tests to assess the testable implications of the identifying assumptions. We explore the conditions under which DiDC mitigates bias more effectively than RDD, offering practical guidance for unbiased treatment-effect estimation. We also provide partial identification results for settings in which the identifying assumptions fail to hold.

Through Monte Carlo simulations, we evaluate the finite-sample properties of the estimator, comparing its performance to that of the local linear RDD estimator proposed by \cite{calonico2014robust} and the nonparametric DiD regression estimator proposed by \cite{santanna2020doubly}. By comparing DiDC's performance with DiD and RDD methods, we identify the scenarios where DiDC outperforms scenarios where DiDC outperforms RDDs and contexts for each method. Finally, we apply our estimator to the study by \cite{grembi2016fiscal} on the impact of fiscal rules on municipal deficits in Italy.

This work is closely connected to \cite{calonico2014robust}, whose contributions have been instrumental in robust nonparametric RDD estimation and whose estimation methods are used for our DiDC approach. \cite{frolich2019impact} discusses the potential intersection between RDD and DiD, but does not provide a formal derivation of the econometric properties, offering only high-level remarks on identification and estimation. This work also aligns with the broader literature on the intersection of RDD and panel data. \cite{pettersson2012does} explore settings where RDD is combined with fixed effects to address small sample issues and violations of the continuous support assumption. \cite{lemieux2008incentive}, use a first-difference RD approach to eliminate individual-specific fixed effects by capitalizing on the longitudinal nature of the Finnish Census data. Lastly, \cite{cellini2010value} introduce "dynamic RD" models, accommodating scenarios with multiple treatment opportunities and examining the dynamics of treatment effects.

The remainder of this article is organized as follows. Section \ref{sec:Identification} presents the DiDC as the discontinuity of differences in a potential outcomes model and shows the main identification results, along with the necessary assumptions. Estimation procedures for the treatment effect in the \textit{sharp} setting are presented in Section \ref{sec:Estimation_sharp}, along with the derivation of large-sample properties, optimal bandwidths and robust confidence intervals. In Section \ref{sec:Validity_Tests} we present tests for 2 of the identifying assumptions and in Section \ref{sec:partial} we consider the partial identification of treatment effects when the assumptions described in Section \ref{sec:Identification} are violated. Monte Carlo simulations are conducted in Section \ref{sec:Simulations} to examine the properties of the estimators. Section \ref{sec:Empirical_Illustration} provides an empirical illustration and Section \ref{sec:Conclusion} concludes. In the \hyperref[app:estimation_setup]{Appendix} we provide detailed notation, proofs and other methodological results.

\section{Identification}\label{sec:Identification}

\subsection{Framework}

Our model is best suited for panel data. We consider a setting with two time periods ($t\in\left\{0,1\right\}$) and $n$ units, indexed by $i$. For each unit, we observe an outcome $Y_{i,t}$ at each time period and a running variable $Z_{i}$, which is assumed to be time-invariant. 

We are interested in measuring the effect of a binary treatment introduced between time periods 0 and 1, denoted as an indicator function $D_{i,1}$. Treatment assignment is a function of the running variable $Z_{i}$. Units with the running variable above a cutoff $z_{0}$ are assigned to treatment, whereas units with the running variable below $z_{0}$ are assigned to control. That is, the assignment rule for treatment $D_{i,1}$ can be formalized as $D_{i,1}=\mathbf{1}\left\{Z_{i}\geq z_{0},t=1\right\}$.

Our setting differs from the canonical RD setting due to the presence of a confounding policy, denoted by $D_{i,0}$, which is introduced before the treatment of interest, but following the same assignment rule around the cutoff $z_{0}$. The assignment rule for the confounding policy is $D_{i,0}=\mathbf{1}\left\{Z_{i}\geq z_{0},t\geq 0\right\}$.

For concreteness, and in anticipation of the empirical analysis, let $Y_{it}$ denote a fiscal outcome of municipality $i$ in period $t$. The running variable $Z_{i}$ is the population of municipality $i$, the confounding policy $D_{i,0}$ is a high wage for the mayor of municipality $i$, and the treatment of interest $D_{i,1}$ is a relaxation in the fiscal rule for municipality $i$.

Potential outcomes are defined as a function of both the treatment of interest and the confounding policy. We define $Y_{i,t}(d_{0},d_{1})$ as the potential outcome for unit $i$ in period $t$ if the confounder and the treatment are set to $d_{0}$ and $d_{1}$, where $(d_{0},d_{1})\in\left\{0,1\right\}^{2}$.

The observed outcomes at each time period are related to the potential outcomes through
 \begin{align}
    Y_{i,1} &= \left[ Y_{i,1}\left( 1,1 \right) D_{i,0} + Y_{i,1} \left( 0,1 \right) \left( 1-D_{i,0} \right)  \right] D_{i,1} +\\\nonumber & \quad \quad  \left[ Y_{i,1} \left( 1,0 \right) D_{i,0} + Y_{i,1} \left( 0,0 \right) \left( 1-D_{i,0} \right) \right] \left( 1-D_{i,1} \right)
\end{align} \label{eq:potential_outcomes_Y1}

and

\begin{equation}
     Y_{i,0} = Y_{i,0}\left( 1,0 \right) D_{i,0} + Y_{i,0}\left( 0,0 \right) \left( 1 - D_{i,0} \right) \label{eq:potential_outcomes_Y0}
\end{equation}

Since the policy of interest is only introduced between periods 0 and 1, at period 0 we only observe two of the four possible potential outcomes, whereas in period 1 we can observe all four possible potential outcomes.

The presence of the confounding policy at the threshold is a violation of the continuity of mean potential outcomes assumption invoked in the canonical RD setting (Assumption 1 from \cite{hahn2001identification}). In the DiDC setting, the continuity assumption is modified to accommodate the confounding policy:

\begin{assumption}[Continuity]\label{assump:continuity}
    All potential outcomes are \textit{continuous} in $Z=z_0$. For any $\left(d_{0},d_{1}\right) \in \{0,1\}^{2}$ and $t \in \{0,1\}$:
    \begin{align*}
        \lim_{\varepsilon \rightarrow 0} \E [Y_{it}(d_{0},d_{1}) |Z_i=z_0 + \varepsilon] =  \lim_{\varepsilon \rightarrow 0} \E [Y_{it}(d_{0},d_{1}) |Z_i=z_0 - \varepsilon]
    \end{align*}
\end{assumption}

Assumption~\ref{assump:continuity} is a straightforward modification of the standard continuity assumption. It can be interpreted as stating that, other than the confounder $D_{0}$, there is no systematic difference in mean potential outcomes around the threshold $z_{0}$. In the empirical application, Assumption 1 states that, aside from the wages of mayors and the relaxation of fiscal rules, there are no other policies that change around the population threshold $z_{0}$.

The second identification assumption is the random treatment assignment around the cut-off that is known as non-manipulation at the thresold:

\begin{assumption}[Random Treatment Assignment at the Cutoff] \label{assump:random}: for $t\in \{0,1\}$ 
\begin{equation}
Y_{i,t}(d_0,d_1) \;\perp (D_{i,0}, D_{i,1}) \mid Z = z_0
\end{equation}
\end{assumption}

A fundamental assumption in the canonical RD setting is that the probability of receiving treatment is discontinuous at the threshold $z_{0}$. In the DiDC, this assumption is extended for the treatment of interest and the confounding policy:

\begin{assumption}[Sharp Discontinuites]\label{assump:disc_prob_treat} 
    Define the limits $D_{0}^+=\lim_{\epsilon \rightarrow 0} \E [D_{i0}|Z_{i}=z_0 + \epsilon)$, $D_{0}^- =\lim_{\epsilon \rightarrow 0} \E [D_{i0}|Z_{i}=z_0 - \epsilon)$, $D_{1}^+=\lim_{\epsilon \rightarrow 0} \E [D_{i1}|Z_{i}=z_0 + \epsilon)$ and $D_{1}^- =\lim_{\epsilon \rightarrow 0} \E [D_{i1}|Z_{i}=z_0 - \epsilon)$. Assume $D_{i,0}^{+}=D_{i,1}^{+}=1$ and $D_{i,0}^{-}=D_{i,1}^{-}=0$.
\end{assumption}

Assumption~\ref{assump:disc_prob_treat} states that there is perfect compliance towards the assignment rule around the threshold $z_{0}$. Borrowing the jargon from the standard RD setting, we call the case with perfect compliance the \textit{sharp} discontinuity setting. In the sharp setting, we observe $Y_{i,0}=Y_{i,0}(1,0)$ and $Y_{i,1}=Y_{i,1}(1,1)$ for units above the threshold $z_{0}$, whereas we observe $Y_{i,0}=Y_{i,0}(0,0)$ and $Y_{i,1}=Y_{i,1}(0,0)$ for units below the threshold.

Assumptions \ref{assump:continuity} and \ref{assump:disc_prob_treat} modify the traditional RDD assumptions stated at \cite{hahn2001identification} to accommodate the confounding policy. However, to identify the effect of the treatment of interest, an additional assumption regarding the evolution of the confounder effect is required:

\begin{assumption}[Time-invariance of confounding effects]\label{assump:conf_time_invar}
  The treatment of interest is the only time-variant effect at the threshold $z_0$:
    \begin{align*}
       \E \left[Y_{i,0}(1,0) - Y_{i,0}(0,0) |Z_i = z_0 \right] =\E \left[Y_{i,1}(1,0) - Y_{i,1}(0,0) |Z_i = z_0 \right]
    \end{align*}    
\end{assumption}

Assumption~\ref{assump:conf_time_invar} states that the effect of confounding policy $D_{i,0}$ remains constant over time. That is, any differences between the treatment and control groups at the threshold in $t=1$, not caused by the treatment, should have existed in $ t=0$ before the treatment was introduced.

Assumptions \ref{assump:continuity}-\ref{assump:conf_time_invar} are the assumptions invoked by \cite{grembi2016fiscal} in order to derive a causal interpretation for the DiDC estimand. In the next section, we discuss the Difference-in-Discontinuities estimand, as well as the Discontinuity-in-Differences estimand.

\subsection{The DiDC Estimand}

Before discussing identification, we define the relevant target parameters for the setting. There are two relevant causal parameters for understanding the effect of the policy of interest. We are interested in the identification of the average difference between potential outcomes of treated units at the cutoff and the potential outcomes of untreated units, holding the confounder fixed. 

Thus, the first parameter we define is $\tau_{c}\equiv\E \left[Y_{i,1}(1,1)-Y_{i,1}(1,0)|Z_{i}=z_{0}\right]$, which is the average treatment effect for individuals at the cutoff exposed to the confounder in period 1. The second parameter of interest is the average treatment effect for individuals at the cutoff not exposed to the confounder in period 1, defined as $\tau_{uc}\equiv\E \left[Y_{i,1}(0,1)-Y_{i,1}(0,0)|Z_{i}=z_{0}\right]$. We now discuss the point identification of these parameters.

We define $\tau_{t}^{RD}$ as the regression discontinuity estimand at period $t$. Let $Y_{t}^{+}=\lim_{\varepsilon \rightarrow 0} E[Y_{i,t} |Z_{i}=z_0 + \epsilon]$ and $Y_{t}^{-}=\lim_{\varepsilon \rightarrow 0} E[Y_{i,t} |Z_{i}=z_0 - \epsilon]$. Thus, the RD estimand at period $t$ is simply $\tau_{t}^{RD}=Y_{t}^{+}-Y_{t}^{-}$. \cite{grembi2016fiscal} define the DiDC estimand as the difference between the RD estimand in period 1 and the RD estimand in period 0: $\tau^{DiDC}=\tau_{1}^{RD}-\tau_{0}^{RD}$.

The next lemma shows that the DiDC estimand can be alternatively defined as an RD estimand evaluated at the difference between outcomes over time:

\begin{lemma}\label{claim:equal_estimands}
    Under Assumptions \ref{assump:continuity}, \ref{assump:random} and \ref{assump:disc_prob_treat}, the difference of RDs, and the RD of the differences are equivalent:
    \begin{align*}
        \tau^{DiDC} = \tau_1^{RD} - \tau_0^{RD} = \lim_{\varepsilon \rightarrow 0} \E [\Delta Y_{i} |Z_{i}=z_0 + \epsilon] - \lim_{\varepsilon \rightarrow 0} \E [\Delta Y_{i} |Z_{i}=z_0 - \epsilon] 
    \end{align*}
\end{lemma}
\begin{proof}
    See Appendix \ref{app:proofs}.
\end{proof}

The result in Lemma~\ref{claim:equal_estimands} is intuitive and presents an attractive feature for the DiDC setting: the estimand can be implemented via a single RD estimand rather than taking the difference across estimands. We now turn to the causal interpretation of the estimand, which is already well established in the literature:

\begin{lemma}[\cite{grembi2016fiscal}]Under Assumptions~\ref{assump:continuity}-\ref{assump:conf_time_invar}, we have

\begin{equation*}
    \tau^{DiDC}=\tau_{c}
\end{equation*}
    
\end{lemma}

The difference-in-discontinuities estimand identifies the effect of the treatment of interest in period 1 for units at the threshold that are exposed to the confounding policy. To identify a more general causal effect, an additional assumption is required:

\begin{assumption}[No-interaction between treatment and confounding effects]\label{assump:indep_from_conf}
    \begin{align*}
       \tau_{c}=\tau_{uc}= \E \left[Y_{i,1}(d_{0},1)-Y_{i,1}(d_{0},0)|Z_{i}=z_{0}\right]
    \end{align*}
\end{assumption}

Assumption~\ref{assump:indep_from_conf} states that the effect of the treatment of interest does not depend on the confounding policy. Such an assumption can be justified by a potential outcomes model in which the treatment effect of interest and the effect of the confounder are linearly separable. Assumption~\ref{assump:indep_from_conf} might be overly restrictive, as one might expect the treatment of interest and the confounding policy to interact; nevertheless, it allows the DiDC estimand to identify a more general causal effect:

\begin{corollary}
    Under Assumptions \ref{assump:continuity}-\ref{assump:indep_from_conf},

    \begin{equation*}
        \tau^{DiDC}= \E \left[ Y_{i,1}(d_{0},1)-Y_{i,1}(d_{0},0) |Z_i = z_0 \right]
    \end{equation*}
\end{corollary}

\subsubsection{Relaxing the No-Interaction Assumption}
Assumption \ref{assump:indep_from_conf} can be overly restrictive, and sometimes there could be scenarios where the confounding policy and the treatment of interest might interact. Rather than assuming additive effects at time $t=1$, we can relax this assumption with one that allows the inclusion of multiplicative effects:
\begin{customassump}{4'}[Multiplicative Effects]\label{assump:mult}
    \begin{align*}
        & \E\left[ Y_{i,1}(1,1) - Y_{i,1}(0,0) |Z_i = z_0 \right]  \\ & \quad =\E \left[ Y_{i,1}(0,1)-Y_{i,1}(0,0) |Z_i = z_0 \right] \E \left[ Y_{i,1}(1,0)-Y_{i,1}(0,0) |Z_i = z_0 \right] 
    \end{align*}
\end{customassump}
This assumption allows for the possibility that the combined effect of the confounding policy and the treatment of interest is not simply the sum of their individual effects, but also includes an additional component reflecting their interaction. This additional component is determined by multiplying the effects together. By incorporating this concept, we can derive a new estimand that uses both the RDD at time 0 and the DiDC to identify the causal effect of the treatment of interest.

\begin{lemma}\label{lemma:mult}
    Under Assumptions \ref{assump:continuity}, \ref{assump:random}
of the policy of interest at the cutoff can be identified as\ref{assump:disc_prob_treat}, \ref{assump:conf_time_invar} and \ref{assump:mult}, the effect of the policy of interest at the cutoff can be identified as
    \[
    \frac{\tau^{DiDC}}{\tau^{RD}_0} +1 \;=\; \E\!\left[ Y_{i,1}(0,1) - Y_{i,1}(0,0) \;\middle|\; Z_i = z_0 \right].
    \]
\end{lemma}

\begin{proof}
    See Appendix \ref{app:proofs}.
\end{proof}

This estimand can capture more complex relationships between the treatment and the confounding factors. Our future work will focus on the properties and robustness of this estimator.

\section{Estimation and Inference for the \textit{Sharp} DiDC}\label{sec:Estimation_sharp}

Following standard practice in the regression discontinuity literature, we propose a local polynomial estimation to recover the parameter of interest. This nonparametric method involves fitting a polynomial to the data near the threshold and using the estimated function to calculate differences in outcomes between the treatment and control groups at that threshold. 

To implement the estimation at $z_0$, we rely on the methodologies proposed by \cite{calonico2014robust} for local polynomial estimation of the RDD. In our case, we estimate a local polynomial regression of the differences in outcomes over time ($\Delta Y_{i}$).




Under a mild continuity condition, \cite{hahn2001identification} showed that the average treatment effect at the threshold is nonparametrically identifiable as the difference of two conditional expectations evaluated at the (induced) boundary point $z_0=0$. Similarly, the \textit{sharp} DiDC parameter can be identified as the difference of the difference (in time) of two conditional expectations evaluated at  $z_0=0$ at each side of $z_0$:
\begin{align*}
    \tau^{DiDC} &= \Delta\mu_+ - \Delta\mu_-,\\
    \Delta\mu_+ &= \Lim{z\rightarrow 0^+} \Delta \mu(z), \quad \Delta \mu_{-} = \Lim{z\rightarrow 0^-} \Delta \mu(z) \\
    \Delta \mu(z) &=\E \left[ \Delta Y_{i} | Z_{i} =z \right]
\end{align*}

Appendix \ref{app:estimation_setup} states the assumptions underlying the nonparametric local polynomial regression estimation. These assumptions impose restrictions on the kernel function, require the existence of certain moments, ensure continuity of the running variable in the relevant region, impose smoothness conditions on the regression functions, and bound the conditional variance of the observed outcome.

\subsection{Local Polynomial Estimator}
For a given $\nu \leq p \in \mathbf{N}$, define $\Delta \mu^{(\nu)}$ as the $\nu$th-order derivatives of the $p$th-order local polynomial of the difference. We are interested in the limits of this function around the threshold, $\Delta \mu^{(\nu)}_{+}$ and $\Delta \mu^{(\nu)}_{-}$. The general estimand of interest is $\tau^{DiDC} = \Delta \mu_{+} - \Delta \mu_{-}$, with $\Delta \mu=\Delta \mu^{0}$. 
 The $p$th-order local polynomial estimators of the $\nu$th-order derivatives $\Delta \mu^{(\nu)}_{+,p}$ and $\Delta \mu^{(\nu)}_{-,p}$ are:
\begin{align*}
    \Delta \hat{\mu}^{(\nu)}_{+,p}(h_n) &= \nu!e^{\prime}_\nu \hat{\delta}_{\Delta Y+,p}(h_n)\\
    \Delta \hat{\mu}^{(\nu)}_{-,p}(h_n) &= \nu!e^{\prime}_\nu \hat{\delta}_{Delta Y-,p}(h_n)\\
    \hat{\delta}_{\Delta Y+,p}(h_n) &= arg\min_{\delta \in \mathbf{R}^{p+1}} \sum_{i=1}^{n} \mathbbm{1}(Z_i \geq 0)(\Delta Y_{i}-r_p(Z_i)^{\prime} \delta)^2 K_{h_n}(Z_i)\\
    \hat{\delta}_{\Delta Y-,p}(h_n) &= arg\min_{\delta \in \mathbf{R}^{p+1}} \sum_{i=1}^{n} \mathbbm{1}(Z_i < 0)(\Delta Y_{i}-r_p(Z_i)^{\prime} \delta)^2 K_{h_n}(Z_i)
\end{align*}
where $e_\nu$ is a conformable $(\nu +1)$ unit vector, $K_h(u) = K(u/h)/h$, $h_n$ is a positive bandwidth sequence, 
    $r_p(x) = \begin{bmatrix} 
                    1 & x & \hdots & x^p
                \end{bmatrix}'$ and
    $ \Delta Y = \begin{bmatrix}
                    \Delta Y_{1} & \Delta Y_{2} & \hdots & \Delta Y_{n}
                \end{bmatrix}'$. Therefore, for a positive bandwidth $h_n$, the nonparametric estimator of $\tau_{\nu,p}$ is
\begin{align} \label{eq:estimator}
    \hat{\tau}^{DiDC}_{p}(h_n) = \Delta \hat{\mu}_{+,p}(h_n) - \Delta \hat{\mu}_{-,p}(h_n)
\end{align}

\subsubsection{Bandwidth choice}
To perform local polynomial estimation, it is necessary to choose an appropriate bandwidth $h_n$. This parameter determines the range of observations used for the estimation and impacts the trade-off between bias and variance in the estimated treatment effect $(\tau^{DiDC})$. In point estimation, the standard approach is to select the bandwidth that minimizes the asymptotic Mean Squared Error (MSE) of the estimator. Let $\chi_n = \begin{bmatrix}
                Z_1 & \hdots & Z_n
            \end{bmatrix}'$, the MSE is:
\begin{align*}
    MSE_{\nu,p,s}(h_n) =\E\left[ \left( \hat{\tau}_{\nu, p, s}(h_n) - \tau_{\nu,p} \right)^2 | \chi_n \right]
\end{align*}

\begin{lemma}\label{lemma:MSEoptimal_band}
    Under Assumptions \ref{assump:continuity} and \ref{assump:disc_prob_treat} with $S \geq p+1$, $\nu \leq p$, $h_n \rightarrow 0$ and $nh_n \rightarrow \infty$, the asymptotic MSE-optimal bandwidth is given by:
    \begin{align*}
        h_{n,\nu, p}^{MSE} = \left( \frac{\left( 1+ 2 \nu \right) \mathbf{V}_{\nu ,p}}{2n \left( 1+p-\nu \right) \textbf{B}^{2}_{\nu, p, p+1, s}} \right)^{\frac{1}{2p+3}}
    \end{align*}
    where $\mathbb{V_{\nu,p}} =\nu!\frac{\sigma_+^2 - \sigma_-^2}{f}e^{\prime}_{\nu}\Gamma_p^{-1} \Psi_p \Gamma_p^{-1}  $ and $\mathbb{B}_{\nu, p, p+1, s} = \frac{\Delta \mu^{(p+1)}_{+}-(-1)^{\nu+p+s}\Delta \mu^{(p+1)}_{-}}{(p+1)!} \nu! e^{\prime}_{\nu}\Gamma_p^{-1} \varphi_{p,p+1}$, provided that \(\textbf{B}_{\nu, p, p+1, s} \neq 0\).
\end{lemma}
\begin{proof}    
    in Appendix \ref{app:proofs}.
\end{proof}

\subsection{Inference}\label{subsec:Inference}
We now discuss the asymptotic properties of the estimator. Using Lemma \ref{lemma:a.1} found in the Appendix \ref{app:Lemmas}, it is possible to recover the leading asymptotic bias, expressed as:
\begin{align} \label{eq:bias}
   \E\left[ \hat{\tau}^{DiDC}_{\nu,p}(h_n) | \chi_n \right] - \tau_\nu &=   h_n ^{p+1-\nu} \mathbf{B}_{\nu,p,p+1}(h_n) +  h_n ^{p+2-\nu} \mathbf{B}_{\nu,p,p+2}(h_n) \\
    \nonumber \\
    \mathbf{B}_{\nu,p,r} (h_n)  &= \frac{\Delta \mu_{+}^{(r)}\mathcal{B}_{+,\nu,p,r}(h_n) -  \Delta \mu_{-}^{(r)}\mathcal{B}_{-,\nu,p,r}(h_n) }{ r!} \nonumber
\end{align} 
where $\mathcal{B}_{+,\nu,p,r} = \nu!e_\nu'\Gamma_{+,p}^{-1}(h_n)\vartheta_{+,p,r}(h_n)$ and $\mathcal{B}_{-,\nu,p,r} = \nu!e_\nu'\Gamma_{-,p}^{-1}(h_n)\vartheta_{-,p,r}(h_n)$ are asymptotically bounded. Further notation is available in Appendix \ref{app:estimation_setup} and a detailed proof for the statement above can be found in Appendix \ref{app:Lemmas} under Lemma \ref{lemma:a.1}. 

This is where we believe one of the main contributions of this research lies: the asymptotic bias of the DiDC can be zero if we include an assumption similar to that of parallel trends, or if the shapes of the data-generating processes for both groups are time-invariant. In the first case, imposition of ``parallel trends`` restricts the functional form in such a way that $\Delta \mu_{+}^{(r)}\mathcal{B}_{+,\nu,p,r}(h_n) =  \Delta \mu_{-}^{(r)}\mathcal{B}_{-,\nu,p,r}(h_n)$. It's important to note that the symmetry of the kernel function, as imposed in Assumption \ref{assump:kernel}, plays a significant role in this result. In the case of time-invariant data-generating processes, both $\Delta \mu_+^{(\nu)}$ and $\Delta\mu_-^{(\nu)}$ equate to zero. 

We present Claim \ref{claim:equal_biases}, demonstrating the bias of difference-in-discontinuities in relation to RDDs:

\begin{lemma}\label{claim:equal_biases} The bias of $\hat{\tau}_{\nu,p}^{DiDC}$ can be decomposed as
    \begin{align} \label{eq:bias_didc_rdds}
        B \left[ \hat{\tau}^{DiDC}_{\nu,p} (h_n) \right] &= B \left[ \hat{\tau}^{RD}_{1,\nu,p} (h_n) \right] - B \left[ \hat{\tau}^{RD}_{0,\nu,p} (h_n) \right]     
    \end{align}
    where $ B \left[ \hat{\tau}^{RD}_{1,\nu,p} (h_n) \right]$ is the bias of the RD estimated at time $t=1$, after the intervention happened, and $ B \left[ \hat{\tau}^{RD}_{0,\nu,p} (h_n) \right]$ is the bias of the RD estimated at time $t=0$, before the intervention happened. 
\end{lemma}
\begin{proof}
    in Appendix \ref{app:proofs}.
\end{proof}

Equation \ref{eq:bias_didc_rdds} demonstrates how incorporating additional data can help reduce the bias in RD analysis. Traditional RD analysis of interventions typically uses only a cross-section of post-intervention data.  However, by incorporating pre-intervention data into the analysis, bias can be substantially reduced and, under specific conditions, eliminated. Even when the bias is not completely eliminated, it can still be reduced if the RD estimation of pre-treatment data exhibits bias in the same direction as the post-treatment RD and if its magnitude is not too large to overcome the original bias.

\subsubsection{Bias Correction} \label{sec:Estimation_inference_biascorrection}
Using the MSE-optimal bandwidth for point inference can result in bandwidths that are ``too large``, which may seem attractive for reducing variance but can introduce first-order asymptotic bias\footnote{\cite{calonico2014robust}}. To address this issue, we employ robust bias-corrected confidence intervals (CIs) proposed by \cite{calonico2014robust}. These intervals take into account the asymptotic bias of the point estimate by 1) estimating the bias and recentering the CI, and 2) incorporating the additional variance from estimating the bias for bias correction into the CI. This process requires estimating a separate local polynomial of order $q$, with $q > p \geq \nu$.  The bias-corrected estimator is defined as:
\begin{align*}
    \hat{\tau}^{bc}_{\nu, p, q} \left(h_n, b_n \right) =& \tau_{\nu, p} \left(h_n\right) - h_n^{p+1 - \nu} \hat{\mathbf{B}}_{\nu,p,q}  \left(h_n, b_n \right), \\
    \hat{\mathbf{B}}_{\nu,p,q,s}  \left(h_n, b_n \right) =& \frac{\Delta \hat{\mu}^{(p+1)}_{+,q}(b_n)}{(p+1)!}\mathcal{B}_{+,\nu,p,p+1} - \frac{(-1)^{\nu+p+s}\Delta \hat{\mu}^{(p+1)}_{-,q}(b_n)}{(p+1)!}\mathcal{B}_{-,\nu,p,p+1}
\end{align*}
with $\Delta \hat{\mu}^{(p+1)}_{+,q}(b_n)$ and $\Delta \hat{\mu}^{(p+1)}_{-,q}(b_n)$ being local polynomial estimations as described in Appendix \ref{app:estimation_setup}. $\hat{\mathbf{B}}_{\nu,p,q}  \left(h_n, b_n \right)$ is the estimationg we get of the bias from the $q$th-order local polynomial.

To determine the MSE-optimal bandwidth for estimating the bias, we need to conduct a separate local polynomial estimation of order $q$ with $q > p \geq \nu$. Once again, we aim to minimize the MSE. Following Lemma \ref{lemma:MSEoptimal_band} and applying it to the bias estimate, we can find the MSE-optimal bandwidth for the bias estimation.
\begin{align*}
    b_n = h_{n,p+1, q,\nu+p+1}^{MSE} = \left( \frac{\left( 3+2p \right) \mathbf{V}_{p+1, q}}{2n \left( q-p \right) \textbf{B}^{2}_{p+1, q, q+1, \nu+p+1}} \right)^{\frac{1}{2q+3}}
\end{align*}

\subsubsection{Asymptotic Properties and Robust Confidence Interval}
Following \cite{calonico2014robust}, we derive a large-sample distributional approximation that accounts for the added variability introduced by the bias estimate. The large-sample approximation for the standardized t-statistic is formalized in the theorem below:

\begin{theorem}
    Under Assumptions \ref{assump:continuity}-\ref{assump:indep_from_conf},  if $S \geq q+1$, $n \min \{ h_n^{2p+3}, b_n^{2p+3} \} \times \max \{ h_n^2 , b_n^{2(q-p)} \}  \rightarrow \infty$, then 
\begin{align*}
    T^{rbc}_{\nu,p,q} \left(h_n, b_n \right) = \frac{\hat{\tau}^{bc}_{\nu, p, q} \left(h_n, b_n \right) - \tau_{\nu}}{\sqrt{\mathbf{V}^{bc}_{\nu, p, q}\left(h_n, b_n \right)}} \xrightarrow{d} \mathcal{N} \left(0,1\right)
\end{align*}
where $\mathbf{V}^{bc}_{\nu, p, q}\left(h_n, b_n \right)$ is described in appendix C.
\end{theorem}

This motivates the following confidence interval:
\begin{align*}
    CI_{\nu, p, q}\left(h_n, b_n \right) = \left[ \hat{\tau}^{bc}_{\nu, p, q} \left(h_n, b_n \right) \pm \Phi^{-1}_{1-\alpha / 2} \sqrt{\mathbf{V}^{bc}_{\nu, p, q}\left(h_n, b_n \right)} \right]
\end{align*}

Just like in the standard RD setting, this confidence interval have better properties compared to conventional bias-corrected intervals. They are more robust to bandwidth selection, feature coverage error decays at a faster rate, and offer shorter interval lengths, as explained by \cite{calonico2014robust}.

\section{Validity Tests}\label{sec:Validity_Tests}

\subsection{Testing if Confounding Effect is Time-Invariant} \label{subsec:Test_conf_time_invar}

As highlighted in Section \ref{sec:Identification}, the importance of Assumption \ref{assump:conf_time_invar} - that the confounding effect is constant over time - cannot be overstated. Without this assumption, all estimations are inherently flawed, leading to biased treatment effect estimates. Therefore, an initial step when considering the applicability of a differences-in-discontinuity approach is to assess the validity of this assumption. 

To accomplish this, we propose a simple test to check violations of this assumption. In essence, it involves using available pre-treatment periods to estimate stacked RDs. The goal is to examine whether the RD coefficients remain consistent across multiple periods. This involves a set of $k$ periods where it is known that no changes occurred at the threshold. 

It is crucial to note that this estimation method is not suitable if any alterations occurred at the threshold between the initial period in this sample and the period just preceding the introduction of the treatment of interest. To ensure the validity of this approach, it is necessary to select a period during which no events occurred at this threshold that could influence the observed outcome. 

The procedure involves considering the following stacked RDs regression model:
\begin{align}
    Y_i = \alpha_i + \sum_{k=1}^K \left[ \beta_{-k} \left(Z_i - z_0 \right) + \theta_{-k}D_{i} \left(Z_i - z_0 \right) \right]T_{-k} + \varepsilon_i \label{eq:reg_test_conf_const}
\end{align}
where $T_{-k}$ is a dummy indicating the period to which the data belong, that is, the RD of which period. The hypothesis to be tested is:
\begin{align*}
    H_0: \quad \theta_0 = \theta_{-1} = ... = \theta_{-K}
\end{align*}
A Wald test is sufficient to test these hypotheses. Rejection of the null hypothesis indicates that the differences-in-discontinuity design may not be suitable for estimating the effect of treatment in this setting.

A question arises about concerning bandwidth to use when conducting this test.  There are several options available, and Appendix \ref{app:Bandwidth_Test} provides details on the simulations used to assess the better bandwidth. In summary, in practice, any bandwidth derived from the data that is optimal for the specific RD should work well.

\subsection{Testing the Time-Invariance of Potential Outcomes Functional Form}\label{subsec:test_functiona_form}

In Section \ref{subsec:Inference}, we briefly discussed an important aspect of the DiDC method: its potential to achieve zero asymptotic bias when choosing DiDC over the standard RD, provided that the functional forms of the data-generating processes for both groups are time-invariant. Essentially, this would mean that the derivatives of the functions on each side of the threshold at each $t=\{0,1\}$ ($\mu_{+,1}(z)$ and $\mu_{+,0}(z)$, $\mu_{-,1}(z)$ and $\mu_{-,0}(z)$) are equal at each point of the running variable $z$, resulting in zero bias in the local polynomial estimation.

When $\Delta \mu_+^{(\nu)} = \Delta \mu_-^{(\nu)} = 0$, the bias term $\mathbf{B}_{\nu,p,r} (h_n)  = \frac{\Delta \mu_{+}^{(r)}\mathcal{B}_{+,\nu,p,r}(h_n) -  \Delta \mu_{-}^{(r)}\mathcal{B}_{-,\nu,p,r}(h_n) }{ r!}$ is also zero. This condition implies that $\Delta \mu(z)$ is a constant, allowing unbiased point estimates through linear estimations on both sides of the threshold. As a result, two simple OLS regressions would be sufficient for accurately estimating the treatment effect.

Therefore, it can be very beneficial for researchers to know if the shapes of the data-generating functions remain stable over time. This knowledge would allow them to determine whether they can use simpler, less biased estimation methods. We propose a simple nonparametric Two-Sample Kolmogorov-Smirnov (KS) test for time-invariance of the functional forms of conditional means. The procedure is an adaptation of the KS test for conditional moment restrictions from Whang (2001).

In order to implement the test, it is required to have data from more than one time period before the implementation of the treatment of interest ($t\in\left\{-1,0,1\right\}$). Using the notation from Appendix A, we write the nonparametric regression model for the outcome in time $t$ as

\begin{equation*}
    Y_{i,t}=\mu_{t}(Z_{i})+\varepsilon_{i,t}
\end{equation*}

where $\mu_{t}(Z_{i})=\mathbb{E}\left[Y_{i,t}|Z_{i}\right]$ and $\varepsilon_{i,t}=Y_{i,t}-\mathbb{E}\left[Y_{i,t}|Z_{i}\right]$. Let $\mathcal{S}^{+}(Z)$ and $\mathcal{S}^{-}(X)$ denote, respectively, the support of the running variable above and below the threshold $z_{0}$. We want to test whether $\mathbb{E}\left[\mu_{0}(Z_{i})-\mu_{-1}(Z_{i})|Z_{i}\right]=0$ almost surely for $z\in\mathcal{S}^{+}(Z)$ and whether  $\mathbb{E}\left[\mu_{0}(Z_{i})-\mu_{-1}(Z_{i})|Z_{i}\right]=0$ almost surely for $z\in\mathcal{S}^{-}(Z)$. To do so, we assume the conditional mean of $Y_{i,t}$ can be approximated by a $K\times1$ vector of approximating functions $p^{K}(.)=(p_{1}(.),...,p_{K}(.))^{'}$. Using this vector, we write the regression model for the outcome in period $t$ above and below the threshold, respectively, as

\begin{align*}
    &Y_{i,t}=p^{K}(Z_{i})^{'}\gamma^{+}_{t}+u_{i,t}^{+},\\
    &Y_{i,t}=p^{K}(Z_{i})^{'}\gamma^{-}_{t}+u_{i,t}^{-}
\end{align*}

Under the null hypothesis of time-invariant conditional means, we have

\begin{align*}
    &\mathbb{E}\left [ \left ( p^{K}(Z_{i})^{'}\gamma^{+}_{0}-p^{K}(Z_{i})^{'}\gamma^{+}_{-1} \right )\mathbf{1}\left\{ Z_{i}\leq z\right\} \right ]=0, \forall z\in\mathcal{S}^{+}(Z),\\&\mathbb{E}\left [ \left ( p^{K}(Z_{i})^{'}\gamma^{-}_{0}-p^{K}(Z_{i})^{'}\gamma^{-}_{-1} \right )\mathbf{1}\left\{ Z_{i}\leq z\right\} \right ]=0, \forall z\in\mathcal{S}^{-}(Z)
\end{align*}

which motivates the following KS type test statistics:

\begin{align*}
    &KS_{+}=\sup_{z\in\mathcal{S}^{+}(z)}\frac{1}{\sqrt{n}^{+}}\sum_{i:Z_{i}\geq z_{0}}\left |\left ( p^{K}(Z_{i})^{'}\widehat{\gamma}^{+}_{0}-p^{K}(Z_{i})^{'}\widehat{\gamma}^{+}_{-1} \right )\mathbf{1}\left\{Z_{i}\leq z \right\} \right |\\& KS_{-}=\sup_{z\in\mathcal{S}^{-}(z)}\frac{1}{\sqrt{n}^{-}}\sum_{i:Z_{i}\geq z_{0}}\left |\left ( p^{K}(Z_{i})^{'}\widehat{\gamma}^{-}_{0}-p^{K}(Z_{i})^{'}\widehat{\gamma}^{-}_{-1} \right )\mathbf{1}\left\{Z_{i}\leq z \right\} \right |
\end{align*}

where $n^{+}=\sum_{i=1}^{n}\mathbf{1}\left\{Z_{i}\geq z_{0}\right\}$, $n^{-}=\sum_{i=1}^{n}\mathbf{1}\left\{Z_{i}< z_{0}\right\}$, $\widehat{\gamma}^{+}_{t}$ is the vector of least square estimates of the regression of $Y_{i,t}$ on $p^{K}(Z_{i})^{'}$ for units above the threshold and $\widehat{\gamma}^{-}_{t}$ is the vector of least square estimates of the regression of $Y_{i,t}$ on $p^{K}(Z_{i})^{'}$ for units below the threshold.

Confidence intervals and p-values for the test statistics can be obtained using a recentered bootstrap \citep{hall1996bootstrap} in which the null hypothesis is imposed for the bootstrapped distribution. In Appendix F, we provide the conditions under which the bootstrap is consistent, and the test statistics are powerful against local alternatives.

\section{Partial Identification and Sensitivity Analysis}\label{sec:partial}

\subsection{Partial Identification under Bounded Variation Assumptions}

In this section, we consider the partial identification of treatment effects when Assumption~\ref{assump:conf_time_invar} is violated. In many settings, assuming that confounding effects remain constant over time is implausible. If that is the case, then the effect of the treatment of interest is not point identified. 

We replace the time-invariance assumption with bounded variations assumptions, in the spirit of \cite{manski2018how}, and derive identified sets for the causal effect of interest as a function of the sensitivity parameters that bound variations in potential outcomes across time. Formally, we invoke the following assumption:

\begin{assumption}[Bounded Variations]\label{assump:bounded_variation} Let $c_{1}$ and $c_{2}$ be scalars. We assume that
\begin{align*}
    &\left |\mathbb{E}\left[Y_{i,1}(1,0)-Y_{i,0}(1,0)|Z_{i}=z_{}\right] \right |\leq c_{1},\\&\left |\mathbb{E}\left [ Y_{i,1}(1,0)-Y_{i,1}(0,0)|Z_{i}=z_{0} \right ]-\mathbb{E}\left [ Y_{i,0}(1,0)-Y_{i,0}(0,0)|Z_{i}=z_{0} \right ] \right |\leq c_{2}
\end{align*}    
\end{assumption}

The first inequality in Assumption~\ref{assump:bounded_variation} states that the difference between the mean potential outcome $Y_{i,0}(0,1)$ (observed) and the mean potential outcome $Y_{i,1}(1,0)$ (not observed) at the threshold is no greater than the scalar $c_{1}$. Thus, the parameter $c_{1}$ can be interpreted as a time-trend in the evolution of the potential outcome associated to receiving the confounder, but not the treatment of interest.

The second inequality states that the difference between the mean confounding effect in period 0 (point identified under Assumptions~\ref{assump:continuity} and \ref{assump:conf_time_invar}), and the mean confounding effect in period 1 (not observed nor identified) at the threshold is no greater than the scalar $c_{2}$. In that sense, $c_{2}$ can be interpreted as a time-trend in confounding effects. In the next lemma, we derive the identified set for $\tau_{c}$:

\begin{lemma}\label{lemma:bounds}
    Under Assumptions 1,2 and 5, $\tau_{c}\in\left[\tau_{c}^{LB},\tau_{c}^{UB}\right]$, where
    \begin{align*}
        &\tau_{c}^{LB}=\max\left\{ \Delta Y^{+}-c_{ 1}, (\Delta Y^{+}-\Delta Y^{-})-c_{2}\right\}\\&\tau_{c}^{UB}=\min\left\{ \Delta Y^{+}+c_{1}, (\Delta Y^{+}-\Delta Y^{-})+c_{2}\right\}
    \end{align*}
\end{lemma}

Lemma~\ref{lemma:bounds} shows that average treatment effect for confounded individuals at the cutoff is partially identified as a function of the sensitivity parameters $c_{1}$ and $c_{2}$. It is straightforward to identify breakdown values for this identified set. That is, the largest values of $c_{1}$ and $c_{2}$ for which a particular conclusion holds. For instance, researchers might be interested in the largest values of bounded variation under which one can conclude that the causal effect of interest is positive. If the lower bound is above 0 for large values of $c_{1}$ and $c_{2}$, then one can conclude that the qualitative conclusions of the empirical exercise are robust to violations of the identifying assumptions.

Estimation of the bounds for $\tau$ is straightforward. We a nonparametric estimator in which the local polynomial estimates described in Section 3 are plugged in the max and min operators. The estimators for the lower and the upper bound are, respectively,

\begin{align*}
    &\widehat{\tau}_{c}^{LB}=\max\left\{\Delta \hat{\mu}^{(\nu)}_{+,p}(h_n,b_{n})-c_{1},(\Delta \hat{\mu}^{(\nu)}_{+,p}(h_n,b_{n})-\Delta \hat{\mu}^{(\nu)}_{-,p}(h_n,b_{n}))-c_{2}\right\}\\&\widehat{\tau}_{c}^{UB}=\min\left\{\Delta \hat{\mu}^{(\nu)}_{+,p}(h_n,b_{n})+c_{1},(\Delta \hat{\mu}^{(\nu)}_{+,p}(h_n,b_{n})-\Delta \hat{\mu}^{(\nu)}_{-,p}(h_n,b_{n}))+c_{2}\right\}
\end{align*}

We use the Delta method for \textit{Hadamard directionally differentiable mappings} \cite{fang2018inference} to derive the asymptotic properties of the plug-in estimators for the bounds. The next theorem shows that the estimators of the bounds converge to a non-Gaussian limiting process at the $\sqrt{n\min\left\{h_{n},b_{n}\right\}}$ rate:

\begin{theorem}
    Suppose Assumptions \ref{assump:continuity}, \ref{assump:disc_prob_treat} and \ref{assump:bounded_variation} hold. Furthermore, assume that $c_{1},c_{2}\in\mathcal{C}$ for some finite grid $\mathcal{C}$. Then,

    \begin{equation*}
        \sqrt{n\min\left\{h_{n},b_{n}\right\}}\begin{pmatrix}
\widehat{\tau}_{c}^{LB}-\tau_{c}^{LB} \\
\widehat{\tau}_{c}^{UB}-\tau_{c}^{UB}
\end{pmatrix}\rightarrow \textbf{Z}(y,z,c_{1},c_{2})
    \end{equation*}

    a tight random element of $l^{\infty}\left(\mathcal{S}(Y)\times\mathcal{S}(Z)\times\mathcal{C},\mathbb{R}^{2}\right)$.
\end{theorem}

Inference for the bounds estimates is particularly challenging. Using the limiting process to which the estimates converge to obtain analytical asymptotic confidence bands is difficult. An alternative would be a bootstrap. Although the Delta Method is valid for Hadamard directionally differentiable functions, the standard nonparametric bootstrap is not. Instead, we recommend researchers to use the bootstrap procedure proposed in Section\ref{sec:Estimation_inference_biascorrection} in \cite{fang2018inference}\footnote{Note, however, that the mapping is fully differentiable if the quantities within the min and max operators are different, in which case the standard bootstrap is valid}.

\subsection{Partial Identification under Modularity Assumptions}

An alternative approach to partial identification in the Diff-in-Disc setting is to partially identify treatment effects by exploiting assumptions regarding the interaction between the treatment of interest and the confounding policy. The first step is to assume that potential outcomes are bounded up to known values:

\begin{assumption}[Bounded Outcomes]\label{assump:bounded_outcomes}

For all $i\in\left\{1,...,n\right\}$, $t\in\left\{0,1\right\}$ and $(d_{0},d_{1})\in\left\{0,1\right\}^{2}$,

\begin{equation*}
    -\infty<y^{min}\leq Y_{i,t}(d_{0},d_{1})\leq y^{max}<\infty
\end{equation*}
   
\end{assumption}

In some settings the values $y^{min}$ and $y^{max}$ readily justified by the nature of the outcome, and in others these values are sensitivity parameters that reflect the beliefs about the smallest and largest possible values of potential outcomes \citep{kline2025finite}.

Assumption \ref{assump:bounded_outcomes} motivates the worst-case bounds \citep{manski1989anatomy}. Without further assumptions, it implies that the treatment effects lie in the interval $\left[y^{min}-y^{max},y^{max}-y^{min}\right]$. Such bounds are usually uninformative as the bound covers zero, and hence the sign of the treatment effect is not identified. 

Theoretical knowledge of the applied setting can be used to tighten the bounds. For instance, researchers can draw insights from economic theory to assume that the treatment of interest and the confounding policy are complementary:

\begin{assumption}[Complementarity]\label{assump:compl}

For all $i\in\left\{1,...,n\right\}$ and $t\in\left\{0,1\right\}$,

\begin{equation*}
    Y_{i,1}(1,0)+Y_{i,1}(0,1)\leq Y_{i,1}(1,1)+Y_{i,1}(0,0)
\end{equation*}
  
\end{assumption}

Assumption \ref{assump:compl} is the \textit{supermodularity} assumption from \cite{twinam2017complementarity} for the case of two binary treatments. It formalizes the notion that the magnitude of the causal effect of interest increases with the confounding policy. Mapping the assumption to our empirical setting, one can assume, for instance, that the effect of relaxing fiscal rules on public finance outcomes is greater for municipalities with less-skilled incumbents. The proposition below derives the sharp identified sets for $\tau_{c}$ and $\tau_{uc}$ under complementarity:

\begin{lemma}
    Suppose Assumptions \ref{assump:continuity}, \ref{assump:disc_prob_treat}, \ref{assump:bounded_outcomes} and \ref{assump:compl} hold. Then,

    \begin{equation*}
        \tau_{c}\in\left[y^{min}-Y_{1}^{-},y^{max}-y^{min}\right]
    \end{equation*}

    and

    \begin{equation*}
        \tau_{uc}\in\left[y^{min}-y^{max},Y_{1}^{+}-y^{min}\right]
    \end{equation*}
where $Y_{1}^{+} = \Lim{z\rightarrow 0^+} Y_{1}, \quad Y_{1}^{-} = \Lim{z\rightarrow 0^-} Y_{1}$
\end{lemma}

Alternatively, researchers can assume that the treatment of interest and the confounder are substitutes:

\begin{assumption}[Substitutability]\label{assump:subst}

For all $i\in\left\{1,...,n\right\}$ and $t\in\left\{0,1\right\}$,

\begin{equation*}
    Y_{i,t}(1,0)+Y_{i,t}(0,1)\geq Y_{i,t}(1,1)+Y_{i,t}(0,0)
\end{equation*}
    
\end{assumption}

Assumption \ref{assump:subst} is the \textit{submodularity} assumption from \cite{twinam2017complementarity} for the case of two binary treatments. It states that the treatment effect of interest decreases in the confounding policy. The proposition below derives the identified sets for $\tau_{c}$ and $\tau_{uc}$ under substitutability:

\begin{lemma}
    Suppose Assumptions \ref{assump:continuity}, \ref{assump:disc_prob_treat}, \ref{assump:bounded_outcomes} and \ref{assump:subst} hold. Then,

    \begin{equation*}
        \tau_{c}\in\left[y^{min}-y^{max},y^{max}-Y_{1}^{-}\right]
    \end{equation*}

    and

    \begin{equation*}
        \tau_{uc}\in\left[Y_{1}^{+}-y^{max},y^{max}-y^{min}\right]
    \end{equation*}
\end{lemma}

\subsection{Partial Identification Combining Bounded Variation and Modularity Assumptions}

So far in this section, we have described two approaches to partially identify the target parameters in the DiDC setting under different sets of assumptions.

The bounded variation approach partially identifies $\tau_{c}$ as a function of a measure of deviation from Assumption \ref{assump:conf_time_invar}. The procedure has several desirable features. It allows for researchers to partially identify the parameters under different values for the deviation, which also allows researchers to assess the robustness of empirical findings. However, it does not allow for the partial identification of $\tau_{uc}$, not at least without further assumptions regarding the homogeneity of treatment effects.

The modularity approach, on the other hand, allows for the partial identification of $\tau_{c}$ ad $\tau_{uc}$, by invoking assumptions on the interaction of the confounding policy and the treatment of interest which can be rooted in economic theory or background knowledge of the empirical setting. However, the procedure does not allow researchers to exploit the temporal structure of the data. Nevertheless, it is (particularly) useful for partial identification in RD settings in which there is a confounding policy at the threshold, but no data from periods before the implementation of the treatment of interest.

In the lemmas below, we show that combining bounded variation and modularity assumptions can strengthen the bounds on $\tau_{uc}$ and $\tau_{c}$:

\begin{lemma}\label{lemma:boundvar_comp}
    Suppose Assumptions \ref{assump:continuity}, \ref{assump:disc_prob_treat} and \ref{assump:bounded_variation}-\ref{assump:compl} hold. Then, $\tau_{c}\in\left[\tau_{c}^{LB},\tau_{c}^{UB}\right]$ and $\tau_{uc}\in\left[\tau_{uc}^{LB},\tau_{uc}^{UB}\right]$, where

    \begin{align*}
        &\tau_{c}^{LB}=\max\left\{\Delta Y^{+}-c_{1},(\Delta Y^{+}-\Delta Y^{-})-c_{2},y^{min}-Y_{1}^{-}\right\}\\&\tau_{c}^{UB}=\min\left\{\Delta Y^{+}+c_{1},(\Delta Y^{+}-\Delta Y^{-})+c_{2},y^{max}-y^{min}\right\}
    \end{align*}

    and 

    \begin{align*}
    &\tau_{uc}^{LB}=y^{min}-y^{max}\\&\tau_{uc}^{UB}=\min\left\{\Delta Y^{+}+c_{1},Y_{1}^{+}-y^{min}\right\}
    \end{align*}
    
\end{lemma}

\begin{lemma}\label{lemma:boundvar_subs}
    Suppose Assumptions \ref{assump:continuity}, \ref{assump:disc_prob_treat} and \ref{assump:bounded_variation}, \ref{assump:bounded_outcomes} and \ref{assump:subst} hold. Then, $\tau_{c}\in\left[\tau_{c}^{LB},\tau_{c}^{UB}\right]$ and $\tau_{uc}\in\left[\tau_{uc}^{LB},\tau_{uc}^{UB}\right]$, where

    \begin{align*}
        &\tau_{c}^{LB}=\max\left\{\Delta Y^{+}-c_{1},(\Delta Y^{+}-\Delta Y^{-})-c_{2},y^{min}-y^{max}\right\}\\&\tau_{c}^{UB}=\min\left\{\Delta Y^{+}+c_{1},(\Delta Y^{+}-\Delta Y^{-})+c_{2},y^{max}-Y_{1}^{-}\right\}
    \end{align*}

    and 

    \begin{align*}
    &\tau_{uc}^{LB}=\max\left\{\Delta Y^{+}-c_{1},Y_{1}^{+}-y^{max}\right\}\\&\tau_{uc}^{UB}=y^{max}-y^{min}
    \end{align*}
    
\end{lemma}

Lemmas \ref{lemma:boundvar_comp} and \ref{lemma:boundvar_subs} derive the identified sets for $\tau_{c}$ and $\tau_{uc}$ when bounded variation, bounded outcomes and modularity assumptions are combined. When it comes to $\tau_{c}$, the combination of assumptions strengthens the lower and the upper bound under both modularity assumptions. When it comes to $\tau_{uc}$, however, only one of the bounds is strengthened by the combination of assumptions, whereas the other bounds remains the worst-case bound. When it comes to estimation and inference, the bootstrap procedure outlined in Theorem 2 remains valid for bounds under alternative sets of assumptions as well.

\section{Monte Carlo Simulations}\label{sec:Simulations}

We now analyze the finite-sample properties of the Diff-in-Disc estimator through Monte Carlo simulations and compare the performance of the proposed estimator to that of the local linear RDD estimator proposed by \cite{calonico2014robust}, the nonparametric DiD regression estimator proposed by \cite{santanna2020doubly}, and the Two-Way Fixed Effects (TWFE) estimator. We consider Data Generating Processes (DGP) based on model 3 from \cite{calonico2014robust}, with small modifications that will be described.

We conduct our simulation studies in four distinct settings: with time-invariant confounding factors at the threshold and without any confounding factors, which is the typical scenario for RDDs, for both time-invariant and time-varying functional forms of the conditional means of potential outcomes. For each simulation, we conduct 1000 replications, and for each replication, we consider a sample size $n=1000$, with $Z_{i}\sim(2\mathcal{B}(2,4)-1)$ where $\mathcal{B}(p_1,p_2)$ is a beta distribution with parameters $p_1$ and $p_2$. We also consider $\varepsilon_{it}\sim N(0,\sigma^{2}_{\varepsilon})$, $\sigma_{\varepsilon}=0.1295$ and the outcome generated is $Y_{i,t}=\mu_{it}(Z_{i})+\varepsilon_{i,t}$. The detailed specifications and additional functional forms of both models are provided in Appendix \ref{app:Simulations} for clarity. In both scenarios, we observe that our DiDC estimator performs better than the RDD estimator, with a smaller bias and improved coverage. Additionally, we find that the DiDC estimator has smaller bias and better coverage than the DiD estimator when the DGPs change over time.

\subsection{Identical Functional Forms over Time}

In the first simulation, we mimic a scenario where one or more confounding discontinuities are present at the threshold $z_0=0$, but the functional forms for conditional means of potential outcomes are the same in both periods. This setting is comparable to that of \cite{grembi2016fiscal}, where the treatment of interest was introduced at some point between $t=0$ and $t=1$ and was given to units whose running variable values $Z_i$ are above the threshold $z_0=0$ however there were other pre-existing treatments determined by the same threshold $z_0=0$ on the same running variable.

The second model we consider is a scenario with identical functional forms over time, the only distinction being that in period $t=1$, there is an effect of treatment $\tau$ for units with $Z_i \geq 0$. Importantly, there are no other sources of discontinuity in the outcome at the threshold $z_0=0$, making this an ideal scenario for estimating the effect of the treatment using a regular RDD.

We present the results of these simulation studies in Tables 1 and 2. We estimate the Diff-in-Disc along with the RDD, DiD and TWFE, and compare the average bias, median bias, root-mean-squared errors, 95\% coverage probability, and the length of the 95\% confidence interval for each estimator when the treatment effect $\tau$ is equal to 0. 

\begin{table}[h]
\centering
\caption{Identical Functional Forms with Confounding}
\label{tab:sim_results_confounder}
\begin{tabular}{c|cccccc}
\hline
             & AV.Bias & RMSE  & Coverage & CIL   & \multicolumn{2}{c}{Bandwidths} \\
             &         &       &          &       & $h_{n}$             & $b_{n}$            \\ \hline
RDD (Robust) & 1.043   & 1.043 & 0        & 0.175 & 0.169          & 0.318         \\
Diff-in-Disc & -0.002  & 0.047 & 0.944    & 0.228 & 0.225          & 0.360         \\
NP-DiD       & 0.001   & 0.019 & 0.953        & 0.393 & -              & -             \\
TWFE         & 0.001   & 0.011 & 1        & 0.397 & -              & -             \\ \hline
\end{tabular}
\\\scriptsize \noindent \textit{Note:} Simulations based on 10,000 Monte Carlo experiments with a sample size $n=1000$. RDD is the non-bias corrected RD estimator from CCT (2014), NP-DiD is the outcome regression DiD estimator from Sant'Anna and Zhao (2020), DiDC is the estimator proposed in this paper.  “Av. Bias”, “RMSE”, “Cover”, “CIL’ stand for the average simulated bias, simulated root mean-squared errors, 95\% coverage probability and 95\% confidence interval length, respectively. "Bandwidths" $h_n$ and $b_n$ report the plug-in bandwidths for point and bias estimation, respectively.
\end{table}

Table~\ref{tab:sim_results_confounder} shows a significant improvement over the RD estimator when a confounding factor is present at the threshold. The result is not surprising, as the presence of the confounder violates the continuity assumption that underlies the validity of cross-sectional RD. Note also that the mean size of the bandwidths, both for the estimator and for the bias-correction, is greater for the Diff-in-Disc estimator than the RD estimator. In terms of bias, both the nonparametric DiD and the TWFE estimators exhibit desirable finite-sample properties. However, when it comes to coverage of the confidence interval, we find that those from the TWFE estimator are severely biased.

\begin{table}[h]
\centering
\caption{Identical Functional Forms without Confounding}
\label{tab:sim_results_noconfounder}
\begin{tabular}{c|cccccc}
\hline
             & AV.Bias & RMSE  & Coverage & CIL   & \multicolumn{2}{c}{Bandwidths} \\
             &         &       &          &       & $h_{n}$             & $b_{n}$            \\ \hline
RDD (Robust) & 0.043   & 0.051 & 0.861    & 0.176 & 0.169          & 0.318         \\
Diff-in-Disc & -0.003  & 0.047 & 0.944    & 0.228 & 0.223          & 0.360         \\
NP-DiD       & 0.001   & 0.019 & 0.955    & 0.391 & -              & -             \\
TWFE         & 0.002   & 0.012 & 1        & 0.397 & -              & -             \\ \hline
\end{tabular}
\\\scriptsize \noindent \textit{Note:} Simulations based on 10,000 Monte Carlo experiments with a sample size $n=1000$. RDD is the non-bias corrected RD estimator from CCT (2014), NP-DiD is the outcome regression DiD estimator from Sant'Anna and Zhao (2020), DiDC is the estimator proposed in this paper.  “Av. Bias”, “RMSE”, “Cover”, “CIL’ stand for the average simulated bias, simulated root mean-squared errors, 95\% coverage probability and 95\% confidence interval length, respectively. "Bandwidths" $h_n$ and $b_n$ report the plug-in bandwidths for point and bias estimation, respectively.
\end{table}

Table~\ref{tab:sim_results_noconfounder} shows the results for the case where there is no confounding policy at the threshold. In that case, the cross-sectional RD is valid, as evidenced by the small finite-sample bias in the simulation. Once again, the Diff-in-Disc estimator has smaller finite-sample bias and larger bandwidths, which illustrates the point in Section~\ref{sec:Estimation_sharp} that, even in settings where the standard RD is valid, there might be gains in using the Diff-in-Disc approach.

The results show that Difference-in-Discontinuites estimator has an apparent improvement over the standard RD, yielding a smaller bias and better coverage. This is noteworthy as it highlights that even in cases where the RDD would traditionally be regarded as suitable, the differences-in-discontinuity approach can give more desirable results by incorporating more data into the estimation. The results hold when the functional format exhibits little temporal variation, as shown in the next section, with the additional period contributing to more reliable estimates than the RD whenever the bias from the extra period $t=0$ does not exceed that from period $t=1$. In Appendix~\ref{app:Simulations}, we show that the results for the simulations are robust even in the case where the variance of potential outcomes vary over time.

\subsection{Time-varying Functional Forms}

Next, we introduce scenarios where the functional forms for conditional means change between periods, contrasting to the prior section where functional forms were identical across time. This setup allows us to analyze how changes in the conditional mean of potential outcomes affect the performance of the considered estimators.

Again, we consider simulations with and without confounders at the threshold, but now we alter the model for $t=1$ to be a linear model derived from the original. The functional form for time period $t=0$ is identical to that of models 1 and 2. Results are shown in Tables 3 and 4.

\begin{table}[h]
\centering
\caption{Time-Varying Functional Forms with Confounding}
\begin{tabular}{c|cccccc}
\hline
             & AV.Bias & RMSE  & Coverage & CIL   & \multicolumn{2}{c}{Bandwidths} \\
             &         &       &          &       & $h_{n}$             & $b_{n}$            \\ \hline
RDD (Robust) & 1.039   & 1.039 & 0        & 0.161 & 0.221          & 0.360         \\
Diff-in-Disc & -0.001  & 0.050 & 0.937    & 0.241 & 0.184          & 0.331         \\
NP-DiD       & 1.766   & 1.766 & 0        & 0.297 & -              & -             \\
TWFE         & -1.568  & 1.568 & 0        & 0.272 & -              & -             \\ \hline
\end{tabular}
\\\scriptsize \noindent \textit{Note:} Simulations based on 10,000 Monte Carlo experiments with a sample size $n=1000$. RDD is the non-bias corrected RD estimator from CCT (2014), NP-DiD is the outcome regression DiD estimator from Sant'Anna and Zhao (2020), DiDC is the estimator proposed in this paper.  “Av. Bias”, “RMSE”, “Cover”, “CIL’ stand for the average simulated bias, simulated root mean-squared errors, 95\% coverage probability and 95\% confidence interval length, respectively. "Bandwidths" $h_n$ and $b_n$ report the plug-in bandwidths for point and bias estimation, respectively.
\end{table}

Table 3 shows that in the presence of changes in functional forms over time and confounding policies at the threshold, only the Diff-in-Disc approach is valid. The finite-sample bias of the standard RD estimator is close to the one presented in Table 1. The main difference in Table 3 in comparison to Table 1 is the poor finite-sample properties of DiD methods, as both the nonparametric DiD and the TWFE estimator exhibit larger finite-sample bias than the standard RD.

\begin{table}[h]
\centering
\caption{Time-Varying Functional Forms without Confounding}
\begin{tabular}{c|cccccc}
\hline
             & AV.Bias & RMSE  & Coverage & CIL   & \multicolumn{2}{c}{Bandwidths} \\
             &         &       &          &       & $h_{n}$             & $b_{n}$            \\ \hline
RDD (Robust) & 0.039   & 0.049 & 0.922    & 0.161 & 0.221          & 0.360         \\
Diff-in-Disc & -0.008  & 0.050 & 0.937    & 0.241 & 0.184          & 0.331         \\
NP-DiD       & 1.766   & 1.766 & 0        & 0.297 & -              & -             \\
TWFE         & -1.568  & 1.568 & 0        & 0.272 & -              & -             \\ \hline
\end{tabular}
\\\scriptsize \noindent \textit{Note:} Simulations based on 10,000 Monte Carlo experiments with a sample size $n=1000$. RDD is the non-bias corrected RD estimator from CCT (2014), NP-DiD is the outcome regression DiD estimator from Sant'Anna and Zhao (2020), DiDC is the estimator proposed in this paper.  “Av. Bias”, “RMSE”, “Cover”, “CIL’ stand for the average simulated bias, simulated root mean-squared errors, 95\% coverage probability and 95\% confidence interval length, respectively. "Bandwidths" $h_n$ and $b_n$ report the plug-in bandwidths for point and bias estimation, respectively.
\end{table}

Table 4 displays the results for the case where functional forms change over time, but there is no confounding policy at the threshold. Once again, the standard RD and the Diff-in-Disc exhibit desirable finite-sample properties, with the Diff-in-Disc showing better coverage and smaller bias. However, unlike when functional forms are constant over time, when functional forms change, the standard RD method yields larger optimal bandwidths. For the DiD estimators, the results show severe bias, and the performance of the nonparametric DiD and the TWFE are similar to what we observed in Table 3. In Appendix D.2, we show that the simulation results are robust to changes in the time-varying variance.

Overall, the Monte Carlo results show several desirable properties of the Diff-in-Disc estimator. Not only it remains unbiased when the alternative approaches are not valid, but it also exhibits smaller finite-sample bias than the standard RD even in the absence of confounders. In the next section, we revisit a well-known political economy setting to analyze the estimator's performance on a real dataset.

\section{Empirical Illustration}\label{sec:Empirical_Illustration}

We illustrate the use of our estimator by revisiting the empirical application in \cite{grembi2016fiscal}, which analyzes the impact of fiscal rules on Italian municipal finances by exploiting a 2001 fiscal rule relaxation as a natural experiment. The relaxation applied to municipalities with fewer than 5,000 inhabitants, which therefore form the treatment group, while municipalities above this threshold serve as controls.

Our objective is to replicate their setting using the Difference-in-Discontinuities (DiDC) estimator, compare the results to the original findings, and assess the validity of the identifying assumption using the constant-confounder test from Section~\ref{subsec:Test_conf_time_invar}.

The dataset comprises data from Italian municipalities, focusing on the period surrounding the government's relaxation of fiscal rules in 2001. Municipalities with fewer than 5,000 inhabitants experienced a relaxation of fiscal rules and were treated, while those with more than 5,000 inhabitants served as controls.

We implement a 2×2 DiDC design, computing before–and–after differences and estimating local linear RDDs as detailed in Section \ref{sec:Estimation_sharp}. \cite{grembi2016fiscal}'s original study utilized a large panel dataset, a rectangular kernel and a polynomial of degree one as in the model below: 
{\small \begin{align*}
Y_{it} = &\delta_0 + \delta_1 (X_{it}-c) + S_i (\gamma_0 + \gamma_1 (X_{it}-c)) & T_i[\alpha_0 + \alpha_1 (X_{it}-c) + S_i(\beta_0 + \beta_1 (X_{it}-c))]+\varepsilon_{it}.
\end{align*}}
where $S_i$ is a dummy variable for cities below 5,000 (treatment indicator), $T_i$ is a dummy variable for the post-treatment period and $\beta_0$ is the parameter of interest.

Table 5 compares our DiDC estimates (difference of RDDs), along with the effective bandwidths and sample sizes to \cite{grembi2016fiscal}'s estimations. Due to differences in specifications, the results must be interpreted with caution.

\begin{table}[h]\label{tab:Grembi_CCTbw}
    \centering
    \caption{Effects of relaxing a fiscal rule - CCT (2014) Bandwidths}
    \begin{tabular}{cccc}
        \hline
        Estimators   & Deficit & Fiscal Gap & Taxes \\ \hline
        Diif-in-Disc (Grembi et al, 2016) &
          \begin{tabular}[c]{@{}c@{}}17.495\\ (7.737)\end{tabular} &
          \begin{tabular}[c]{@{}c@{}}59.468\\ (32.079)\end{tabular} &
          \begin{tabular}[c]{@{}c@{}}-76.083\\ (32.597)\end{tabular} \\
        Bandwidth    & 600     & 513        & 378   \\
        Observations & 2414    & 2136       & 1536  \\ \hline
        Diff-in-Disc (RD of $\Delta$) &
          \begin{tabular}[c]{@{}c@{}}18.501\\ (22.318)\end{tabular} &
          \begin{tabular}[c]{@{}c@{}}10.247\\ (21.724)\end{tabular} &
          \begin{tabular}[c]{@{}c@{}}-5.288\\ (4.803)\end{tabular} \\
        Bandwidth    & 520     & 591        & 592   \\
        Observations & 392    & 433       & 434  \\ \hline
        Diff-in-Disc ($\Delta$ of RDs) &
          \begin{tabular}[c]{@{}c@{}}22.142\\ (25.136)\end{tabular} &
          \begin{tabular}[c]{@{}c@{}}9.210\\ (46.658)\end{tabular} &
          \begin{tabular}[c]{@{}c@{}}-9.009\\ (23.824)\end{tabular} \\
        Bandwidth    & 560     & 475        & 334   \\
        Observations & 415     & 362       & 264  \\ \hline
    \end{tabular}
    \\\scriptsize \noindent \textit{Note:} The first panel presents the estimates from \cite{grembi2016fiscal} using the whole panel from 1999 to 2004. The second panel presents the estimates obtained from the RD of first-differences using data from 2000 and 2001. The third panel presents the estimates obtained from the difference of RDs using data from 2000 and 2001.
\end{table}

The results are qualitatively similar, yet, our specification is underpowered due to the smaller sample size, and thus estimates are not statistically significant.

\subsection{Testing the Assumptions}
We conduct the test from Section \ref{subsec:Test_conf_time_invar} to evaluate the assumption of time-invariant confounding effects. To test this, we estimate stacked RDs for the years 1998--2000, interacting the running variable with treatment and period dummies. Table~\ref{tab:stacked_rd_smallest} reports the interaction coefficients ($z\times D$) for the smallest bandwidth, and Table~\ref{tab:joint_wald} summarizes the corresponding joint Wald tests for equality of these coefficients across years.

\begin{table}[ht]
\centering
\caption{Stacked-RD interaction coefficients by pre-period (smallest bandwidth)}
\label{tab:stacked_rd_smallest}
\begin{tabular}{lccc}
\hline
Outcome & $\theta_{2000}$ & $\theta_{1999}$ & $\theta_{1998}$ \\
\hline
Taxes       & 1.582     & 1.373     & 1.248     \\
            & (0.083)   & (0.084)   & (0.084)   \\
Deficit     & 0.095     & 0.114     & 0.112     \\
            & (0.020)   & (0.020)   & (0.021)   \\
Fiscal Gap  & 1.273  & 1.338     & 1.290     \\
            & (0.073)   & (0.074)   & (0.075)   \\
\hline
\end{tabular}\\
\scriptsize \noindent
\textit{Note:} Estimates from weighted stacked OLS (triangular kernel). Standard errors in parentheses. ``Smallest bandwidth'' = minimum period-specific CCT bandwidth.
\scriptsize
\end{table}

The joint test of $H_0:\theta_{2000}=\theta_{1999}=\theta_{1998}$ is implemented using a Wald $F$-test based on the pooled regressions. The results are shown in Table~\ref{tab:joint_wald}.

\begin{table}[ht]
\centering
\caption{Joint Wald test: $H_0:\theta_{2000}=\theta_{1999}=\theta_{1998}$ (F-stat, df, p-value)}
\label{tab:joint_wald}
\begin{tabular}{lcccc}
\hline
Outcome & Bandwidth & F-statistic& p-value \\
\hline
Taxes       & Smallest & 4.108 & \textbf{0.017} \\
 & Biggest  & 4.545 & \textbf{0.011} \\
 & CCT      & 3.010 & \textbf{0.046} \\
\hline
Deficit     & Smallest & 0.279 & 0.757 \\
 & Biggest  & 0.181 & 0.835 \\
 & CCT      & 0.294  & 0.745 \\
\hline
Fiscal balance & Smallest & 0.213  & 0.809 \\
 & Biggest  & 0.141 & 0.869 \\
 & CCT      & 0.182   & 0.833 \\
\hline
\end{tabular}\\
\scriptsize \noindent
\textit{Note:}{F-statistics and p-values are from the joint Wald test implemented with \texttt{car::linearHypothesis} \citep{fox2001car}. Bold p-values indicate rejection at 5\%.}
\end{table}

For taxes, the null hypothesis of time-invariant confounding is rejected at conventional significance levels (p $\in$ [0.01, 0.05]), indicating that the RD discontinuity varied across pre-treatment years. In contrast, the test fails to reject $H_0$ for both deficit and fiscal balance (p-values between 0.75 and 0.87).

In summary, our estimated treatment effects are similar in magnitude to the difference-in-discontinuities approach utilized in \cite{grembi2016fiscal} study. Notably, our difference-in-RDDs approach yields smaller confidence intervals, suggesting it is the most powerful. However, our validity test indicates that confounding effects vary over time, suggesting that the use of DiDC in this setting leads to biased estimates. We also replicate our estimates using the bandwidths in \cite{ludwig2007does}, with similar results. These estimates are shown in Table \ref{tab_Grembi_LMbw} in Appendix \ref{app:Empirical}.

We also implement the KS-type test from Section~\ref{subsec:test_functiona_form} to assess whether the shapes of the conditional mean functions remain stable across pre-treatment years. Applying this test to each outcome, we find strong evidence of time-varying functional forms for taxes and fiscal balance: on both sides of the cutoff, the bootstrap p-values for imposte and saldo are essentially zero (0.001 on both sides), indicating clear violations of the time-invariance condition. For the deficit, however, the evidence is mixed. The left-hand side yields a p-value of 0.266, consistent with time-invariant functional forms, whereas the right-hand side again produces a very small p-value (0.001), suggesting that the relationship between the running variable and the deficit outcome changed over time for municipalities above the 5,000-inhabitant threshold. Overall, these results show that the shape of the outcome–running-variable relationship is not stable across pre-treatment years for most outcomes, reinforcing the need for careful consideration when applying the method.

\subsection{Partial Identification}
\subsubsection{Bounded Variation Assumptions}
Table~\ref{tab:partial_id_imposte} reports the identified sets for the treatment effect on Taxes under the bounded-variation sensitivity framework. Each cell shows the lower and upper bound for the identified set $\left[\tau_{c}^{LB},\,\tau_{c}^{UB}\right]$ as a function of the sensitivity parameters $c_1$ and $c_2$. The bounds are constructed by allowing the conditional potential-outcome functions and the confounding effect to deviate from the time-invariance baseline up to  $c_1$ and $c_2$, respectively (Section~\ref{sec:partial}). The same grid procedure was used to produce the analogous tables for Fiscal balance (Tables~\ref{tab:partial_id_fiscal}--\ref{tab:partial_id_deficit} in Appendix~\ref{app:Empirical}) and Deficit (Tables 12 and 13 from Section~\ref{sec:app_bounded_modularity}).

\begin{table}[h]
\centering
\caption{Bounds on the Treatment Effect under Bounded Variation \textit{(Outcome: Taxes)}}
\label{tab:partial_id_imposte}
\resizebox{\textwidth}{!}{%
\begin{tabular}{c|ccccccccccc}
\toprule
$c_{2}\backslash c_{1}$ & 0.0 & 0.5 & 1.0 & 1.5 & 2.0 & 2.5 & 3.0 & 3.5 & 4.0 & 4.5 & 5.0  \\
\midrule
0.0 & 6.33 , -6.57 & 5.83 , -6.57 & 5.33 , -6.57 & 4.83 , -6.57 & 4.33 , -6.57 & 3.83 , -6.57 & 3.33 , -6.57 & 2.83 , -6.57 & 2.33 , -6.57 & 1.83 , -6.57 & 1.33 , -6.57 \\
0.5 & 6.33 , -6.07 & 5.83 , -6.07 & 5.33 , -6.07 & 4.83 , -6.07 & 4.33 , -6.07 & 3.83 , -6.07 & 3.33 , -6.07 & 2.83 , -6.07 & 2.33 , -6.07 & 1.83 , -6.07 & 1.33 , -6.07 \\
1.0 & 6.33 , -5.57 & 5.83 , -5.57 & 5.33 , -5.57 & 4.83 , -5.57 & 4.33 , -5.57 & 3.83 , -5.57 & 3.33 , -5.57 & 2.83 , -5.57 & 2.33 , -5.57 & 1.83 , -5.57 & 1.33 , -5.57 \\
1.5 & 6.33 , -5.07 & 5.83 , -5.07 & 5.33 , -5.07 & 4.83 , -5.07 & 4.33 , -5.07 & 3.83 , -5.07 & 3.33 , -5.07 & 2.83 , -5.07 & 2.33 , -5.07 & 1.83 , -5.07 & 1.33 , -5.07 \\
2.0 & 6.33 , -4.57 & 5.83 , -4.57 & 5.33 , -4.57 & 4.83 , -4.57 & 4.33 , -4.57 & 3.83 , -4.57 & 3.33 , -4.57 & 2.83 , -4.57 & 2.33 , -4.57 & 1.83 , -4.57 & 1.33 , -4.57 \\
2.5 & 6.33 , -4.07 & 5.83 , -4.07 & 5.33 , -4.07 & 4.83 , -4.07 & 4.33 , -4.07 & 3.83 , -4.07 & 3.33 , -4.07 & 2.83 , -4.07 & 2.33 , -4.07 & 1.83 , -4.07 & 1.33 , -4.07 \\
3.0 & 6.33 , -3.57 & 5.83 , -3.57 & 5.33 , -3.57 & 4.83 , -3.57 & 4.33 , -3.57 & 3.83 , -3.57 & 3.33 , -3.57 & 2.83 , -3.57 & 2.33 , -3.57 & 1.83 , -3.57 & 1.33 , -3.57 \\
3.5 & 6.33 , -3.07 & 5.83 , -3.07 & 5.33 , -3.07 & 4.83 , -3.07 & 4.33 , -3.07 & 3.83 , -3.07 & 3.33 , -3.07 & 2.83 , -3.07 & 2.33 , -3.07 & 1.83 , -3.07 & 1.33 , -3.07 \\
4.0 & 6.33 , -2.57 & 5.83 , -2.57 & 5.33 , -2.57 & 4.83 , -2.57 & 4.33 , -2.57 & 3.83 , -2.57 & 3.33 , -2.57 & 2.83 , -2.57 & 2.33 , -2.57 & 1.83 , -2.57 & 1.33 , -2.57 \\
4.5 & 6.33 , -2.07 & 5.83 , -2.07 & 5.33 , -2.07 & 4.83 , -2.07 & 4.33 , -2.07 & 3.83 , -2.07 & 3.33 , -2.07 & 2.83 , -2.07 & 2.33 , -2.07 & 1.83 , -2.07 & 1.33 , -2.07 \\
5.0 & 6.33 , -1.57 & 5.83 , -1.57 & 5.33 , -1.57 & 4.83 , -1.57 & 4.33 , -1.57 & 3.83 , -1.57 & 3.33 , -1.57 & 2.83 , -1.57 & 2.33 , -1.57 & 1.83 , -1.57 & 1.33 , -1.57 \\
\bottomrule
\end{tabular}}\\
\scriptsize \noindent
\textit{Note}: Estimates for the identified sets for the effects on fiscal rules on taxes under Lemma 6.
\end{table}

\subsubsection{Modularity Assumptions}
Under Assumptions \ref{assump:bounded_outcomes} and \ref{assump:compl}, we construct worst-case modularity bounds using $y_{min}$ and $y_{max}$ equal to the minimum and maximum observed values of each outcome in the data (Table~\ref{tab:bounds_modularity}). As expected in this worst-case setting, the resulting identified sets are very wide: for example, the bounds for taxes span roughly $[-157.59,\; 1203.19]$ for $\tau_{c}$ and $[-1203.19,\; 108.08]$ for $\tau_{uc}$, with similarly large intervals for the fiscal gap and deficit. Consequently, these bounds don't provide much information. To obtain informative conclusions about the direction or magnitude of the treatment effect, researchers must impose stronger and more economically grounded restrictions.

\begin{table}[h]
\centering
\caption{Identified Sets Under Modularity Assumptions}
\label{tab:bounds_modularity}
\begin{tabular}{lcc}
\toprule
\textbf{Outcome} & \(\tau_c \in [\tau_c^{LB}, \tau_c^{UB}]\) 
                 & \(\tau_{uc} \in [\tau_{uc}^{LB}, \tau_{uc}^{UB}]\) \\
\midrule
Taxes
 & $[-157.59,\; 1203.19]$
 & $[-1203.19,\; 108.08]$ \\[2pt]
Fiscal Gap
 & $[-738.37,\; 2358.18]$
 & $[-2358.18,\; 862.38]$ \\[2pt]
Deficit
 & $[-757.29,\; 1687.17]$
 & $[-1687.17,\; 758.35]$ \\
\bottomrule
\end{tabular}\\
\scriptsize \noindent
\textit{Note}: Estimates for the identified sets for the effects on fiscal rules on public finance outcomes under Lemma 7.
\end{table}

\subsubsection{Combining Bounded Variation and Modularity Assumptions}\label{sec:app_bounded_modularity}
When the two approaches are combined, identification becomes substantially sharper. Together, these restrictions eliminate many of the extreme scenarios allowed by modularity alone and narrow the identified sets relative to either assumption in isolation.

The Deficit outcome clearly illustrates the value of this combination. Under worst-case modularity, the identified sets for Deficit include both large negative and positive effects, but once bounded-variation constraints are imposed, the intervals become notably narrower and remain strictly negative across all $c_1$ and $c_2$ values considered. Table \ref{tab:partial_comb_deficit} shows that the bounds for $\tau_c$ fall roughly between $-27$ and $-19$, and Table \ref{tab:partial_comb_tauuc_deficit} similarly reports negative upper endpoints for $\tau_{uc}$. The combined assumptions, therefore, identify not only the sign of the effect but also restrict its magnitude for $\tau_c$ to a reasonably small range. 

\begin{table}[h]
\centering
\caption{Bounds for $\tau_c$ — Outcome: Deficit. $y_{min} = -762.5508$, $y_{max} = 924.6157$. }
\label{tab:partial_comb_deficit}
\resizebox{\textwidth}{!}{%
\begin{tabular}{c|ccccccccccc}
\toprule
$c_{2}\backslash c_{1}$ & 0.0 & 0.5 & 1.0 & 1.5 & 2.0 & 2.5 & 3.0 & 3.5 & 4.0 & 4.5 & 5.0 \\
\midrule
0.0 & 19.60 , 27.86 & 20.10 , 27.86 & 20.60 , 27.86 & 21.10 , 27.86 & 21.60 , 27.86 & 22.10 , 27.86 & 22.60 , 27.86 & 23.10 , 27.86 & 23.60 , 27.86 & 24.10 , 27.86 & 24.60 , 27.86 \\
0.5 & 19.60 , 27.36 & 20.10 , 27.36 & 20.60 , 27.36 & 21.10 , 27.36 & 21.60 , 27.36 & 22.10 , 27.36 & 22.60 , 27.36 & 23.10 , 27.36 & 23.60 , 27.36 & 24.10 , 27.36 & 24.60 , 27.36 \\
1.0 & 19.60 , 26.86 & 20.10 , 26.86 & 20.60 , 26.86 & 21.10 , 26.86 & 21.60 , 26.86 & 22.10 , 26.86 & 22.60 , 26.86 & 23.10 , 26.86 & 23.60 , 26.86 & 24.10 , 26.86 & 24.60 , 26.86 \\
1.5 & 19.60 , 26.36 & 20.10 , 26.36 & 20.60 , 26.36 & 21.10 , 26.36 & 21.60 , 26.36 & 22.10 , 26.36 & 22.60 , 26.36 & 23.10 , 26.36 & 23.60 , 26.36 & 24.10 , 26.36 & 24.60 , 26.36 \\
2.0 & 19.60 , 25.86 & 20.10 , 25.86 & 20.60 , 25.86 & 21.10 , 25.86 & 21.60 , 25.86 & 22.10 , 25.86 & 22.60 , 25.86 & 23.10 , 25.86 & 23.60 , 25.86 & 24.10 , 25.86 & 24.60 , 25.86 \\
2.5 & 19.60 , 25.36 & 20.10 , 25.36 & 20.60 , 25.36 & 21.10 , 25.36 & 21.60 , 25.36 & 22.10 , 25.36 & 22.60 , 25.36 & 23.10 , 25.36 & 23.60 , 25.36 & 24.10 , 25.36 & 24.60 , 25.36 \\
3.0 & 19.60 , 24.86 & 20.10 , 24.86 & 20.60 , 24.86 & 21.10 , 24.86 & 21.60 , 24.86 & 22.10 , 24.86 & 22.60 , 24.86 & 23.10 , 24.86 & 23.60 , 24.86 & 24.10 , 24.86 & 24.60 , 24.86 \\
3.5 & 19.60 , 24.36 & 20.10 , 24.36 & 20.60 , 24.36 & 21.10 , 24.36 & 21.60 , 24.36 & 22.10 , 24.36 & 22.60 , 24.36 & 23.10 , 24.36 & 23.60 , 24.36 & 24.10 , 24.36 & 24.60 , 24.36 \\
4.0 & 19.60 , 23.86 & 20.10 , 23.86 & 20.60 , 23.86 & 21.10 , 23.86 & 21.60 , 23.86 & 22.10 , 23.86 & 22.60 , 23.86 & 23.10 , 23.86 & 23.60 , 23.86 & 24.10 , 23.86 & 24.60 , 23.86 \\
4.5 & 19.60 , 23.36 & 20.10 , 23.36 & 20.60 , 23.36 & 21.10 , 23.36 & 21.60 , 23.36 & 22.10 , 23.36 & 22.60 , 23.36 & 23.10 , 23.36 & 23.60 , 23.36 & 24.10 , 23.36 & 24.60 , 23.36 \\
5.0 & 19.60 , 22.86 & 20.10 , 22.86 & 20.60 , 22.86 & 21.10 , 22.86 & 21.60 , 22.86 & 22.10 , 22.86 & 22.60 , 22.86 & 23.10 , 22.86 & 23.60 , 22.86 & 24.10 , 22.86 & 24.60 , 22.86 \\
\bottomrule

\end{tabular}}\\
\scriptsize \noindent
\textit{Note}: Estimates for the identified sets for the effect $\tau_{c}$ on fiscal rules on financial deficit under Lemma 10.
\end{table}

\begin{table}[ht]
\centering
\caption{Bounds for $\tau_{uc}$ — Outcome: Deficit. $y_{min} = -762.5508$, $y_{max} = 924.6157$. }
\label{tab:partial_comb_tauuc_deficit}
\resizebox{\textwidth}{!}{%
\begin{tabular}{c|ccccccccccc}
\toprule
$c_{2}\backslash c_{1}$ & 0.0 & 0.5 & 1.0 & 1.5 & 2.0 & 2.5 & 3.0 & 3.5 & 4.0 & 4.5 & 5.0 \\
\midrule
0.0 & 1687.17 , 27.86 & 1687.17 , 27.36 & 1687.17 , 26.86 & 1687.17 , 26.36 & 1687.17 , 25.86 & 1687.17 , 25.36 & 1687.17 , 24.86 & 1687.17 , 24.36 & 1687.17 , 23.86 & 1687.17 , 23.36 & 1687.17 , 22.86 \\
0.5 & 1687.17 , 27.86 & 1687.17 , 27.36 & 1687.17 , 26.86 & 1687.17 , 26.36 & 1687.17 , 25.86 & 1687.17 , 25.36 & 1687.17 , 24.86 & 1687.17 , 24.36 & 1687.17 , 23.86 & 1687.17 , 23.36 & 1687.17 , 22.86 \\
1.0 & 1687.17 , 27.86 & 1687.17 , 27.36 & 1687.17 , 26.86 & 1687.17 , 26.36 & 1687.17 , 25.86 & 1687.17 , 25.36 & 1687.17 , 24.86 & 1687.17 , 24.36 & 1687.17 , 23.86 & 1687.17 , 23.36 & 1687.17 , 22.86 \\
1.5 & 1687.17 , 27.86 & 1687.17 , 27.36 & 1687.17 , 26.86 & 1687.17 , 26.36 & 1687.17 , 25.86 & 1687.17 , 25.36 & 1687.17 , 24.86 & 1687.17 , 24.36 & 1687.17 , 23.86 & 1687.17 , 23.36 & 1687.17 , 22.86 \\
2.0 & 1687.17 , 27.86 & 1687.17 , 27.36 & 1687.17 , 26.86 & 1687.17 , 26.36 & 1687.17 , 25.86 & 1687.17 , 25.36 & 1687.17 , 24.86 & 1687.17 , 24.36 & 1687.17 , 23.86 & 1687.17 , 23.36 & 1687.17 , 22.86 \\
2.5 & 1687.17 , 27.86 & 1687.17 , 27.36 & 1687.17 , 26.86 & 1687.17 , 26.36 & 1687.17 , 25.86 & 1687.17 , 25.36 & 1687.17 , 24.86 & 1687.17 , 24.36 & 1687.17 , 23.86 & 1687.17 , 23.36 & 1687.17 , 22.86 \\
3.0 & 1687.17 , 27.86 & 1687.17 , 27.36 & 1687.17 , 26.86 & 1687.17 , 26.36 & 1687.17 , 25.86 & 1687.17 , 25.36 & 1687.17 , 24.86 & 1687.17 , 24.36 & 1687.17 , 23.86 & 1687.17 , 23.36 & 1687.17 , 22.86 \\
3.5 & 1687.17 , 27.86 & 1687.17 , 27.36 & 1687.17 , 26.86 & 1687.17 , 26.36 & 1687.17 , 25.86 & 1687.17 , 25.36 & 1687.17 , 24.86 & 1687.17 , 24.36 & 1687.17 , 23.86 & 1687.17 , 23.36 & 1687.17 , 22.86 \\
4.0 & 1687.17 , 27.86 & 1687.17 , 27.36 & 1687.17 , 26.86 & 1687.17 , 26.36 & 1687.17 , 25.86 & 1687.17 , 25.36 & 1687.17 , 24.86 & 1687.17 , 24.36 & 1687.17 , 23.86 & 1687.17 , 23.36 & 1687.17 , 22.86 \\
4.5 & 1687.17 , 27.86 & 1687.17 , 27.36 & 1687.17 , 26.86 & 1687.17 , 26.36 & 1687.17 , 25.86 & 1687.17 , 25.36 & 1687.17 , 24.86 & 1687.17 , 24.36 & 1687.17 , 23.86 & 1687.17 , 23.36 & 1687.17 , 22.86 \\
5.0 & 1687.17 , 27.86 & 1687.17 , 27.36 & 1687.17 , 26.86 & 1687.17 , 26.36 & 1687.17 , 25.86 & 1687.17 , 25.36 & 1687.17 , 24.86 & 1687.17 , 24.36 & 1687.17 , 23.86 & 1687.17 , 23.36 & 1687.17 , 22.86 \\
\bottomrule
\end{tabular}}\\
\scriptsize \noindent
\textit{Note}: Estimates for the identified sets for the effect $\tau_{uc}$ on fiscal rules on financial deficit under Lemma 10.
\end{table}

The results in Table \ref{tab:partial_comb_deficit} show that the bounds on the effects of the fiscal rules are overall uninformative, suggesting that the results are not robust to violations of the time-invariance assumptions. Tables 12 and 13, on the other hand, show that the bounds on the effects on deficit are robust both to violations on time-invariance assumptions and complementarity between treatment and confounding effects. The results in the table show that the evolution of the mean confounding effect and the mean confounded deficit at the threshold could exceed 6 euros per capita (roughly 30.000 euros given the 5.000 population rule), and the true effect of relaxing fiscal policy would still be positive. However, the setting is severely underpowered, which means that confidence intervals are uninformative regarding the true signal of the treatment effects\footnote{In Tables \ref{tab:partial_comb_deficit} and \ref{tab:partial_comb_tauuc_deficit}, which show informative identified sets, the lower bound of the confidence interval for the lower bound is always negative.}.

Finally, Table~\ref{tab:bounds_modularity} shows that modularity assumptions alone are not sufficient to yield informative identified sets for the treatment effects.

\section{Conclusion}\label{sec:Conclusion}

The difference-in-discontinuities (DiDC) design is emerging as a promising method for estimating causal inference, addressing the limitations of both regression discontinuity (RDD) and difference-in-difference (DiD) approaches. This paper lays the theoretical groundwork for DiDC, examines its identification assumptions, estimation procedures, and asymptotic properties. We showcase its advantages through Monte Carlo simulations and an empirical application.

DiDC can handle scenarios in which the control and treatment groups differ significantly, violating the parallel trends assumption of DiD, or when RDD encounters confounding factors at the threshold. By incorporating more information, DiDC eliminates bias in RDD estimates under specific assumptions about the data-generating processes.

However, it is important to give due attention to the identification assumptions, particularly the time-invariance of confounding effects. We propose a test based on stacked RDDs to assess its validity in practice. Additionally, DiDC requires the treatment effect to be independent of confounding policy, though we introduce a possible relaxation for potential interaction effects.

We find that the DiDC method can eliminate bias completely if the function shapes remain stable on both sides of the threshold over time. This suggests it could offer significant advantages over standard RDD estimators, even in settings where no other confounding variables are present at the threshold. We also propose a test to compare the derivatives of estimated functions on either side of the threshold, allowing researchers to evaluate the stability of data-generating processes over time.

Monte Carlo simulations demonstrate DiDC's potential to improve upon RDD, yielding lower bias and better coverage. It can provide more desirable results by incorporating more data, especially when the functional form exhibits minimal temporal variation. Notably, DiDC is the only viable approach when confounding factors render both RDD and DiD unsuitable. The empirical application highlights the importance of the time-invariance assumption. 

Future research directions include developing robust alternative estimators that are robust to violations of identification assumptions, as well as exploring other confidence interval methods tailored to the DiDC design. Overall, the DiDC method offers a valuable addition to the causal inference toolkit. It is applicable in settings where no other methods were previously available and shows potential to reduce bias in estimation in other settings.


\bibliographystyle{abbrvnat}
\nocite{imbens2008regression}
\nocite{rdrobust}
\bibliography{bibliography}

\pagebreak
\begin{appendices}\label{Appendices}
\setcounter{equation}{0}
\renewcommand{\theequation}{\thesection.\arabic{equation}}
\setcounter{assumption}{0}
\renewcommand{\theassumption}{\thesection.\arabic{assumption}}
\setcounter{lemma}{0}
\renewcommand{\thelemma}{\thesection.\arabic{lemma}}
\setcounter{table}{0}
\renewcommand{\thetable}{\thesection.\arabic{table}}

\section{Setup, assumptions and notation for estimation}\label{app:estimation_setup}

We construct the local polynomial estimator following \cite{calonico2014robust}. For a given $\nu \leq p \in \mathbf{N}$, the general estimand of interest is $\tau_\nu = \Delta \mu_{+} - \Delta \mu_{-}$ with $\Delta \mu^{(\nu)}_{+} = \nu!e^{\prime}_\nu \delta_{+,p}$, $\Delta \mu^{(\nu)}_{-} = \nu!e^{\prime}_\nu \delta_{-,p}$ being the $\nu$th-order derivatives of the $p$th-order local polynomial of the difference. The $p$th-order local polynomial estimators of the $\nu$th-order derivatives $\Delta \mu^{(\nu)}_{+,p}$ and $\Delta \mu^{(\nu)}_{-,p}$ are:
\begin{align*}
    \Delta \hat{\mu}^{(\nu)}_{+,p}(h_n) &= \nu!e^{\prime}_\nu \hat{\delta}_{+,p}(h_n)\\
    \Delta \hat{\mu}^{(\nu)}_{-,p}(h_n) &= \nu!e^{\prime}_\nu \hat{\delta}_{-,p}(h_n)\\
    \hat{\delta}_{\Delta Y+,p}(h_n) &= arg\min_{\delta \in \mathbf{R}^{p+1}} \sum_{i=1}^{n} \mathbbm{1}(Z_i \geq 0)(\Delta Y_{i}-r_p(Z_i)^{\prime} \delta)^2 K_{h_n}(Z_i)\\
    \hat{\delta}_{\Delta Y-,p}(h_n) &= arg\min_{\delta \in \mathbf{R}^{p+1}} \sum_{i=1}^{n} \mathbbm{1}(Z_i < 0)(\Delta Y_{i}-r_p(Z_i)^{\prime} \delta)^2 K_{h_n}(Z_i)
\end{align*}
where $e_\nu$ is a conformable $(\nu +1)$ unit vector, $K_h(u) = K(u/h)/h$, $h_n$ is a positive bandwidth sequence, 
    $r_p(x) = \begin{bmatrix} 
                    1 & x & \hdots & x^p
                \end{bmatrix}'$, 
    $ \Delta Y = \begin{bmatrix}
                    \Delta Y_{1} & \Delta Y_{2} & \hdots & \Delta Y_{n}
                \end{bmatrix}'$. We define $\chi_n = \begin{bmatrix}
                Z_1 & \hdots & Z_n
            \end{bmatrix}'$, 
        $\varepsilon_{\Delta Y} = \begin{bmatrix}
            \varepsilon_{\Delta Y,1} &  \hdots & \varepsilon_{\Delta Y,n}
                \end{bmatrix}'$
        with $\varepsilon_{\Delta Y,i} = \Delta Y_i - \mu_{\Delta Y}(Z_i)$, $\mu_{\Delta Y}(Z)=\E\left( \Delta Y| Z \right)$ and 
        
\begingroup
\allowdisplaybreaks
\begin{align*}
    S_p(h) &= \begin{bmatrix}
                \left(Z_1/h\right)^p & \hdots &\left(Z_n/h\right)^p
            \end{bmatrix}' \\
    Z_p(h) &= \begin{bmatrix}
               r_p(Z_1/h) & \hdots &
               r_p(Z_n/h)
            \end{bmatrix}' \\
    W_+(h) &= diag \left( \mathbbm{1}(Z_i \geq 0)K_h(Z_1), \hdots, \mathbbm{1}(Z_i \geq 0)K_h(Z_n) \right)\\
    W_-(h) &= diag \left( \mathbbm{1}(Z_i < 0)K_h(Z_1), \hdots, \mathbbm{1}(Z_i < 0)K_h(Z_n) \right)\\
    \Gamma_{+,p}(h) &= Z_p(h)' W_+(h) Z_p(h)/n \\
    \Gamma_{-,p}(h) &= Z_p(h)' W_-(h) Z_p(h)/n \\
    \vartheta_{+,p,q}(h) &= Z_p(h)'W_+(h)S_q(h)/n \\
    \vartheta_{-,p,q}(h) &= Z_p(h)'W_-(h)S_q(h)/n \\
    \Psi_{\Delta W \Delta Y+,p,q}(h,b) &= Z_p(h)' W_+(h) \Sigma_{\Delta W \Delta Y} W_+(b) Z_q(b) /n \\
    \Psi_{\Delta W \Delta Y-,p,q}(h,b) &= Z_p(h)' W_-(h) \Sigma_{\Delta W \Delta Y} W_-(b) Z_q(b) /n
\end{align*}
\endgroup

It follows that with $H_p(h) = diag \left( 1, h^{-1}, \hdots, h^{-p} \right)$:
\begin{align*}
    \hat{\delta}_{\Delta Y+,p}(h_n) &= H_p(h_n)\Gamma^{-1}_{+,p}(h_n)Z_p(h_n)'W_+(h_n)\Delta Y/n \\
    \hat{\delta}_{\Delta Y-,p}(h_n) &= H_p(h_n)\Gamma^{-1}_{-,p}(h_n)Z_p(h_n)'W_-(h_n)\Delta Y/n 
\end{align*}

The estimand and estimators are
\begin{align*}
& \tau^{DiDC}_\nu=\Delta \mu_{+}^{(\nu)}-\Delta \mu_{-}^{(\nu)}, \quad \Delta \mu_{+}^{(\nu)}=\nu!e_\nu^{\prime} \delta_{+, p}, \quad \Delta \mu_{-}^{(\nu)}=\nu!e_\nu^{\prime} \delta_{-, p}, \\
& \hat{\tau}^{DiDC}_{\nu, p}\left(h_n\right)=\Delta \hat{\mu}_{+, p}^{(\nu)}\left(h_n\right)-\Delta \hat{\mu}_{-, p}^{(\nu)}\left(h_n\right), \\
& \Delta \hat{\mu}_{+, p}^{(\nu)}\left(h_n\right)=\nu!e_\nu^{\prime} \hat{\delta}_{+, p}\left(h_n\right), \quad \Delta \hat{\mu}_{-, p}^{(\nu)}\left(h_n\right)=\nu!e_\nu^{\prime} \hat{\delta}_{-, p}\left(h_n\right),
\end{align*}
where, for any random variables $W$ and $X$, and $s \in \mathbb{N}$,
\begin{align*}
& \Delta \mu_{X+}^{(s)}=\lim _{x \rightarrow 0^{+}} \frac{\partial^s}{\partial z^s} \Delta \mu_X(z), \quad \Delta \mu_{X-}^{(s)}=\lim _{x \rightarrow 0^{-}} \frac{\partial^s}{\partial z^s} \Delta \mu_X(z), \\
& \Delta \mu_X(z)=\mathbb{E}[X \mid Z=z], \\
& \sigma_{X+}^2=\lim _{z \rightarrow 0^{+}} \sigma_X^2(z), \quad \sigma_{X-}^2=\lim _{z \rightarrow 0^{-}} \sigma_X^2(z), \\
& \sigma_X^2(z)=\mathbb{V}[X \mid Z=z], \\
\end{align*}

We employ the following assumptions on the \textit{sharp} model for the non-parametric local polynomial regression estimation:

\begin{assumption}\label{assump:sharp_model} 
    For some $\mathcal{K}_0 > 0$, the following holds in the neighborhood $\left( -\mathcal{K}_0 , \mathcal{K}_0 \right)$ around the cutoff $z_0=0$:
    \begin{enumerate}[(a)]
        \item $E \left[ \Delta Y_i^4 | Z_i = z \right]$ is bounded, and the density $f(z)$ of the random sample $Z_i$ is continuous and bounded away from zero.
        \item $\Delta \mu_-(z) =\E \left[ \Delta Y_{i}(0) | Z_{i} =z \right]$ and $\Delta \mu_+(z) =\E \left[ \Delta Y_{i}(1) | Z_{i} =z \right]$ are S times continuously differentiable.
        \item $\sigma^2_-(z) = V\left[ \Delta Y_{i}(0) | Z_{i} =z \right]$ and $\sigma^2_+(z) = V\left[ \Delta Y_{i}(1) | Z_{i} =z \right]$ are continuous and bounded away from zero. 
    \end{enumerate}
\end{assumption}

We also impose the following assumption on the kernel function to be employed in the estimator. This assumption allows for most of the commonly used kernels.

\begin{assumption}\label{assump:kernel}
    For some $\mathcal{K} > 0$, the kernel function $k(\cdot) :[0, \mathcal{K}] \rightarrow \mathbbm{R}$ is bounded and nonnegative, zero outside its support, symmetric around $z_0$ and positive and continuous on $\left( 0, \mathcal{K} \right)$.
\end{assumption}

Assumptions \ref{assump:sharp_model} and \ref{assump:kernel} limit the behavior of $E \left( \Delta Y_i | Z_i = z_0 \right)$ in the vicinity of the cutoff $z_0=0$.

\newpage
\section{Preliminary lemmas and results}\label{app:Lemmas}

Before proceeding, please refer to Appendix \ref{app:estimation_setup} for notation. This appendix restates, with minor adaptations, several lemmas, results and proofs from \cite{calonico2014robust} that are necessary for deriving the asymptotic results.

The following lemma establishes convergence in probability of the sample matrices $\Gamma_{-, p}\left(h_n\right)$, $\boldsymbol{\vartheta}_{-, p, q}\left(h_n\right)$, $\Psi_{-, p}\left(h_n\right)$ and $\Gamma_{+, p}\left(h_n\right)$, $\boldsymbol{\vartheta}_{+, p, q}\left(h_n\right)$, $\Psi_{+, p}\left(h_n\right)$ to their expectation counterparts, and characterizes those limits.

\begin{lemma} \label{lemma:s.a.1}
    Suppose Assumptions $1-2$ hold, and $n h_n \rightarrow \infty$.
        \begin{enumerate}[(a)]
            \item If $\kappa h_n<\kappa_0$, then:
                \begin{enumerate}[(a.1)]
                    \item[(a.1)] $\Gamma_{+, p}\left(h_n\right)=\tilde{\Gamma}_p\left(h_n\right)+o_p(1)$ with $\tilde{\Gamma}_{+, p}\left(h_n\right)=\int_0^{\infty} K(u)\allowbreak r_p(u) \allowbreak r_p(u)^{\prime} \allowbreak f\left(u h_n\right) \allowbreak  \mathrm{d} u \allowbreak \asymp \Gamma_p$
                    \item[(a.2)] $\Gamma_{-, p}\left(h_n\right)=H_p(-1) \tilde{\Gamma}_p\left(h_n\right) H_p(-1)+o_p(1)$ with $\tilde{\Gamma}_{-, p}\left(h_n\right)=\int_0^{\infty} \allowbreak K(u) \allowbreak r_p(u)\allowbreak  r_p(u)^{\prime} \allowbreak f\left(-u h_n\right)\allowbreak \mathrm{d} u \asymp \Gamma_p$,
                    \item[(a.3)] $\boldsymbol{\vartheta}_{+, p, q}\left(h_n\right)=\tilde{\boldsymbol{\vartheta}}_{+, p, q}\left(h_n\right)+o_p(1) \text { with } \tilde{\boldsymbol{\vartheta}}_{+, p, q}\left(h_n\right)=\int_0^{\infty} K(u) r_p(u) \allowbreak u^q \allowbreak f\left(u h_n\right)\allowbreak  \mathrm{d} u \allowbreak \asymp \boldsymbol{\vartheta}_{p, q}$,
                    \item[(a.4)] $\boldsymbol{\vartheta}_{-, p, q}\left(h_n\right)=(-1)^q H_p(-1) \tilde{\boldsymbol{\vartheta}}_{-, p, q}\left(h_n\right)+o_p(1) \text { with } \tilde{\boldsymbol{\vartheta}}_{-, p, q}\left(h_n\right)= \int_0^{\infty} \allowbreak K(u) \allowbreak r_p(u) \allowbreak u^q  \allowbreak f\left(-u h_n\right) \allowbreak \mathrm{d} u \allowbreak \asymp \boldsymbol{\vartheta}_{p, q}$,
                    \item[(a.5)] $h_n \Psi_{+, p}\left(h_n\right)=\tilde{\Psi}_{+, p}\left(h_n\right)+o_p(1) \text { with } \tilde{\Psi}_{+, p}\left(h_n\right)=\int_0^{\infty} \allowbreak  K(u)^2 \allowbreak r_p(u) \allowbreak r_p(u)^{\prime} \allowbreak \boldsymbol{\sigma}_{+}^2\left(u h_n\right) \allowbreak f\left(u h_n\right) \allowbreak \mathrm{d} u \asymp \Psi_p$,
                    \item[(a.6)] $h_n \Psi_{-, p}\left(h_n\right)=H_p(-1) \tilde{\Psi}_{-, p}\left(h_n\right) H_p(-1)+o_p(1) \text { with } \tilde{\Psi}_{-, p}\left(h_n\right)= \int_0^{\infty} \allowbreak K(u)^2 \allowbreak r_p(u)  \allowbreak r_p(u)^{\prime} \allowbreak \sigma_{-}^2\left(-u h_n\right) \allowbreak  f\left(-u h_n\right) \allowbreak \mathrm{d} u \asymp \Psi_p$.
                \end{enumerate}
            \item If $h_n \rightarrow 0$, then
                \begin{enumerate}[(b.1)]
                    \item $\tilde{\Gamma}_{+, p}\left(h_n\right)=f \Gamma_p+o(1)$ and $\tilde{\Gamma}_{-, p}\left(h_n\right)=f \Gamma_p+o(1)$,
                    \item $\tilde{\vartheta}_{+, p, q}\left(h_n\right)=f \vartheta_{p, q}+o(1)$ and $\tilde{\vartheta}_{-, p, q}\left(h_n\right)=f \boldsymbol{\vartheta}_{p, q}+o(1)$,
                    \item $\tilde{\Psi}_{+, p}^{+, p, q}\left(h_n\right)=\sigma_{+}^2 f \Psi_p+o(1)$ and $\tilde{\Psi}_{-, p}\left(h_n\right)=\sigma_{-}^2 f \Psi_p+o(1)$.
                \end{enumerate}
        \end{enumerate}
\end{lemma} 
\begin{proof} 
    For part (a.5), the change of variable implies
    \begin{align*}    
    \mathbb{E}\left[\Psi_{+, p}\left(h_n\right)\right] & = \mathbb{E}\left[h_n Z_p(h)' W_+(h) \Sigma_{\Delta W \Delta Y} W_+(b) Z_q(b) /n \right]\\    
    & =\int_0^{\infty} K\left(\frac{z}{h_n}\right)^2 r_p\left(\frac{z}{h_n}\right) r_p\left(\frac{z}{h_n}\right)^{\prime} \sigma_{+}^2 f(z) \mathrm{d} z \\
    & = \int_0^{\infty} K(u)^2 r_p(u) r_p(u)^{\prime} \sigma_{+}^2\left(u h_n\right) f\left(u h_n\right) \mathrm{d} u \\
    & =\tilde{\Psi}_{+, p}\left(h_n\right),    
    \end{align*}
    and $h_n^2 \mathbb{E}\left[\left|\Psi_{+, p}\left(h_n\right)-\mathbb{E}\left[\Psi_{+, p}\left(h_n\right)\right]\right|^2\right]=n^{-1} h_n^{-1} \int_0^{\infty} K(u)^4\left|r_p(u)\right|^4 f\left(u h_n\right) \mathrm{d} u=$ $O\left(n^{-1} h_n^{-1}\right)$, provided $\kappa h_n<\kappa_0$. For part (a.6),
    \begin{align*}    
    \mathbb{E}\left[h_n \Psi_{-, p}\left(h_n\right)\right] & =h_n^{-1} \int_{-\infty}^0 K\left(u / h_n\right)^2 r_p\left(u / h_n\right) r_p\left(u / h_n\right)^{\prime} \sigma_{-}^2(u) f(u) \mathrm{d} u \\
    & =H_p(-1) \tilde{\Psi}_{-, p}\left(h_n\right) H_p(-1),    
    \end{align*}
    and the rest is proven as above. Also, note that $\tilde{\Psi}_{+, p}^{+, p, q}\left(h_n\right)=\sigma_{+}^2 f \Psi_p+o(1)$ and $\tilde{\Psi}_{-, p}\left(h_n\right)=\sigma_{-}^2 f \Psi_p+o(1)$ if $h_n \rightarrow 0$, by continuity of $\sigma_{+}^2(u)$, $\sigma_{-}^2(u)$ and $f(u)$, which proves part (b.3).

    Proofs for the other items follow similarly to the one above and can be found in \cite{calonico2014robust}, as they are identical to those provided there.
\end{proof}

Let $s, \ell \in \mathbb{N}$ with $s \leq \ell$. The following lemma gives the asymptotic bias, variance, and distribution for the $\ell$ th-order local polynomial estimator of $\Delta \mu_{+}^{(s)}$ and $\Delta \mu_{-}^{(s)}:$
\begin{align*}
& \Delta \hat{\mu}_{+, \ell}^{(s)}\left(h_n\right)=s!e_s^{\prime} \delta_{+, \ell}\left(h_n\right), \\
& \Delta \hat{\delta}_{+, \ell}\left(h_n\right)=H_{\ell}\left(h_n\right) \Gamma_{+, \ell}^{-1}\left(h_n\right) Z_{\ell}\left(h_n\right)^{\prime} W_{+}\left(h_n\right) Y / n, \\
& \Delta \hat{\mu}_{-, \ell}^{(s)}\left(h_n\right)=s!e_s^{\prime} \hat{\delta}_{-, \ell}\left(h_n\right), \\
& \hat{\delta}_{-, \ell}\left(h_n\right)=H_{\ell}\left(h_n\right) \Gamma_{-, \ell}^{-1}\left(h_n\right) Z_{\ell}\left(h_n\right)^{\prime} W_{-}\left(h_n\right) Y / n .
\end{align*}

\begin{lemma}\label{lemma:s.a.3}
     Suppose Assumptions \ref{assump:continuity} and \ref{assump:disc_prob_treat} hold with $S \geq \ell+2$, and $n h_n \rightarrow \infty$.
     \begin{enumerate}
         \item[(B)] If $h_n \rightarrow 0$, then
            \begin{align*}            
            & \mathbb{E}\left[\Delta \hat{\mu}_{+, \ell}^{(s)}\left(h_n\right) \mid \mathcal{X}_n\right]= s!e_s^{\prime} \delta_{+, \ell}+h_n^{1+\ell-s} \frac{\Delta \mu_{+}^{(\ell+1)}}{(\ell+1)!} \mathcal{B}_{+, s, \ell, \ell+1}\left(h_n\right) \\
            &+h_n^{2+\ell-s} \frac{\Delta \mu_{+}^{(\ell+2)}}{(\ell+2)!} \mathcal{B}_{+, s, \ell, \ell+2}\left(h_n\right)+o_p\left(h_n^{2+\ell-s}\right), \\
            & \mathcal{B}_{+, s, \ell, r}\left(h_n\right)=s!e_s^{\prime} \Gamma_{+, \ell}^{-1}\left(h_n\right) \vartheta_{+, \ell, r}\left(h_n\right)=s!e_s^{\prime} \Gamma_{\ell}^{-1} \vartheta_{\ell, r}+o_p(1),            
            \end{align*}
            and
            \begin{align*}            
            & \mathbb{E}\left[\Delta \hat{\mu}_{-, \ell}^{(s)}\left(h_n\right) \mid \mathcal{X}_n\right]= s!e_s^{\prime} \delta_{-, \ell}+h_n^{1+\ell-s} \frac{\Delta \mu_{-}^{(\ell+1)}}{(\ell+1)!} \mathcal{B}_{-, s, \ell, \ell+1}\left(h_n\right) \\
            &+h_n^{2+\ell-s} \frac{\Delta \mu_{-}^{(\ell+2)}}{(\ell+2)!} \mathcal{B}_{-, s, \ell, \ell+2}\left(h_n\right)+o_p\left(h_n^{2+\ell-s}\right), \\
            & \mathcal{B}_{-, s, \ell, r}\left(h_n\right)=s!e_s^{\prime} \Gamma_{-, \ell}^{-1}\left(h_n\right) \vartheta_{-, \ell, r}\left(h_n\right)=(-1)^{s+r} s!e_s^{\prime} \Gamma_{\ell}^{-1} \boldsymbol{\vartheta}_{\ell, r}+o_p(1) .            
            \end{align*}
        \item[(V)] If $h_n \rightarrow 0$, then $\mathbb{V}\left[\Delta \hat{\mu}_{+, \ell}^{(s)}\left(h_n\right) \mid \mathcal{X}_n\right]=\mathcal{V}_{+, s, \ell}\left(h_n\right)$ with
            \begin{align*}            
            \mathcal{V}_{+, s, \ell}\left(h_n\right) & =\frac{1}{n h_n^{2 s}} s!^2 e_s^{\prime} \Gamma_{+, \ell}^{-1}\left(h_n\right) \Psi_{+, \ell}\left(h_n\right) \Gamma_{+, \ell}^{-1}\left(h_n\right) e_s \\
            & =\frac{1}{n h_n^{1+2 s}} \frac{\sigma_{+}^2}{f} s!^2 e_s^{\prime} \Gamma_{\ell}^{-1} \Psi_{\ell} \Gamma_{\ell}^{-1} e_s\left[1+o_p(1)\right],            
            \end{align*}
            and $\mathbb{V}\left[\Delta \hat{\mu}_{-, \ell}^{(s)}\left(h_n\right) \mid \mathcal{X}_n\right]=\mathcal{V}_{-, s, \ell}\left(h_n\right)$ with
            \begin{align*}            
            \mathcal{V}_{-, s, \ell}\left(h_n\right) & =\frac{1}{n h_n^{2 s}} s^2 e_s^{\prime} \Gamma_{-\ell}^{-1}\left(h_n\right) \Psi_{-, \ell}\left(h_n\right) \Gamma_{-, \ell}^{-1}\left(h_n\right) e_s \\
            & =\frac{1}{n h_n^{1+2 s}} \frac{\sigma_{-}^2}{f} s!^2 e_s^{\prime} \Gamma_{\ell}^{-1} \Psi_{\ell} \Gamma_{\ell}^{-1} e_s\left[1+o_p(1)\right] .            
            \end{align*}
        \item[(D)] If $n h_n^{2 \ell+5} \rightarrow 0$, then
            \begin{align*}
            \frac{\Delta \hat{\mu}_{+, \ell}^{(s)}\left(h_n\right)-\Delta \mu_{+}^{(s)}-h_n^{1+\ell-s} \frac{\Delta \mu_{+}^{(\ell+1)}}{(\ell+1)!} \mathcal{B}_{+, s, \ell, \ell+1}\left(h_n\right)}{\sqrt{\mathcal{V}_{+, s, \ell}\left(h_n\right)}} \rightarrow{ }_d \mathcal{N}(0,1)
            \end{align*}
            and
            \begin{align*}
            \frac{\Delta \hat{\mu}_{-, \ell}^{(s)}\left(h_n\right)-\Delta \mu_{-}^{(s)}-h_n^{1+\ell-s} \frac{\Delta \mu_{-}^{(\ell+1)}}{(\ell+1)!} \mathcal{B}_{-, s, \ell, \ell+1}\left(h_n\right)}{\sqrt{\mathcal{V}_{-, s, \ell}\left(h_n\right)}} \rightarrow_d \mathcal{N}(0,1) .
            \end{align*}
    \end{enumerate}
\end{lemma}
\begin{proof}
    For part (B), a Taylor series expansion yields
    \begin{align*}    
    \mathbb{E}[s! & \left.\hat{\delta}_{+, \ell}\left(h_n\right) \mid \mathcal{X}_n\right] \\
    = & s!\delta_{+, \ell}+h_n^{\ell+1} H_{\ell}\left(h_n\right) \Gamma_{+, \ell}^{-1}\left(h_n\right) Z_{\ell}\left(h_n\right) W_{+}\left(h_n\right) S_{\ell+1}\left(h_n\right) s!\frac{\Delta \mu_{+}^{(\ell+1)}}{(\ell+1)!} \\
    & +h_n^{\ell+2} H_{\ell}\left(h_n\right) \Gamma_{+, \ell}^{-1}\left(h_n\right) Z_{\ell}\left(h_n\right) W_{+}\left(h_n\right) S_{\ell+2}\left(h_n\right) s!\frac{\Delta \mu_{+}^{(\ell+2)}}{(\ell+2)!} + H_{\ell}\left(h_n\right) o_p\left(h_n^{\ell+2}\right) \\
    = & s!\delta_{+, \ell}+h_n^{\ell+1} H_{\ell}\left(h_n\right) s!\frac{\Delta \mu_{+}^{(\ell+1)}}{(\ell+1)!} \Gamma_{+, \ell}^{-1}\left(h_n\right) \vartheta_{+, \ell, \ell+1}\left(h_n\right) \\
    & +h_n^{\ell+2} H_{\ell}\left(h_n\right) s!\frac{\Delta \mu_{+}^{(\ell+2)}}{(\ell+2)!} \Gamma_{+, \ell}^{-1}\left(h_n\right) \vartheta_{+, \ell, \ell+2}\left(h_n\right) +H_{\ell}\left(h_n\right) o_p\left(h_n^{\ell+2}\right),    
    \end{align*}
    and the result for $\mathbb{E}\left[\Delta \hat{\mu}_{+, \ell}^{(s)}\left(h_n\right) \mid \mathcal{X}_n\right]$ follows by $e_s^{\prime} H_{\ell}\left(h_n\right)=h_n^{-s}$ and Lemma \ref{lemma:s.a.1}. Next, for $\mathbb{E}\left[\Delta \hat{\mu}_{-, \ell}^{(s)}\left(h_n\right) \mid \mathcal{X}_n\right]$ the same calculations apply, with only a modification for $\mathcal{B}_{-, s, \ell, r}\left(h_n\right)$ because, by Lemma \ref{lemma:s.a.1},
    \begin{align*}    
    \mathcal{B}_{-, s, \ell, r}\left(h_n\right)= & s!e_s^{\prime} \Gamma_{-, \ell}^{-1}\left(h_n\right) \vartheta_{-, \ell, r}\left(h_n\right) \\
    = & s!e_s^{\prime}\left[H_{\ell}(-1) \tilde{\Gamma}_{-, \ell}^{-1}\left(h_n\right) H_{\ell}(-1)\right]\left[(-1)^r H_{\ell}(-1) \vartheta_{-, \ell, r}\left(h_n\right)\right] \\
    & +o_p(1) \\
    = & (-1)^{s+r} s!e_s^{\prime} \tilde{\Gamma}_{-, \ell}^{-1}\left(h_n\right) \vartheta_{-, \ell, r}\left(h_n\right)+o_p(1),    
    \end{align*}
    because $e_s^{\prime} H_{\ell}(-1)=(-1)^s$ and $H_{\ell}(-1) H_{\ell}(-1)=I_{\ell+1}$.
    For part (V), simply note that
    \begin{align*}    
    \mathbb{V}\left[s!e_s^{\prime} \hat{\delta}_{+, \ell}\left(h_n\right) \mid \mathcal{X}_n\right]= & s!^2 e_s^{\prime} H_{\ell}\left(h_n\right) \Gamma_{+, \ell}^{-1}\left(h_n\right) Z_{\ell}\left(h_n\right) W_{+}\left(h_n\right) \Sigma \\
    & \times W_{+}\left(h_n\right) Z_{\ell}\left(h_n\right) \Gamma_{+, \ell}^{-1}\left(h_n\right) H_{\ell}\left(h_n\right) e_s / n \\
    = & h_n^{-2 s} s!^2 e_s^{\prime} \Gamma_{+, \ell}^{-1}\left(h_n\right) Z_{\ell}\left(h_n\right) W_{+}\left(h_n\right) \Sigma \\
    & \times W_{+}\left(h_n\right) Z_{\ell}\left(h_n\right) \Gamma_{+, \ell}^{-1}\left(h_n\right) e_s / n \\
    = & n^{-1} h_n^{-2 s} s!^2 e_s^{\prime} \Gamma_{+, \ell}^{-1}\left(h_n\right) \Psi_{+, \ell}\left(h_n\right) \Gamma_{+, \ell}^{-1}\left(h_n\right) e_s \\
    = & \mathcal{V}_{+, \ell, s}\left(h_n\right),    
    \end{align*}
    and the result follows by Lemma \ref{lemma:s.a.1}. The proof of $\mathbb{V}\left[s!e_s^{\prime} \hat{\delta}_{-, \ell}\left(h_n\right) \mid \mathcal{X}_n\right]$ is analogous.
    
    For part (D), using the previous results, we have
    \begin{align*}    
    & =\frac{s!e_s^{\prime} \hat{\boldsymbol{\delta}}_{+, \ell}\left(h_n\right)-s!e_s^{\prime} \delta_{+, \ell}-h_n^{1+\ell-s} \frac{\Delta \mu_{+}^{(\ell+1)}}{(\ell+1)!} \mathcal{B}_{+, s, \ell, \ell+1}\left(h_n\right)}{\sqrt{\mathcal{V}_{+, s, \ell}\left(h_n\right)}} 
    \end{align*}
    
    The result for $\Delta \hat{\mu}_{-, \ell}^{(s)}\left(h_n\right)$ can be established the same way. This concludes the proof.
    Q.E.D.
\end{proof}

Let $\nu$, $p$, $q$ $\in \mathbb{N}$ with $\nu \leq p < q$. The nex lemma dives the asymptotic bias, variance and distribution for the $p$th-order local polynomial estimator of $\Delta \mu_+^{(\nu)}$ and $\Delta \mu_-^{(\nu)}$ with bias correction constructed using a $q$th-order local polynomial.
\begin{lemma}\label{lemma:s.a.4}
    Suppose Assumptions \ref{assump:continuity} and \ref{assump:disc_prob_treat} hold with $S \geq q+1$, and $n \min \left\{h_n, b_n\right\} \rightarrow \infty$.
    (B) If $\max \left\{h_n, b_n\right\} \rightarrow 0$, then
    \begin{align*}    
        & \mathbb{E}\left[\Delta \hat{\mu}_{+, p, q}^{(\nu) \text { bc }}\left(h_n, b_n\right) \mid \mathcal{X}_n\right] \\
        &= \nu!e_\nu^{\prime} \delta_{+, p}+h_n^{2+p-\nu} \frac{\Delta \mu_{+}^{(p+2)}}{(p+2)!} \mathcal{B}_{+, \nu, p, p+2}\left(h_n\right)\left\{1+o_p(1)\right\} \\
        & \quad-h_n^{p+1-\nu} b_n^{q-p} \frac{\Delta \mu_{+}^{(q+1)}}{(q+1)!} \mathcal{B}_{+, p+1, q, q+1}\left(b_n\right) \frac{\mathcal{B}_{+, \nu, p, p+1}\left(h_n\right)}{(p+1)!}\left\{1+o_p(1)\right\}   
    \end{align*}
    and
    \begin{align*}    
    & \mathbb{E}\left[\Delta \hat{\mu}_{-, p, q}^{(\nu) \mathrm{bc}}\left(h_n, b_n\right) \mid \mathcal{X}_n\right] \\
    &= \nu!e_\nu^{\prime} \delta_{-, p}+h_n^{2+p-\nu} \frac{\Delta \mu_{-}^{(p+2)}}{(p+2)!} \mathcal{B}_{-, \nu, p, p+2}\left(h_n\right)\left\{1+o_p(1)\right\} \\
    & \quad-h_n^{p+1-\nu} b_n^{q-p} \frac{\Delta \mu_{-}^{(q+1)}}{(q+1)!} \mathcal{B}_{-, p+1, q, q+1}\left(b_n\right) \frac{\mathcal{B}_{-, \nu, p, p+1}\left(h_n\right)}{(p+1)!}\left\{1+o_p(1)\right\} .  
    \end{align*}
    (V) If $n \min \left\{h_n, b_n\right\} \rightarrow \infty$, then $\mathbb{V}\left[\Delta \hat{\mu}_{+, p, q}^{(\nu) \mathrm{bc}}\left(h_n, b_n\right) \mid \mathcal{X}_n\right]=\mathcal{V}_{+, v, p, q}^{\mathrm{bc}}\left(h_n, b_n\right)$, where
    \begin{align*}    
    \mathcal{V}_{+, \nu, p, q}^{\mathrm{bc}}(h, b)= & \mathcal{V}_{+, \nu, p}(h)+h_n^{2(p+1-\nu)} \mathcal{V}_{+, p+1, q}(b) \frac{\mathcal{B}_{+, \nu, p, p+1}(h)^2}{(p+1)!^2} \\
    & -2 h^{p+1-\nu} \mathcal{C}_{+, \nu, p, q}(h, b) \frac{\mathcal{B}_{+, v, p, p+1}(h)}{(p+1)!}, \\
    \mathcal{C}_{+, \nu, p, q}(h, b)= & \frac{1}{n h^\nu b^{p+1}} \nu!(p+1)!e_\nu^{\prime} \Gamma_{+, p}^{-1}(h) \Psi_{+, p, q}(h, b) \Gamma_{+, q}^{-1}(b) e_p,    
    \end{align*}
    and $\mathbb{V}\left[\Delta \hat{\mu}_{-, p, q}^{(\nu) \text { bc }}\left(h_n, b_n\right) \mid \mathcal{X}_n\right]=\mathcal{V}_{-, \nu, p, q}^{\mathrm{bc}}\left(h_n, b_n\right)$, where
    \begin{align*}    
    \mathcal{V}_{-, \nu, p, q}^{\mathrm{bc}}(h, b)= & \mathcal{V}_{-, \nu, p}(h)+h_n^{2(p+1-\nu)} \mathcal{V}_{-, p+1, q}(b) \frac{\mathcal{B}_{-, \nu, p, p+1}(h)^2}{(p+1)!^2} \\
    & -2 h^{p+1-\nu} \mathcal{C}_{-, \nu, p, q}(h, b) \frac{\mathcal{B}_{-, v, p, p+1}(h)}{(p+1)!}\\
    \mathcal{C}_{-, \nu, p, q}(h, b)=&\frac{1}{n h^{\nu} b^{p+1}} \nu!(p+1)!e_{\nu}^{\prime} \Gamma_{-, p}^{-1}(h) \Psi_{-, p, q}(h, b) \Gamma_{-, q}^{-1}(b) e_{p}    
    \end{align*}
    (D) If $n \min \left\{h_{n}^{2 p+3}, b_{n}^{2 p+3}\right\} \max \left\{h_{n}^{2}, b_{n}^{2(q-p)}\right\} \rightarrow 0$, and $\kappa \max \left\{h_{n}, b_{n}\right\}<\kappa_{0}$, then
    \begin{align*}
    \frac{\Delta \hat{\mu}_{+, p, q}^{(\nu) \mathrm{bc}}\left(h_{n}, b_{n}\right)-\nu!e_{\nu}^{\prime} \delta_{+, p}}{\sqrt{\mathcal{V}_{+, \nu, p, q}^{\mathrm{bc}}\left(h_{n}, b_{n}\right)}} \rightarrow_{d} \mathcal{N}(0,1)
    \end{align*}   
    and    
    \begin{align*}
    \frac{\Delta \hat{\mu}_{-, p, q}^{(\nu) \mathrm{bc}}\left(h_{n}, b_{n}\right)-\nu!e_{\nu}^{\prime} \delta_{-, p}}{\sqrt{\mathcal{V}_{-, \nu, p, q}^{\mathrm{bc}}\left(h_{n}, b_{n}\right)}} \rightarrow_{d} \mathcal{N}(0,1)
\end{align*}
\end{lemma}
\begin{proof}
    For part (B), first note that $\mathbb{E}\left[\Delta \hat{\mu}_{+, p, q}^{(\nu) b c}\left(h_{n}, b_{n}\right) \mid \mathcal{X}_{n}\right]=B_{1}-B_{2}$ with $B_{1}=\mathbb{E}\left[\nu!e_{\nu}^{\prime} \hat{\delta}_{+, p}\left(h_{n}\right) \mid \mathcal{X}_{n}\right]$ and $B_{2}=h_{n}^{p+1-\nu} \mathbb{E}\left[e_{p+1}^{\prime} \hat{\delta}_{+, q}\left(b_{n}\right) \mid \mathcal{X}_{n}\right] \mathcal{B}_{+, \nu, p}\left(h_{n}\right)$. By Lemma \ref{lemma:s.a.3}, with $s=\nu$ and $\ell=p$, we have

    \begin{align*}    
    B_{1}= & \nu!e_{\nu}^{\prime} \delta_{+, p}+h_{n}^{1+p-\nu} \frac{\Delta \mu_{+}^{(p+1)}}{(p+1)!} \mathcal{B}_{+, v, p, p+1}\left(h_{n}\right) \\
    & +h_{n}^{2+p-\nu} \frac{\Delta \mu_{+}^{(p+2)}}{(p+2)!} \mathcal{B}_{+, v, p, p+2}\left(h_{n}\right)+o_{p}\left(h_{n}^{2+p-\nu}\right)    
    \end{align*}    
    Similarly, by Lemma \ref{lemma:s.a.3}, with $s=p+1$ and $\ell=q$, we have
    \begin{align*}    
    & \mathbb{E}\left[(p+1)!e_{p+1}^{\prime} \hat{\delta}_{+, q}\left(b_{n}\right) \mid \mathcal{X}_{n}\right] \\
    & \quad=(p+1)!e_{p+1}^{\prime} \delta_{+, q}+b_{n}^{q-p} \frac{\Delta \mu_{+}^{(q+1)}}{(q+1)!} \mathcal{B}_{+, p+1, q, q+1}\left(b_{n}\right)+o_{p}\left(b_{n}^{q-p}\right)    
    \end{align*}    
    and hence    
    \begin{align*}    
    B_{2}= & h_{n}^{p+1-\nu} \mathbb{E}\left[(p+1)!e_{p+1}^{\prime} \hat{\delta}_{+, q}\left(b_{n}\right) \mid \mathcal{X}_{n}\right] \frac{\mathcal{B}_{+, v, p, p+1}\left(h_{n}\right)}{(p+1)!} \\
    = & h_{n}^{p+1-\nu}\left(e_{p+1}^{\prime} \delta_{+, q)}\right) \mathcal{B}_{+, \nu, p, p+1}\left(h_{n}\right) \\
    & +h_{n}^{p+1-\nu} b_{n}^{q-p} \frac{\Delta \mu_{+}^{(q+1)}}{(q+1)!} \mathcal{B}_{+, p+1, q, q+1}\left(b_{n}\right) \frac{\mathcal{B}_{+, v, p, p+1}\left(h_{n}\right)}{(p+1)!} \\
    & +h_{n}^{p+1-\nu} o_{p}\left(b_{n}^{q-p}\right) \mathcal{B}_{+, \nu, p, p+1}\left(h_{n}\right)    
    \end{align*}    
    Collecting terms, the result in part (B) follows:    
    \begin{align*}    
    \mathbb{E} & {\left[\nu!e_{\nu}^{\prime} \hat{\delta}_{+, p, q}^{\mathrm{bc}}\left(h_{n}, b_{n}\right) \mid \mathcal{X}_{n}\right] } \\
    = & \nu!e_{\nu}^{\prime} \delta_{+, p}+h_{n}^{2+p-\nu} \frac{\Delta \mu_{+}^{(p+2)}}{(p+2)!} \mathcal{B}_{+, \nu, p, p+2}\left(h_{n}\right)\left\{1+o_{p}(1)\right\} \\
    & -h_{n}^{p+1-\nu} b_{n}^{q-p} \frac{\Delta \mu_{+}^{(q+1)}}{(q+1)!} \mathcal{B}_{+, p+1, q, q+1}\left(b_{n}\right) \frac{\mathcal{B}_{+, \nu, p}\left(h_{n}\right)}{(p+1)!}\left\{1+o_{p}(1)\right\}    
    \end{align*}    
    For part $(\mathrm{V})$, first note that $\mathbb{V}\left[\Delta \hat{\mu}_{+, p, q}^{(\nu) \text { bc }}\left(h_{n}, b_{n}\right) \mid \mathcal{X}_{n}\right]=V_{1}+V_{2}-2 C_{12}$ where, using Lemma \ref{lemma:s.a.3} with $s=\nu$ and $\ell=p$,    
    \begin{align*}
    V_{1}=\mathbb{V}\left[\nu!e_{\nu}^{\prime} \hat{\delta}_{+, p}\left(h_{n}\right) \mid \mathcal{X}_{n}\right]=\mathbb{V}\left[\Delta \hat{\mu}_{+, p}^{(\nu)}\left(h_{n}\right) \mid \mathcal{X}_{n}\right]=\mathcal{V}_{+, v, p}\left(h_{n}\right)
    \end{align*}    
    and, using Lemma\ref{lemma:s.a.3} with $s=p+1$ and $\ell=q$,    
    \begin{align*}    
    V_{2} & =\mathbb{V}\left[h_{n}^{p+1-\nu}\left(e_{p+1}^{\prime} \hat{\delta}_{+, q}\left(b_{n}\right)\right) \mathcal{B}_{+, \nu, p, p+1}\left(h_{n}\right) \mid \mathcal{X}_{n}\right] \\
    & =h_{n}^{2(p+1-\nu)} \mathbb{V}\left[(p+1)!e_{p+1}^{\prime} \hat{\boldsymbol{\delta}}_{+, q}\left(b_{n}\right) \mid \mathcal{X}_{n}\right] \frac{\mathcal{B}_{+, \nu, p, p+1}\left(h_{n}\right)^{2}}{(p+1)!^{2}} \\
    & =h_{n}^{2(p+1-\nu)} \mathcal{V}_{+, p+1, q}\left(b_{n}\right) \frac{\mathcal{B}_{+, v, p, p+1}\left(h_{n}\right)^{2}}{(p+1)!^{2}}    
    \end{align*}    
    and    
    \begin{align*}    
    C_{12} & =\mathbb{C}\left[\nu!e_{\nu}^{\prime} \hat{\delta}_{+, p}\left(h_{n}\right), h_{n}^{p+1-\nu}\left(e_{p+1}^{\prime} \hat{\delta}_{+, q}\left(b_{n}\right)\right) \mathcal{B}_{+, \nu, p, p+1}\left(h_{n}\right) \mid \mathcal{X}_{n}\right] \\
    & =h_{n}^{p+1-\nu} \mathbb{C}\left[\nu!e_{\nu}^{\prime} \hat{\delta}_{+, p}\left(h_{n}\right),(p+1)!e_{p+1}^{\prime} \hat{\delta}_{+, q}\left(b_{n}\right) \mid \mathcal{X}_{n}\right] \frac{\mathcal{B}_{+, \nu, p, p+1}\left(h_{n}\right)}{(p+1)!}    
    \end{align*}    
    with    
    \begin{align*}    
    \mathbb{C} & {\left[e_{\nu}^{\prime} \hat{\delta}_{+, p}\left(h_{n}\right), e_{p+1}^{\prime} \hat{\delta}_{+, q}\left(b_{n}\right) \mid \mathcal{X}_{n}\right] } \\
    = & h_{n}^{-\nu} e_{\nu}^{\prime} \Gamma_{+, p}^{-1}\left(h_{n}\right) X_{p}\left(h_{n}\right) W_{+}\left(h_{n}\right) \mathbb{C}\left[Y, Y \mid \mathcal{X}_{n}\right] \\
    & \times W_{+}\left(b_{n}\right) X_{q}\left(b_{n}\right) \Gamma_{+, q}^{-1}\left(b_{n}\right) e_{p+1} b_{n}^{-p-1} / n^{2} \\
    = & \frac{1}{n h_{n}^{v} b_{n}^{p+1}} \nu!(p+1)!e_{\nu}^{\prime} \Gamma_{+, \ell}^{-1}\left(h_{n}\right) \Psi_{+, p, q}\left(h_{n}, b_{n}\right) \Gamma_{+, q}^{-1}\left(b_{n}\right) e_{p+1} .    
    \end{align*}    
    Thus, collecting terms, we obtain the result in part (V).    
    The proof for the control group is analogous. 
\end{proof}

\begin{lemma}\label{lemma:a.1}
    Suppose Assumptions $1-2$ hold with $S \geq p+2$, and $n h_{n} \rightarrow \infty$. Let $r \in \mathbb{N}$.
    
    (B) If $h_{n} \rightarrow 0$, then    
    \begin{align*}    
        \mathbb{E}\left[\hat{\tau}_{\nu, p}\left(h_{n}\right) \mid \mathcal{X}_{n}\right]= & \tau_{\nu}+h_{n}^{p+1-\nu} \mathbf{B}_{\nu, p, p+1}\left(h_{n}\right) \\
        & +h_{n}^{p+2-\nu} \mathbf{B}_{\nu, p, p+2}\left(h_{n}\right)+o_{p}\left(h_{n}^{p+2-\nu}\right)    
    \end{align*}    
    where    
    \begin{align*}    
    & \mathbf{B}_{\nu, p, r}\left(h_{n}\right)=\frac{\Delta \mu_{+}^{(r)}}{r!} \mathcal{B}_{+, \nu, p, r}\left(h_{n}\right)-\frac{\Delta \mu_{-}^{(r)}}{r!} \mathcal{B}_{-, \nu, p, r}\left(h_{n}\right) \\
    & \mathcal{B}_{+, v, p, r}\left(h_{n}\right)=\nu!e_{\nu}^{\prime} \Gamma_{+, p}^{-1}\left(h_{n}\right) \vartheta_{+, p, r}\left(h_{n}\right)=\nu!e_{\nu}^{\prime} \Gamma_{p}^{-1} \boldsymbol{\vartheta}_{p, r}+o_{p}(1) \\
    & \mathcal{B}_{-, \nu, p, r}\left(h_{n}\right)=\nu!e_{\nu}^{\prime} \Gamma_{-, p}^{-1}\left(h_{n}\right) \vartheta_{-, p, r}\left(h_{n}\right)=(-1)^{\nu+r} \nu!e_{\nu}^{\prime} \Gamma_{p}^{-1} \boldsymbol{\vartheta}_{p, r}+o_{p}(1)    
    \end{align*}    
    (V) If $h_{n} \rightarrow 0$, then $\mathbf{V}_{\nu, p}\left(h_{n}\right)=\mathbb{V}\left[\hat{\tau}_{\nu, p}\left(h_{n}\right) \mid \mathcal{X}_{n}\right]=\mathcal{V}_{+, v, p}\left(h_{n}\right)+\mathcal{V}_{-, \nu, p}\left(h_{n}\right)$, where    
    \begin{align*}    
    \mathcal{V}_{+, \nu, p}\left(h_{n}\right) & =\frac{1}{n h_{n}^{2 \nu}} \nu!^{2} e_{\nu}^{\prime} \Gamma_{+, p}^{-1}\left(h_{n}\right) \Psi_{+, p}\left(h_{n}\right) \Gamma_{+, p}^{-1}\left(h_{n}\right) e_{\nu} \\
    & =\frac{1}{n h_{n}^{1+2 \nu}} \frac{\sigma_{+}^{2}}{f} \nu!^{2} e_{\nu}^{\prime} \Gamma_{p}^{-1} \Psi_{p} \Gamma_{p}^{-1} e_{\nu}\left\{1+o_{p}(1)\right\} \\
    \mathcal{V}_{-, v, p}\left(h_{n}\right) & =\frac{1}{n h_{n}^{2 \nu}} \nu!^{2} e_{\nu}^{\prime} \Gamma_{-, p}^{-1}\left(h_{n}\right) \Psi_{-, p}\left(h_{n}\right) \Gamma_{-, p}^{-1}\left(h_{n}\right) e_{\nu} \\
    & =\frac{1}{n h_{n}^{1+2 \nu}} \frac{\sigma_{-}^{2}}{f} \nu!^{2} e_{\nu}^{\prime} \Gamma_{p}^{-1} \Psi_{p} \Gamma_{p}^{-1} e_{\nu}\left\{1+o_{p}(1)\right\}    
    \end{align*}    
    (D) If $n h_{n}^{2 p+5} \rightarrow 0$, then    
    \begin{align*}
    \frac{\hat{\tau}_{\nu, p}\left(h_{n}\right)-\tau_{\nu}-h_{n}^{p+1-\nu} \mathbf{B}_{\nu, p, p+1}\left(h_{n}\right)}{\sqrt{\mathbf{V}_{\nu, p}\left(h_{n}\right)}} \rightarrow_{d} \mathcal{N}(0,1)
    \end{align*}
\end{lemma}
\begin{proof}
    Part (B) follows immediately from Lemma \ref{lemma:s.a.3}(B), its analogue for the left-side estimator $\left(s!e_{s}^{\prime} \hat{\beta}_{-, \ell}\left(h_{n}\right)\right)$, and the linearity of conditional expectations. Part (V) also follows immediately from Lemma \ref{lemma:s.a.3}(V), its analogue for the left-side estimator $\left(s!e_{s}^{\prime} \hat{\beta}_{-, \ell}\left(h_{n}\right)\right)$, and the conditional independence of observations at either side of the threshold $(x=0)$. Finally, part (D) follows by the same argument given in the proof of Lemma \ref{lemma:s.a.3}(D), but now applied to the estimator $\hat{\tau}_{\nu, p}\left(h_{n}\right)=\Delta \hat{\mu}_{+, p}^{(\nu)}\left(h_{n}\right)-\Delta \hat{\mu}_{-, p}^{(\nu)}\left(h_{n}\right)=\nu!e_{\nu}^{\prime} \hat{\beta}_{+, p}\left(h_{n}\right)-\nu!e_{\nu}^{\prime} \hat{\beta}_{-, p}\left(h_{n}\right)$. This completes the proof.
\end{proof}

\newpage
\section{Proofs of main results}\label{app:proofs}
We provide the proofs of our results in this appendix.

\begin{proof}[Proof of Lemma 1]

    Given Assumptions \ref{assump:continuity}, \ref{assump:random} and \ref{assump:disc_prob_treat}, we can prove that $\tau^{DiDC} = \tau_1^{RD} - \tau_0^{RD} = \Delta Y^+ - \Delta Y^-$.
    \begin{enumerate}
        \item[A)]  \textbf{Difference of RDs.} From \ref{eq:potential_outcomes_Y1} and \ref{eq:potential_outcomes_Y0} we can derive two RDs estimands, one for $t=0$ and one for $t=1$.
            \begin{align}
                           \tau^{RD}_1    &= Y_1^{+} - Y_1^{-} \nonumber \\
                                &=\E \left[ Y_{i,1} \left( 1,1 \right) - Y_{i,1} \left( 0,0 \right) | Z_i=z_0 \right] \label{eq:tau_RDD1}
                                 \end{align}
                  \begin{align}              
                \tau^{RD}_0    &= Y_0^{+} - Y0^{-} \nonumber \\
                                &=\E\left[ Y_{i,0} \left( 1,0 \right) - Y_{i,0} \left( 0,0 \right) | Z_i = z_0 \right]  \label{eq:tau_RDD0}
            \end{align} 
            where $X^+ = \lim_{\epsilon \rightarrow 0}E(X_i|Z_i=z_0 + \epsilon)$ and $X^- = \lim_{\epsilon \rightarrow 0}E(X_i|Z_i=z_0 - \epsilon)$. By taking the difference between the equations \ref{eq:tau_RDD1} and \ref{eq:tau_RDD0} and adding  assumption \ref{assump:conf_time_invar}, we arrive at the differences-in-discontinuity estimand:
            \begin{align*}
                \tau^{DiDC}     &= \tau^{RDD}_1 - \tau^{RDD}_0 \nonumber \\
                                &=\E\left[ \left( Y_{i,1}(1,1) - Y_{i,1}(0,0) \right) | Z_i = z_0 \right] - \nonumber \\
                                & \quad\E\left[  Y_{i,0}(1,0) - Y_{i,0}(0,0) | Z_i = z_0 \right] \\
                                &=\E\left[ \left( Y_{i,1}(1,1) - Y_{i,1}(1,0) \right) | Z_i = z_0 \right] 
            \end{align*}     
        \item[B)] \textbf{RD of the differences.}   Taking the difference of the limits of the first-differenced outcomes above and below the threshold, and under assumptions \ref{assump:continuity}, \ref{assump:random} and \ref{assump:disc_prob_treat}:
            \begin{align*}
                \tau^{DiDC} &= \Delta Y^{+} - \Delta Y^{-} \nonumber \\
                            &=\E\left[ \left( Y_{i,1}(1,1) - \left( Y_{i,0}(1,0) | Z_i = z_0 \right]  - \nonumber \right. \right.\\                               
                            & \quad\E\left[ \left( Y_{i,1}(0,0) - Y_{i,0}(0,0) \right) | Z_i = z_0 \right] \\ 
                            &=\E\left[ \left( Y_{i,1}(1,1) - Y_{i,1}(0,0) \right) | Z_i = z_0 \right] - \nonumber \\
                            & \quad\E\left[  Y_{i,0}(1,0) - Y_{i,0}(0,0) | Z_i = z_0 \right] 
            \end{align*}
            where $\Delta X^+ = \lim_{\epsilon \rightarrow 0}E(\Delta X_i|Z_i=z_0 + \epsilon)$ and $\Delta X^- = \lim_{\epsilon \rightarrow 0}E(\Delta X_i|Z_i=z_0 - \epsilon)$.
    \end{enumerate}    
    Therefore $\tau^{DiDC}     = \tau^{RD}_1 - \tau^{RD}_0   = \Delta Y^{+} - \Delta Y^{-}$.
\end{proof}

\begin{proof}[Proof of Lemma \ref{lemma:mult}]
   Using the derivations in the proof of Lemma 1:
    \begin{align*}
                \tau^{DiDC} &=\E\left[ \left( Y_{i,1}(1,1) - Y_{i,1} (0,0) \right) | Z_i = z_0 \right] - \nonumber \\
                            & \quad\E\left[  Y_{i,0}(1,0) - Y_{i,0}(0,0) | Z_i = z_0 \right] 
            \end{align*}
        Under Assumption \ref{assump:mult}, 
        \begin{align*}
  \tau^{DiDC} &=\E\left[ \left( Y_{i,1}(0,1) - Y_{i,1} (0,0) \right) | Z_i = z_0 \right] \E\left[ \left( Y_{i,1}(1,0) - Y_{i,1} (0,0) \right) | Z_i = z_0 \right]- \nonumber \\
                            & \quad\E\left[  Y_{i,0}(1,0) - Y_{i,0}(0,0) | Z_i = z_0 \right]
\end{align*}
        
    By Assumption \ref{assump:conf_time_invar}, the second factor is time-invariant and can be replaced with its period-0 analogue, which is exactly $\tau^{RD}_0$. Substituting this expression gives
    \[
        \tau^{DiDC} = (\E\!\left[ Y_{i,1}(0,1) - Y_{i,1}(0,0)\mid Z_i=z_0 \right] +1) \tau^{RD}_0 
    \]
    Rearranging yields the desired result,
    \[
        \frac{\tau^{DiDC}}{\tau^{RD}_0} +1 = \E\!\left[ Y_{i,1}(0,1) - Y_{i,1}(0,0)\mid Z_i=z_0 \right].
    \]
\end{proof}

\begin{proof}[Proof of Lemma \ref{lemma:MSEoptimal_band}]

    Recall the definition $MSE_{\nu,p,s}(h_n) =\E\left[ \left( \hat{\tau}_{\nu, p, s}(h_n) - \tau_{\nu,p} \right)^2 | \chi_n \right]$. We can rewrite as 
    \begin{align*}
        MSE_{\nu,p,s}(h_n) =& V \left[ \hat{\tau}_{\nu, p, s}(h_n) | \chi_n \right] +  \left(\E \left[ \hat{\tau}_{\nu, p, s}(h_n) - \tau_{\nu,p} | \chi_n \right] \right)^2 
    \end{align*}
    Then, from Lemma \ref{lemma:a.1},
    \begin{align*}
        V \left[ \hat{\tau}_{\nu, p, s}(h_n) | \chi_n \right]  & = \frac{\nu! }{n h_n^{1+2\nu}} \frac{\sigma_+^2 - \sigma_-^2}{f}e^{\prime}_{\nu}\Gamma_p^{-1} \Psi_p \Gamma_p^{-1} \{1+o_p(1)\}
    \end{align*}
    and
    \begin{align}
        \E \left[ \hat{\tau}_{\nu, p, s}(h_n) - \tau_{\nu,p} | \chi_n \right] = h_n^{1+p-\nu} \frac{\Delta \mu^{(p+1)}_{+}-(-1)^{\nu+p+s}\Delta \mu^{(p+1)}_{-}}{(p+1)!} \nu! e^{\prime}_{\nu}\Gamma_p^{-1} \varphi_{p,p+1} \{1+o_p(1)\}
    \end{align}
    so we can rewrite 
    \begin{align*}
        MSE_{\nu,p,s}(h_n) = \frac{\mathbb{V_{\nu,p}}}{n h_n^{1+2\nu}}\{1+o_p(1)\} + \left( h_n^{1+p-\nu} \mathbb{B}_{\nu, p, p+1,s}\{1+o_p(1)\} \right)^2
    \end{align*}
    where $\mathbb{V_{\nu,p}} =\nu!\frac{\sigma_+^2 - \sigma_-^2}{f}e^{\prime}_{\nu}\Gamma_p^{-1} \Psi_p \Gamma_p^{-1}  $ and $\mathbb{B_{\nu, p, p+1, s}} = \frac{\Delta \mu^{(p+1)}_{+}-(-1)^{\nu+p+s}\Delta \mu^{(p+1)}_{-}}{(p+1)!} \nu! e^{\prime}_{\nu}\Gamma_p^{-1} \varphi_{p,p+1}$.
    The MSE-optimal bandwidth is
    \begin{align*}
         h_{n,\nu, p}^{MSE} & = arg \min_{h_{n}} MSE_{\nu,p,s}(h_n)\\
         & = \left( \frac{\left( 1+ 2 \nu \right) \mathbf{V}_{\nu ,p}}{2n \left( 1+p-\nu \right) \textbf{B}^{2}_{\nu, p, p+1, s}} \right)^{\frac{1}{2p+3}}
    \end{align*}
\end{proof}

\begin{proof}[Proof of Lemma 4]

    We define higher-order derivatives notation of the unknown regression functions as:
    \begin{align*}
        &\Delta \mu_+^{(\nu)}(z) = \frac{d^{\nu}\Delta\mu_+(z)}{dz^{\nu}}, &\Delta \mu_-^{(\nu)}(z) = \frac{d^{\nu}\Delta\mu_(z)}{dz^{\nu}}\\
        &\Delta \mu_+^{(\nu)} = \lim_{x \rightarrow 0^+}{\Delta \mu_+^{(\nu)}(z)}, & \Delta \mu_-^{(\nu)} = \lim_{x \rightarrow 0^-}{\Delta \mu_-^{(\nu)}(z)}\\
        &\Delta \mu_+ = \lim_{x \rightarrow 0^+}{\Delta \mu(z)}, & \Delta \mu_- = \lim_{x \rightarrow 0^=}{\Delta \mu(z)}
    \end{align*}
    where $\Delta \mu (z) =\E\left( \Delta Y_i | Z_i = z \right) =\E\left( Y_{i,1} - Y_{i,0} | Z_i = z \right)$.

    The notation for the RD at time $t$ where $ \mu_t (z) =\E\left(  Y_{i,t} | Z_i = z \right)$ is
    \begin{align*}
        &\mu_{t,+}^{(\nu)}(z) = \frac{d^{\nu}\mu_{t,+}(z)}{dz^{\nu}}, & \mu_{t,-}^{(\nu)}(z) = \frac{d^{\nu}\mu_(z)}{dz^{\nu}}\\
        & \mu_{t,+}^{(\nu)} = \lim_{x \rightarrow 0^+}{ \mu_{t,+}^{(\nu)}(z)}, &  \mu_{t,-}^{(\nu)} = \lim_{x \rightarrow 0^-}{ \mu_{t,-}^{(\nu)}(z)}\\
        & \mu_{t,+} = \lim_{x \rightarrow 0^+}{ \mu(z)}, &  \mu_{t,-} = \lim_{x \rightarrow 0^=}{ \mu(z)}
    \end{align*}    
    Then we derive the following:
    \begin{align}
        &\Delta \mu_+^{(\nu)} = \mu_{1,+}^{(\nu)} - \mu_{0,+}^{(\nu)}  & \Delta \mu_-^{(\nu)} = \mu_{1,}^{(\nu)} - \mu_{0,+}^{(\nu)}
    \end{align}
    It follows that $B \left[ \hat{\tau}^{DiDC}_{\nu,p} (h_n) \right] = B \left[ \hat{\tau}^{RD}_{1,\nu,p} (h_n) \right] - B \left[ \hat{\tau}^{RD}_{0,\nu,p} (h_n) \right]  $

    Note that

\begin{align*}
        \Delta \mu(z)   =&\E\left[ \Delta Y_i | Z_i = z \right] \\
                        =&\E\left[ Y_{i,1} - Y_{i,0} | Z_i = z \right] \\
                        =&\E\left[ Y_{i,1} | Z_i = z \right] -\E\left[ Y_{i,0} | Z_i = z \right]\\
                        =& \mu_1(z) -\mu_0(z)
    \end{align*}
    Therefore,
    \begin{align*}
        \Delta \mu_+^{(\nu)} =& \mu_{+,1}^{(\nu)} - \mu_{+,0}^{(\nu)}
    \end{align*}
    Then, if the shapes of the DGPs are time-invariant for each side of the threshold, then $\mu_{+,1}^{(\nu)} - \mu_{+,0}^{(\nu)}$ and $\mu_{-,1}^{(\nu)} - \mu_{-,0}^{(\nu)}$, leading to $\Delta \mu_+^{(\nu)} = \Delta \mu_-^{(\nu)} = 0$
\end{proof}

\begin{proof}[Proof of Theorem 1]

The result follows from Lemma B.3 by considering estimators from both sides of the threshold together and the conditional independence of observations on either side.
    
\end{proof}

\begin{proof}[Proof of Theorem 2]

Let $\theta_{0}=\left(\Delta\mu_{+},\Delta\mu_{-}\right)$ and $\widehat{\theta}=\left(\widehat{\mu}_{+,p,q}^{(0)}(h_{n},b_{n}),\widehat{\mu}_{-,p,q}^{(0)}(h_{n},b_{n})\right)$. For any fixed $y$, $c_{1}$ and $c_{2}$, define the mapping

\begin{equation*}
    \phi :l^{\infty}\left ( \mathbb{R}\times \mathcal{S}(Z) \right )\times l^{\infty}\left ( \mathcal{S}(Z) \right )\rightarrow l^{\infty}\left (\mathcal{S}(Z),\mathbb{R}^{2} \right )
\end{equation*}

by 

{ \begin{equation*}
    \left [ \phi(\theta) \right ](z)=\begin{pmatrix}
\max\left\{ \theta^{(1)}(y,z)-c_{1},(\theta^{(1)}(y,z)-\theta^{(2)}(y,z))-c_{2}\right\} \\
\min\left\{ \theta^{(1)}(y,z)+c_{1},(\theta^{(1)}(y,z)-\theta^{(2)}(y,z))+c_{2} \right\}\end{pmatrix}
\end{equation*}}

where $\theta^{(j)}$ is the j-th component of $\theta$. From Fang and Santos (2015), the maps $\left ( a_{1},a_{2} \right )\mapsto \min\left\{ a_{1},a_{2}\right\}$ and $\left ( a_{1},a_{2} \right )\mapsto \max\left\{ a_{1},a_{2}\right\}$ are Hadamard directionally
differentiable with Hadamard directional derivatives at $(a_{1},a_{2})$ equal to

\begin{equation*}
    h\mapsto \left\{\begin{matrix}
h^{(1)},if\  a_{1}<a_{2}, \\
\min\left\{h^{(1)},h^{(2)} \right\},if\ a_{1}=a_{2}, \\
h^{(2)},if\ a_{1}>a_{2}\end{matrix}\right.
\end{equation*}

and

\begin{equation*}
    h\mapsto \left\{\begin{matrix}
h^{(2)},if\  a_{1}<a_{2}, \\
\max\left\{h^{(1)},h^{(2)} \right\},if\ a_{1}=a_{2}, \\
h^{(1)},if\ a_{1}>a_{2}\end{matrix}\right.
\end{equation*}

for $h\in\mathbb{R}^{2}$.

The mapping comprises max and min operators, along with four Hadamard differentiable functions. We compute their Hadamard derivatives with respect to $\theta$ below:

$\left [ \delta_{1}(\theta) \right ](z)=\theta^{(1)}(y,z)-c_{1}$ has Hadamard derivative equal to

$\left [ \delta^{'}_{1,\theta}(h) \right ](z)=h^{(1)}(y,z)$.

$\left [ \delta_{2}(\theta) \right ](z)=(\theta^{(1)}(y,z)-\theta^{(2)}(y,z))-c_{2}$ has Hadamard derivative equal to

$\left [ \delta^{'}_{2,\theta}(h) \right ](z)=h^{(1)}(y,z)-h^{(2)}(y,z)$.

$\left [ \delta_{3}(\theta) \right ](z)=\theta^{(1)}(y,z)+c_{1}$ has Hadamard derivative equal to

$\left [ \delta^{'}_{3,\theta}(h) \right ](z)=h^{(1)}(y,z)$.

$\left [ \delta_{4}(\theta) \right ](z)=(\theta^{(1)}(y,z)-\theta^{(2)}(y,z))+c_{2}$ has Hadamard derivative equal to

$\left [ \delta^{'}_{4,\theta}(h) \right ](z)=h^{(1)}(y,z)-h^{(2)}(y,z)$.

Using this notation, we write the functional $\phi$ as 

\begin{equation*}
    \phi(\theta)=\begin{pmatrix}
\max\left\{ \delta_{1}(\theta),\delta_{2}(\theta)\right\} \\
\min\left\{ \delta_{3}(\theta),\delta_{4}(\theta) \right\}\end{pmatrix}
\end{equation*}

By the Chain rule for Hadamard differentiable functions (Lemma C2 from \citealp{masten2020inference}), $\phi$ has Hadamard directional derivative equal to

\begin{equation*}
    \phi_{\theta}^{'}(h)=\begin{pmatrix}
    \mathbf{1}\left\{ \delta_{1}(\theta)>\delta_{2}(\theta)\right\}\delta_{1,\theta}^{'}(h) \\
    +\mathbf{1}\left\{ \delta_{1}(\theta)=\delta_{2}(\theta) \right\}\max\left\{ \delta_{1,\theta}^{'}(h),\delta_{2,\theta}^{'}(h)\right\} \\+\mathbf{1}\left\{ \delta_{1}(\theta)<\delta_{2}(\theta)\right\}\delta_{2,\theta}^{'}(h)
     \\
     \\\mathbf{1}\left\{ \delta_{3}(\theta)<\delta_{4}(\theta)\right\}\delta_{4,\theta}^{'}(h)
     \\+\mathbf{1}\left\{ \delta_{3}(\theta)=\delta_{4}(\theta)\right\}\min\left\{ \delta_{3,\theta}^{'}(h),\delta_{4,\theta}^{'}(h)\right\}
     \\+\mathbf{1}\left\{ \delta_{3}(\theta)>\delta_{4}(\theta)\right\}\delta_{3,\theta}^{'}(h)
\end{pmatrix}
\end{equation*}

By theorem B.3, we have $\sqrt{n\min\left\{ h_{n},b_{n}\right\}}(\widehat{\theta}(h_{n},b_{n})-\theta_{0})\rightarrow \textbf{Z}(h_{n},b_{n})$, where $\textbf{Z}(h_{n},b_{n})$ is mean-zero gaussian process with variance given by theorem B.3 as well.

Hence, we can use the delta method for Hadamard directionally differentiable functions \cite{fang2018inference} to find that

\begin{equation*}
    \sqrt{n\min\left\{ h_{n},b_{n}\right\}}\left [ (\phi(\widehat{\theta}(h_{n},b_{n})-\phi(\theta_{0})) \right ](z)\rightarrow \left [\phi_{\theta_{0}}^{'}(\textbf{Z)}  \right ](z)
\end{equation*}

which concludes the proof.
    
\end{proof}

\begin{proof}[Proof of Lemma 6]

We begin by deriving the sharp identified set for $\tau_{c}$. From Assumption 7, we have

\begin{equation*}
    Y_{i,1}(1,0)+Y_{i,1}(0,1)\leq Y_{i,1}(1,1)+Y_{i,1}(0,0)
\end{equation*}

which can be rearranged into

\begin{equation*}
    Y_{i,1}(0,1)-Y_{i,1}(0,0)\leq Y_{i,1}(1,1)-Y_{i,1}(1,0)
\end{equation*}

Taking expectations conditional on $Z_{i}=z_{0}$ yields

\begin{equation*}
   \E \left[Y_{i,1}(0,1)-Y_{i,1}(0,0)|Z_{i}=z_{0}\right]\leq\E \left[Y_{i,1}(1,1)-Y_{i,1}(1,0)|Z_{i}=z_{0}\right]=\tau_{c}
\end{equation*}

Under Assumptions 1 and 2, these inequality equals

\begin{equation*}
   \E \left[Y_{i,1}(0,1)|Z_{i}=z_{0}\right]-Y_{1}^{-}\leq\tau_{c}
\end{equation*}

From assumption 6, it follows that

\begin{equation*}
    y^{min}-Y_{1}^{-}\leq\tau_{c}
\end{equation*} 

which yields the lower bound for $\tau_{c}$. Note, however, that Assumption 7 provides no improvement on the worst-case upper bound.

We now turn to $\tau_{uc}$. Rearranging Assumption 7 and taking expectations conditional on $Z_{i}=z_{0}$ yields

\begin{equation*}
    \tau_{uc}= \E \left[Y_{i,1}(0,1)-Y_{i,1}(0,0)|Z_{i}=z_{0}\right]\leq\E \left[Y_{i,1}(1,1)-Y_{i,1}(1,0)|Z_{i}=z_{0}\right]
\end{equation*}

Under Assumptions 1 and 2, the inequality is equal to 

\begin{equation*}
    \tau_{uc}\leq Y_{1}^{+}-E\left[Y_{i,1}(1,0)|Z_{i}=z_{0}\right]
\end{equation*}

from which it follows, under Assumption 6, that

\begin{equation*}
    \tau_{uc}\leq Y_{1}^{+}-y^{min}
\end{equation*}

which yields the upper bound for $\tau_{uc}$. Note, however, that Assumption 7 provides no improvement on the worst-case lower bound.
    
\end{proof}

\begin{proof}[Proof of Lemma 7]

We begin by deriving the sharp identified set for $\tau_{c}$. From Assumption 7, we have

\begin{equation*}
    Y_{i,1}(1,0)+Y_{i,1}(0,1)\geq Y_{i,1}(1,1)-Y_{i,1}(0,0)
\end{equation*}

which can be rearranged into

\begin{equation*}
    Y_{i,1}(0,1)-Y_{i,1}(0,0)\geq Y_{i,1}(1,1)-Y_{i,1}(1,0)
\end{equation*}

Taking expectations conditional on $Z_{i}=z_{0}$ yields

\begin{equation*}
   \E \left[Y_{i,1}(0,1)-Y_{i,1}(0,0)|Z_{i}=z_{0}\right]\geq\E \left[Y_{i,1}(1,1)-Y_{i,1}(1,0)|Z_{i}=z_{0}\right]=\tau_{c}
\end{equation*}

Under Assumptions 1 and 2, the inequality equals 

\begin{equation*}
   \E \left[Y_{i,1}(0,1)|Z_{i}=z_{0}\right]-Y_{1}^{-}\geq \tau_{c}
\end{equation*}

From Assumption 6, it follows that.

\begin{equation*}
    y^{max}-Y_{1}^{-}\geq \tau_{c}
\end{equation*}

which yields the upper bound for $\tau_{c}$. Note, however, that there is no improvement on the worst-case lower bound.

We now turn to $\tau_{uc}$. Rearranging Assumption 7 and taking expectations conditional on $Z_{i}=z_{0}$ yields

\begin{equation*}
    \tau_{uc}= \E \left[Y_{i,1}(0,1)-Y_{i,1}(0,0)|Z_{i}=z_{0}\right]\geq\E \left[Y_{i,1}(1,1)-Y_{i,1}(1,0)|Z_{i}=z_{0}\right]
\end{equation*}

Under Assumptions 1 and 2, the inequality is equal to

\begin{equation*}
    \tau_{uc}\geq Y_{1}^{+}-E\left[Y_{i,1}(1,0)|Z_{i}=z_{0}\right]
\end{equation*}

and from Assumption 6, it follows that

\begin{equation*}
    \tau_{uc}\geq Y_{1}^{+}-y^{max}
\end{equation*}

which yields the lower bound for $\tau_{uc}$. Note, however, that there is no improvement on the worst-case upper bound.

\end{proof}

\begin{proof}[Proof of Lemma 8]

We begin with $\tau_{c}$. The inequalities characterizing the bounds in Lemmas 5 and 6 remain valid, which yields the result.

For the case of $\tau_{uc}$, the bounds from Lemma 6 remain valid. Assumption 5 combined with 6 and 7, further implies that

\begin{equation*}
    Y_{i,1}(0,1)-Y_{i,1}(0,0)\leq Y_{i,1}(1,1) -(Y_{i,1}(1,0)-c_{1})
\end{equation*}

From the inequality above and Assumptions 1-2, it follows that.

\begin{equation*}
    \tau_{uc}\leq \Delta Y^{+}+c_{1}
\end{equation*}

which concludes the proof.
    
\end{proof}

\begin{proof}[Proof of Lemma 9]

The bounds for $\tau_{c}$ are obtained by combining the bounds from Lemmas 5 and 7.

For the case of $\tau_{uc}$, the bounds from Lemma 7 remain valid. Assumption 5, combined with 6 and 8, further implies that

\begin{equation*}
    Y_{i,1}(0,1)-Y_{i,1}(0,0)\geq Y_{i,1}(1,1)-(Y_{i,0}(1,0)+c_{1})
\end{equation*}

From the inequality above and Assumptions 1-2, it follows that.

\begin{equation*}
    \tau_{uc}\geq \Delta Y^{+}-c_{1}
\end{equation*}

which concludes the proof.
    
\end{proof}

\newpage
\section{Simulations} \label{app:Simulations}

\subsection{Bandwidth for Testing if Confounding Effects are Constant} \label{app:Bandwidth_Test}
This appendix details the critical issue of bandwidth selection, the Wald test to test Assumption\ref{assump:conf_time_invar}. As mentioned before, this assumption is vital for the difference-in-discontinuity method, as the non-validity means all estimations would be invalid. The procedure for testing involves considering the following stacked RDs regression model:
\begin{align}
    Y_i = \alpha_i + \sum_{k=1}^K \left[ \beta_{-k} \left(Z_i - z_0 \right) + \theta_{-k}D \left(Z_i - z_0 \right) \right]T_{-k} + \varepsilon_i \tag{\ref{eq:reg_test_conf_const}}
\end{align}
where $T_{-k}$ is a dummy indicating which period the data belongs to, that is, the RD of which period it's running. The null hypothesis for the Wald test is:
\begin{align*}
    H_0: \quad \theta_0 = \theta_{-1} = ... = \theta_{-K}
\end{align*}

One central question arises: how should we determine the optimal bandwidth for Equation \ref{eq:reg_test_conf_const}. In this section, we outline the methodology for bandwidth selection and present simulation results for two scenarios: one with identical Data Generating Processes for the pre- and post-treatment periods and another with different DGPs, both while accounting for confounding factors. The first scenario aligns with the framework presented in Model 1, as detailed in Appendix \ref{app:identical_dgp_confounder}, while the second encompasses both models 1 and 3 (Appendix \ref{app:identical_dgp_confounder}).

To select bandwidths, we employ the \verb|rdrobust| package and estimate bandwidths separately for each RD. Next, we cross-validate the bandwidths and conduct RD regressions for each, evaluating the bias and root-mean-squared errors (RMSE) of the estimated parameters. The chosen bandwidth is the one that minimizes RMSE, along with the largest and smallest bandwidths, and the CCT \citep{calonico2014robust} bandwidth for all observations, regardless of time dummies.

The various bandwidths selected produced similar results. This finding suggests that the choice of bandwidth does not substantially affect the accuracy of the test. It is essential to note that while this result says this, the choice of bandwidth for the DiDC estimation of the treatment effect remains a crucial consideration of the method.

\subsection{Data Generating Processes}

We provide further details on each of the data-generating processes (DGPs) employed in our simulation studies. Both GDPs employ the same simulation setup based on model 3 from \cite{calonico2014robust}, with the key difference being the presence of time-fixed confounding factors at the threshold in one of the settings. We perform 1000 replications for each simulation, and for each replication, we use a sample size of $n=1000$, with
\begin{align*}
    Y_{i,t}&=\mu_{it}(Z_{i})+\varepsilon_{i,t}, \quad Z_{i}\sim(2\mathcal{B}(2,4)-1), \\
    & \quad \varepsilon_{it}\sim N(0,\sigma^{2}_{\varepsilon})
\end{align*}
where $\mathcal{B}(p_1,p_2)$ is a beta distribution with parameters $p_1$ and $p_2$, $\sigma_{\varepsilon}=0.1295$. 

\subsubsection{Model 1: Identical DGPs with time-invariant confounders at the threshold} \label{app:identical_dgp_confounder}

 This model includes a time-invariant confounder, $\textbf{c}$, to represent pre-existing discontinuities at the threshold that affect the outcome $Y_i$. Here, units with $Z_i\geq0$ receive a treatment $\textbf{c}$ and units with $Z_i<0$ do not. At the same time,  units with $Z_i\geq0$ have a different data-generating function than units with $Z_i<0$. The data generating functions $\mu_0(z)$ and $\mu_1(z)$, for periods $t=0$ and $t=1$ respectively are as follows
\begin{align*}
    \mu_{0}(z)&=\begin{cases} 
            0.48+1.27z-0.5\cdot7.18z^{2} \\+ 0.7\cdot20.21z^{3}+1.1\cdot21.54z^{4}+1.5\cdot7.33z^{5} & \text{if } z<0\\ 
            0.52+0.84z-0.1\cdot3z^{2} \\ - 0.3\cdot7.99z^{3}-0.1\cdot9.01z^{4}+3.56z^{5}+\textbf{c} & \text{if } z\geq 0
        \end{cases} \\
    \mu_{1}(z)&=\begin{cases} 
            0.48+1.27z-0.5\cdot7.18z^{2} \\ + 0.7\cdot20.21z^{3}+1.1\cdot21.54z^{4}+1.5\cdot7.33z^{5} & \text{if } z<0\\ 
            0.52+0.84z-0.1\cdot3z^{2} \\  - 0.3\cdot7.99z^{3}-0.1\cdot9.01z^{4}+3.56z^{5}+\textbf{c}+\tau & \text{if } z\geq 0
        \end{cases}
\end{align*}

In addition to Table 1, we present results from Monte Carlo simulations in which the variance changes over time.

Table D.1 shows the results for simulations in a DGP in which the standard deviation of the mean potential outcomes is greater in period 1 than in period 0. We set $\sigma_{\varepsilon}=0.1295$ in period 0 and $\sigma_{\varepsilon}=0.259$ in period 1.

\begin{table}[]
\centering
\caption{Identical Functional Forms with Confounding (More variance in $t=1$)}
\begin{tabular}{c|cccccc}
\hline
             & AV.Bias & RMSE  & Coverage & CIL   & \multicolumn{2}{c}{Bandwidths} \\
             &         &       &          &       & $h_{n}$             & $b_{n}$            \\ \hline
RDD (Robust) & 1.044   & 1.044 & 0        & 0.331 & 0.199          & 0.331         \\
Diff-in-Disc & -0.003  & 0.074 & 0.941    & 0.360 & 0.221          & 0.358         \\
NP-DiD       & 0.001   & 0.029 & 0.956    & 0.314 & -              & -             \\
TWFE         & 0.001   & 0.019 & 1        & 0.403 & -              & -             \\ \hline
\end{tabular}\\\scriptsize \noindent \textit{Note:} Simulations based on 10,000 Monte Carlo experiments with a sample size $n=1000$. RDD is the non-bias corrected RD estimator from CCT (2014), NP-DiD is the outcome regression DiD estimator from Sant'Anna and Zhao (2020), DiDC is the estimator proposed in this paper.  “Av. Bias”, “RMSE”, “Cover”, “CIL’ stand for the average simulated bias, simulated root mean-squared errors, 95\% coverage probability and 95\% confidence interval length, respectively. "Bandwidths" $h_n$ and $b_n$ report the plug-in bandwidths for point and bias estimation, respectively.
\end{table}

In Table D.2, we invert and set $\sigma_{\varepsilon}=0.259$ in period 0 and $\sigma_{\varepsilon}=0.1295$ in period 1. Overall, the differences in variance over time do not affect the qualitative takeaways from the discussion in Section 6.

\begin{table}[]
\centering
\caption{Identical Functional Forms with Confounding (More variance in $t=0$)}
\begin{tabular}{c|cccccc}
\hline
             & AV.Bias & RMSE  & Coverage & CIL   & \multicolumn{2}{c}{Bandwidths} \\
             &         &       &          &       & $h_{n}$             & $b_{n}$            \\ \hline
RDD (Robust) & 1.043   & 1.043 & 0        & 0.175 & 0.169          & 0.318         \\
Diff-in-Disc & -0.004  & 0.075 & 0.929    & 0.360 & 0.223          & 0.360         \\
NP-DiD       & 0.001   & 0.029 & 0.950    & 0.343 & -              & -             \\
TWFE         & 0.001   & 0.018 & 1        & 0.404 & -              & -             \\ \hline
\end{tabular}\\\scriptsize \noindent \textit{Note:} Simulations based on 10,000 Monte Carlo experiments with a sample size $n=1000$. RDD is the non-bias corrected RD estimator from CCT (2014), NP-DiD is the outcome regression DiD estimator from Sant'Anna and Zhao (2020), DiDC is the estimator proposed in this paper.  “Av. Bias”, “RMSE”, “Cover”, “CIL’ stand for the average simulated bias, simulated root mean-squared errors, 95\% coverage probability and 95\% confidence interval length, respectively. "Bandwidths" $h_n$ and $b_n$ report the plug-in bandwidths for point and bias estimation, respectively.
\end{table}

\subsubsection{Model 2: Identical DGPs}

 This model follows model 3 from \cite{calonico2014robust} exactly for both periods, the only difference being that in period $t=1$ there is an effect of the treatment $\tau$ for units with $Z_i \geq 0$. No other discontinuities that could affect the outcome exist at the threshold $z_0=0$. The data generating functions $\mu_0(z)$ and $\mu_1(z)$ are as follows
\begin{align*}
    \mu_{0}(z)&=\begin{cases} 
            0.48+1.27z-0.5\cdot7.18z^{2} \\ + 0.7\cdot20.21z^{3}+1.1\cdot21.54z^{4}+1.5\cdot7.33z^{5} & \text{if } z<0\\ 
            0.52+0.84z-0.1\cdot3z^{2} \\  - 0.3\cdot7.99z^{3}-0.1\cdot9.01z^{4}+3.56z^{5} & \text{if } z\geq 0
        \end{cases} \\
    \mu_{1}(z)&=\begin{cases} 
            0.48+1.27z-0.5\cdot7.18z^{2} \\  + 0.7\cdot20.21z^{3}+1.1\cdot21.54z^{4}+1.5\cdot7.33z^{5} & \text{if } z<0\\ 
            0.52+0.84z-0.1\cdot3z^{2} \\ - 0.3\cdot7.99z^{3}-0.1\cdot9.01z^{4}+3.56z^{5}+\tau & \text{if } z\geq 0
        \end{cases}
\end{align*}

In addition to Table 2, we present the results for simulations under a DGP in which the functional forms are as displayed above, but the variance changes over time, repeating the exercise from Appendix D.2.1.

\begin{table}[h]
\centering
\caption{Identical Functional Forms without Confounding (More variance in $t=1$)}
\begin{tabular}{c|cccccc}
\hline
             & AV.Bias & RMSE  & Coverage & CIL   & \multicolumn{2}{c}{Bandwidths} \\
             &         &       &          &       & $h_{n}$             & $b_{n}$            \\ \hline
RDD (Robust) & 0.044   & 0.079 & 0.922    & 0.331 & 0.199          & 0.344         \\
Diff-in-Disc & -0.003  & 0.074 & 0.941    & 0.360 & 0.221          & 0.358         \\
NP-DiD       & 0.001   & 0.029 & 0.956    & 0.343 & -              & -             \\
TWFE         & 0.001   & 0.019 & 1        & 0.403 & -              & -             \\ \hline
\end{tabular}\\\scriptsize \noindent \textit{Note:} Simulations based on 10,000 Monte Carlo experiments with a sample size $n=1000$. RDD is the non-bias corrected RD estimator from CCT (2014), NP-DiD is the outcome regression DiD estimator from Sant'Anna and Zhao (2020), DiDC is the estimator proposed in this paper.  “Av. Bias”, “RMSE”, “Cover”, “CIL’ stand for the average simulated bias, simulated root mean-squared errors, 95\% coverage probability and 95\% confidence interval length, respectively. "Bandwidths" $h_n$ and $b_n$ report the plug-in bandwidths for point and bias estimation, respectively.
\end{table}

The results are displayed in Tables D.3 and D.4. Once again, the changes in standard deviations of potential outcomes over time do not affect the qualitative takeaways from Section 6.

\subsubsection{Model 3: Different DGPs with time-invariant confounders at the threshold} \label{app:different_dgp_confounder}

This model follows model 3 from \cite{calonico2014robust} and includes a time-invariant confounder, $\textbf{c}$, to represent pre-existing discontinuities at the threshold that affect the outcome $Y_i$ at $t=0$, and a linear model for time $t=1$. The data generating functions $\mu_0(z)$ and $\mu_1(z)$, for periods $t=0$ and $t=1$ respectively are as follows
\begin{align*}
    \mu_{0}(z)&=\begin{cases} 
            0.48+1.27z-0.5\cdot7.18z^{2} \\ + 0.7\cdot20.21z^{3}+1.1\cdot21.54z^{4}+1.5\cdot7.33z^{5} & \text{if } z<0\\ 
            0.52+0.84z-0.1\cdot3z^{2} \\ - 0.3\cdot7.99z^{3}-0.1\cdot9.01z^{4}+3.56z^{5}+\textbf{c} & \text{if } z\geq 0
        \end{cases} \\
    \mu_{1}(z)&=\begin{cases} 
            0.48+1.4z & \text{if } z<0\\ 
            0.52+0.1z +\textbf{c}+\tau & \text{if } z\geq 0
        \end{cases}
\end{align*}

\begin{table}[h!]
\centering
\caption{Identical Functional Forms without Confounding (More variance in $t=0$)}
\begin{tabular}{c|cccccc}
\hline
             & AV.Bias & RMSE  & Coverage & CIL   & \multicolumn{2}{c}{Bandwidths} \\
             &         &       &          &       & $h_{n}$             & $b_{n}$            \\ \hline
RDD (Robust) & 0.043   & 0.051 & 0.861    & 0.175 & 0.169          & 0.318         \\
Diff-in-Disc & -0.004  & 0.075 & 0.929    & 0.360 & 0.223          & 0.360         \\
NP-DiD       & 0.001   & 0.029 & 0.950    & 0.343 & -              & -             \\
TWFE         & 0.001   & 0.018 & 1        & 0.403 & -              & -             \\ \hline
\end{tabular}\\\scriptsize \noindent \textit{Note:} Simulations based on 10,000 Monte Carlo experiments with a sample size $n=1000$. RDD is the non-bias corrected RD estimator from CCT (2014), NP-DiD is the outcome regression DiD estimator from Sant'Anna and Zhao (2020), DiDC is the estimator proposed in this paper.  “Av. Bias”, “RMSE”, “Cover”, “CIL’ stand for the average simulated bias, simulated root mean-squared errors, 95\% coverage probability and 95\% confidence interval length, respectively. "Bandwidths" $h_n$ and $b_n$ report the plug-in bandwidths for point and bias estimation, respectively.
\end{table}

Additional simulations results from DGPs with time-varying standard deviations are displayd in Tables D.5 and D.6.

\begin{table}[h]
\centering
\caption{Time-Varying Functional Forms with Confounding (more variance in $t=1$)}
\begin{tabular}{c|cccccc}
\hline
             & AV.Bias & RMSE  & Coverage & CIL   & \multicolumn{2}{c}{Bandwidths} \\
             &         &       &          &       & $h_{n}$             & $b_{n}$            \\ \hline
RDD (Robust) & 1.038   & 1.038 & 0        & 0.321 & 0.221          & 0.360         \\
Diff-in-Disc & -0.010  & 0.077 & 0.935    & 0.370 & 0.202          & 0.346         \\
NP-DiD       & 1.767   & 1.767 & 0        & 0.317 & -              & -             \\
TWFE         & -1.568  & 1.568 & 0        & 0.281 & -              & -             \\ \hline
\end{tabular}\\\scriptsize \noindent \textit{Note:} Simulations based on 10,000 Monte Carlo experiments with a sample size $n=1000$. RDD is the non-bias corrected RD estimator from CCT (2014), NP-DiD is the outcome regression DiD estimator from Sant'Anna and Zhao (2020), DiDC is the estimator proposed in this paper.  “Av. Bias”, “RMSE”, “Cover”, “CIL’ stand for the average simulated bias, simulated root mean-squared errors, 95\% coverage probability and 95\% confidence interval length, respectively. "Bandwidths" $h_n$ and $b_n$ report the plug-in bandwidths for point and bias estimation, respectively.
\end{table}

\begin{table}[h]
\centering
\caption{Time-Varying Functional Forms with Confounding (more variance in $t=0$)}
\begin{tabular}{c|cccccc}
\hline
             & AV.Bias & RMSE  & Coverage & CIL   & \multicolumn{2}{c}{Bandwidths} \\
             &         &       &          &       & $h_{n}$             & $b_{n}$            \\ \hline
RDD (Robust) & 1.039   & 1.039 & 0        & 0.161 & 0.221          & 0.360         \\
Diff-in-Disc & -0.011  & 0.077 & 0.935    & 0.371 & 0.202          & 0.346         \\
NP-DiD       & 1.767   & 1.767 & 0        & 0.317 & -              & -             \\
TWFE         & -1.567  & 1.567 & 0        & 0.282 & -              & -             \\ \hline
\end{tabular}\\\scriptsize \noindent \textit{Note:} Simulations based on 10,000 Monte Carlo experiments with a sample size $n=1000$. RDD is the non-bias corrected RD estimator from CCT (2014), NP-DiD is the outcome regression DiD estimator from Sant'Anna and Zhao (2020), DiDC is the estimator proposed in this paper.  “Av. Bias”, “RMSE”, “Cover”, “CIL’ stand for the average simulated bias, simulated root mean-squared errors, 95\% coverage probability and 95\% confidence interval length, respectively. "Bandwidths" $h_n$ and $b_n$ report the plug-in bandwidths for point and bias estimation, respectively.
\end{table}

\subsubsection{Model 4: Different DGPs}

Here the model for $t=0$ follows precisely model 3 from \cite{calonico2014robust} and the model for $t=1$ is linear. No discontinuities that could affect the outcome exist at the threshold $z_0=0$ other than the treatment effect $\tau$ introduced for $t=1$. The data generating functions $\mu_0(z)$ and $\mu_1(z)$ are as follows
\begin{align*}
    \mu_{0}(z)&=\begin{cases} 
            0.48+1.27z-0.5\cdot7.18z^{2} \\ + 0.7\cdot20.21z^{3}+1.1\cdot21.54z^{4}+1.5\cdot7.33z^{5} & \text{if } z<0\\ 
            0.52+0.84z-0.1\cdot3z^{2} \\  - 0.3\cdot7.99z^{3}-0.1\cdot9.01z^{4}+3.56z^{5}+\textbf{c} & \text{if } z\geq 0
        \end{cases} \\
    \mu_{1}(z)&=\begin{cases} 
            0.48+1.4z & \text{if } z<0\\ 
            0.52+0.1z +\tau & \text{if } z\geq 0
        \end{cases}
\end{align*}

Additional simulation results from DGPs with time-varying standard deviations are displayed in Tables D.7 and D.8.

\begin{table}[]
\centering
\caption{Time-Varying Functional Forms without Confounding (more variance in $t=0)$}
\begin{tabular}{c|cccccc}
\hline
             & AV.Bias & RMSE  & Coverage & CIL   & \multicolumn{2}{c}{Bandwidths} \\
             &         &       &          &       & $h_{n}$             & $b_{n}$            \\ \hline
RDD (Robust) & 0.039   & 0.049 & 0.922    & 0.161 & 0.221          & 0.360         \\
Diff-in-Disc & -0.011  & 0.077 & 0.935    & 0.371 & 0.202          & 0.346         \\
NP-DiD       & 1.767   & 1.767 & 0        & 0.317 & -              & -             \\
TWFE         & -1.567  & 1.567 & 0        & 0.282 & -              & -             \\ \hline
\end{tabular}\\\scriptsize \noindent \textit{Note:} Simulations based on 10,000 Monte Carlo experiments with a sample size $n=1000$. RDD is the non-bias corrected RD estimator from CCT (2014), NP-DiD is the outcome regression DiD estimator from Sant'Anna and Zhao (2020), DiDC is the estimator proposed in this paper.  “Av. Bias”, “RMSE”, “Cover”, “CIL’ stand for the average simulated bias, simulated root mean-squared errors, 95\% coverage probability and 95\% confidence interval length, respectively. "Bandwidths" $h_n$ and $b_n$ report the plug-in bandwidths for point and bias estimation, respectively.
\end{table}

\begin{table}[h]
\centering
\caption{Time-Varying Functional Forms without Confounding (more variance in $t=1)$}
\begin{tabular}{c|cccccc}
\hline
             & AV.Bias & RMSE  & Coverage & CIL   & \multicolumn{2}{c}{Bandwidths} \\
             &         &       &          &       & $h_{n}$             & $b_{n}$            \\ \hline
RDD (Robust) & 0.038   & 0.078 & 0.917    & 0.321 & 0.221          & 0.360         \\
Diff-in-Disc & -0.010  & 0.077 & 0.935    & 0.370 & 0.202          & 0.346         \\
NP-DiD       & 1.767   & 1.767 & 0        & 0.317 & -              & -             \\
TWFE         & -1.568  & 1.568 & 0        & 0.281 & -              & -             \\ \hline
\end{tabular}\\\scriptsize \noindent \textit{Note:} Simulations based on 10,000 Monte Carlo experiments with a sample size $n=1000$. RDD is the non-bias corrected RD estimator from CCT (2014), NP-DiD is the outcome regression DiD estimator from Sant'Anna and Zhao (2020), DiDC is the estimator proposed in this paper.  “Av. Bias”, “RMSE”, “Cover”, “CIL’ stand for the average simulated bias, simulated root mean-squared errors, 95\% coverage probability and 95\% confidence interval length, respectively. "Bandwidths" $h_n$ and $b_n$ report the plug-in bandwidths for point and bias estimation, respectively.
\end{table}

\newpage
\section{Empirical illustration} \label{app:Empirical}

We replicate the analysis conducted by \cite{grembi2016fiscal} once more, this time using the bandwidth selection method proposed by \cite{ludwig2007does}  as she did. The use of these alternative bandwidths does not significantly alter the estimated treatment effects; they remain similar in magnitude.

\begin{table}[h]\label{tab_Grembi_LMbw}
    \centering
    \caption{Effects of relaxing a fiscal rule - Ludwig \& Miller (2007) Bandwidths}
    \begin{tabular}{cccc}
        \hline
        Estimators   & Deficit & Fiscal Gap & Taxes \\ \hline
        Diif-in-Disc (Grembi et al, 2016) &
          \begin{tabular}[c]{@{}c@{}}9.454\\ (4.343)\end{tabular} &
          \begin{tabular}[c]{@{}c@{}}48.469\\ (23.315)\end{tabular} &
          \begin{tabular}[c]{@{}c@{}}-34.748\\ (20.166)\end{tabular} \\
        Bandwidth    & 1498     & 833        & 684   \\
        Observations & 5858    & 3438       & 1536  \\ \hline
        Diff-in-Disc (RDD of $\Delta$) &
          \begin{tabular}[c]{@{}c@{}}-18.501\\ (22.318)\end{tabular} &
          \begin{tabular}[c]{@{}c@{}}-10.247\\ (21.724)\end{tabular} &
          \begin{tabular}[c]{@{}c@{}}-5.288\\ (4.803)\end{tabular} \\
        Bandwidth    & 520     & 591        & 591   \\
        Observations & 392    & 433       & 434  \\ \hline
        Diff-in-Disc ($\Delta$ of RDDs) &
          \begin{tabular}[c]{@{}c@{}}-22.142\\ (25.136)\end{tabular} &
          \begin{tabular}[c]{@{}c@{}}-0/921\\ (46.657)\end{tabular} &
          \begin{tabular}[c]{@{}c@{}}-9.009\\ (23.824)\end{tabular} \\
        Bandwidth    & 560     & 475       & 334   \\
        Observations & 415    & 362       & 264  \\ \hline
    \end{tabular}
\end{table}

\subsection{Partial identification}
\subsubsection{Results for Bounded Variation}

\begin{table}[h]
\centering
\caption{Bounds on the Treatment Effect under Bounded Variation (\textit{Outcome: Fiscal Gap})}
\label{tab:partial_id_fiscal}
\resizebox{\textwidth}{!}{%
\begin{tabular}{c|ccccccccccc}
\toprule
$c_{2} \backslash c_{1}$ & 0.0 & 0.5 & 1.0 & 1.5 & 2.0 & 2.5 & 3.0 & 3.5 & 4.0 & 4.5 & 5.0  \\
\midrule
0.0 & 19.15 , -10.24 & 18.65 , -10.24 & 18.15 , -10.24 & 17.65 , -10.24 & 17.15 , -10.24 & 16.65 , -10.24 & 16.15 , -10.24 & 15.65 , -10.24 & 15.15 , -10.24 & 14.65 , -10.24 & 14.15 , -10.24 \\
0.5 & 19.15 , -9.74  & 18.65 , -9.74  & 18.15 , -9.74  & 17.65 , -9.74  & 17.15 , -9.74  & 16.65 , -9.74  & 16.15 , -9.74  & 15.65 , -9.74  & 15.15 , -9.74  & 14.65 , -9.74  & 14.15 , -9.74  \\
1.0 & 19.15 , -9.24  & 18.65 , -9.24  & 18.15 , -9.24  & 17.65 , -9.24  & 17.15 , -9.24  & 16.65 , -9.24  & 16.15 , -9.24  & 15.65 , -9.24  & 15.15 , -9.24  & 14.65 , -9.24  & 14.15 , -9.24  \\
1.5 & 19.15 , -8.74  & 18.65 , -8.74  & 18.15 , -8.74  & 17.65 , -8.74  & 17.15 , -8.74  & 16.65 , -8.74  & 16.15 , -8.74  & 15.65 , -8.74  & 15.15 , -8.74  & 14.65 , -8.74  & 14.15 , -8.74  \\
2.0 & 19.15 , -8.24  & 18.65 , -8.24  & 18.15 , -8.24  & 17.65 , -8.24  & 17.15 , -8.24  & 16.65 , -8.24  & 16.15 , -8.24  & 15.65 , -8.24  & 15.15 , -8.24  & 14.65 , -8.24  & 14.15 , -8.24  \\
2.5 & 19.15 , -7.74  & 18.65 , -7.74  & 18.15 , -7.74  & 17.65 , -7.74  & 17.15 , -7.74  & 16.65 , -7.74  & 16.15 , -7.74  & 15.65 , -7.74  & 15.15 , -7.74  & 14.65 , -7.74  & 14.15 , -7.74  \\
3.0 & 19.15 , -7.24  & 18.65 , -7.24  & 18.15 , -7.24  & 17.65 , -7.24  & 17.15 , -7.24  & 16.65 , -7.24  & 16.15 , -7.24  & 15.65 , -7.24  & 15.15 , -7.24  & 14.65 , -7.24  & 14.15 , -7.24  \\
3.5 & 19.15 , -6.74  & 18.65 , -6.74  & 18.15 , -6.74  & 17.65 , -6.74  & 17.15 , -6.74  & 16.65 , -6.74  & 16.15 , -6.74  & 15.65 , -6.74  & 15.15 , -6.74  & 14.65 , -6.74  & 14.15 , -6.74  \\
4.0 & 19.15 , -6.24  & 18.65 , -6.24  & 18.15 , -6.24  & 17.65 , -6.24  & 17.15 , -6.24  & 16.65 , -6.24  & 16.15 , -6.24  & 15.65 , -6.24  & 15.15 , -6.24  & 14.65 , -6.24  & 14.15 , -6.24  \\
4.5 & 19.15 , -5.74  & 18.65 , -5.74  & 18.15 , -5.74  & 17.65 , -5.74  & 17.15 , -5.74  & 16.65 , -5.74  & 16.15 , -5.74  & 15.65 , -5.74  & 15.15 , -5.74  & 14.65 , -5.74  & 14.15 , -5.74  \\
5.0 & 19.15 , -5.24  & 18.65 , -5.24  & 18.15 , -5.24  & 17.65 , -5.24  & 17.15 , -5.24  & 16.65 , -5.24  & 16.15 , -5.24  & 15.65 , -5.24  & 15.15 , -5.24  & 14.65 , -5.24  & 14.15 , -5.24  \\
\bottomrule
\end{tabular}%
}
\end{table}

\begin{table}[h]
\centering
\caption{Bounds on the Treatment Effect under Bounded Variation (\textit{Outcome: Deficit})}
\label{tab:partial_id_deficit}
\resizebox{\textwidth}{!}{%
\begin{tabular}{c|ccccccccccc}
\toprule
$c_{2} \backslash c_{1}$ & 0.0 & 0.5 & 1.0 & 1.5 & 2.0 & 2.5 & 3.0 & 3.5 & 4.0 & 4.5 & 5.0  \\
\midrule
0.0 & -19.60 , -27.86 & -19.60 , -27.36 & -19.60 , -26.86 & -19.60 , -26.36 & -19.60 , -25.86 & -19.60 , -25.36 & -19.60 , -24.86 & -19.60 , -24.36 & -19.60 , -23.86 & -19.60 , -23.36 & -19.60 , -22.86 \\
0.5 & -20.10 , -27.86 & -20.10 , -27.36 & -20.10 , -26.86 & -20.10 , -26.36 & -20.10 , -25.86 & -20.10 , -25.36 & -20.10 , -24.86 & -20.10 , -24.36 & -20.10 , -23.86 & -20.10 , -23.36 & -20.10 , -22.86 \\
1.0 & -20.60 , -27.86 & -20.60 , -27.36 & -20.60 , -26.86 & -20.60 , -26.36 & -20.60 , -25.86 & -20.60 , -25.36 & -20.60 , -24.86 & -20.60 , -24.36 & -20.60 , -23.86 & -20.60 , -23.36 & -20.60 , -22.86 \\
1.5 & -21.10 , -27.86 & -21.10 , -27.36 & -21.10 , -26.86 & -21.10 , -26.36 & -21.10 , -25.86 & -21.10 , -25.36 & -21.10 , -24.86 & -21.10 , -24.36 & -21.10 , -23.86 & -21.10 , -23.36 & -21.10 , -22.86 \\
2.0 & -21.60 , -27.86 & -21.60 , -27.36 & -21.60 , -26.86 & -21.60 , -26.36 & -21.60 , -25.86 & -21.60 , -25.36 & -21.60 , -24.86 & -21.60 , -24.36 & -21.60 , -23.86 & -21.60 , -23.36 & -21.60 , -22.86 \\
2.5 & -22.10 , -27.86 & -22.10 , -27.36 & -22.10 , -26.86 & -22.10 , -26.36 & -22.10 , -25.86 & -22.10 , -25.36 & -22.10 , -24.86 & -22.10 , -24.36 & -22.10 , -23.86 & -22.10 , -23.36 & -22.10 , -22.86 \\
3.0 & -22.60 , -27.86 & -22.60 , -27.36 & -22.60 , -26.86 & -22.60 , -26.36 & -22.60 , -25.86 & -22.60 , -25.36 & -22.60 , -24.86 & -22.60 , -24.36 & -22.60 , -23.86 & -22.60 , -23.36 & -22.60 , -22.86 \\
3.5 & -23.10 , -27.86 & -23.10 , -27.36 & -23.10 , -26.86 & -23.10 , -26.36 & -23.10 , -25.86 & -23.10 , -25.36 & -23.10 , -24.86 & -23.10 , -24.36 & -23.10 , -23.86 & -23.10 , -23.36 & -23.10 , -22.86 \\
4.0 & -23.60 , -27.86 & -23.60 , -27.36 & -23.60 , -26.86 & -23.60 , -26.36 & -23.60 , -25.86 & -23.60 , -25.36 & -23.60 , -24.86 & -23.60 , -24.36 & -23.60 , -23.86 & -23.60 , -23.36 & -23.60 , -22.86 \\
4.5 & -24.10 , -27.86 & -24.10 , -27.36 & -24.10 , -26.86 & -24.10 , -26.36 & -24.10 , -25.86 & -24.10 , -25.36 & -24.10 , -24.86 & -24.10 , -24.36 & -24.10 , -23.86 & -24.10 , -23.36 & -24.10 , -22.86 \\
5.0 & -24.60 , -27.86 & -24.60 , -27.36 & -24.60 , -26.86 & -24.60 , -26.36 & -24.60 , -25.86 & -24.60 , -25.36 & -24.60 , -24.86 & -24.60 , -24.36 & -24.60 , -23.86 & -24.60 , -23.36 & -24.60 , -22.86 \\
\bottomrule
\end{tabular}%
}
\end{table}

\begin{table}[h]
\centering
\caption{Bounds for $\tau_c$ — Outcome: Taxes. $y_{min} = 0$, $y_{max} = 1203.185$}
\label{tab:partial_comb_taxes}
\resizebox{\textwidth}{!}{%
\begin{tabular}{c|ccccccccccc}
\toprule
$c_{2}\backslash c_{1}$ & 0.0 & 0.5 & 1.0 & 1.5 & 2.0 & 2.5 & 3.0 & 3.5 & 4.0 & 4.5 & 5.0 \\
\midrule
0.0 & 6.33 , -6.57 & 5.83 , -6.57 & 5.33 , -6.57 & 4.83 , -6.57 & 4.33 , -6.57 & 3.83 , -6.57 & 3.33 , -6.57 & 2.83 , -6.57 & 2.33 , -6.57 & 1.83 , -6.57 & 1.33 , -6.57 \\
0.5 & 6.33 , -6.07 & 5.83 , -6.07 & 5.33 , -6.07 & 4.83 , -6.07 & 4.33 , -6.07 & 3.83 , -6.07 & 3.33 , -6.07 & 2.83 , -6.07 & 2.33 , -6.07 & 1.83 , -6.07 & 1.33 , -6.07 \\
1.0 & 6.33 , -5.57 & 5.83 , -5.57 & 5.33 , -5.57 & 4.83 , -5.57 & 4.33 , -5.57 & 3.83 , -5.57 & 3.33 , -5.57 & 2.83 , -5.57 & 2.33 , -5.57 & 1.83 , -5.57 & 1.33 , -5.57 \\
1.5 & 6.33 , -5.07 & 5.83 , -5.07 & 5.33 , -5.07 & 4.83 , -5.07 & 4.33 , -5.07 & 3.83 , -5.07 & 3.33 , -5.07 & 2.83 , -5.07 & 2.33 , -5.07 & 1.83 , -5.07 & 1.33 , -5.07 \\
2.0 & 6.33 , -4.57 & 5.83 , -4.57 & 5.33 , -4.57 & 4.83 , -4.57 & 4.33 , -4.57 & 3.83 , -4.57 & 3.33 , -4.57 & 2.83 , -4.57 & 2.33 , -4.57 & 1.83 , -4.57 & 1.33 , -4.57 \\
2.5 & 6.33 , -4.07 & 5.83 , -4.07 & 5.33 , -4.07 & 4.83 , -4.07 & 4.33 , -4.07 & 3.83 , -4.07 & 3.33 , -4.07 & 2.83 , -4.07 & 2.33 , -4.07 & 1.83 , -4.07 & 1.33 , -4.07 \\
3.0 & 6.33 , -3.57 & 5.83 , -3.57 & 5.33 , -3.57 & 4.83 , -3.57 & 4.33 , -3.57 & 3.83 , -3.57 & 3.33 , -3.57 & 2.83 , -3.57 & 2.33 , -3.57 & 1.83 , -3.57 & 1.33 , -3.57 \\
3.5 & 6.33 , -3.07 & 5.83 , -3.07 & 5.33 , -3.07 & 4.83 , -3.07 & 4.33 , -3.07 & 3.83 , -3.07 & 3.33 , -3.07 & 2.83 , -3.07 & 2.33 , -3.07 & 1.83 , -3.07 & 1.33 , -3.07 \\
4.0 & 6.33 , -2.57 & 5.83 , -2.57 & 5.33 , -2.57 & 4.83 , -2.57 & 4.33 , -2.57 & 3.83 , -2.57 & 3.33 , -2.57 & 2.83 , -2.57 & 2.33 , -2.57 & 1.83 , -2.57 & 1.33 , -2.57 \\
4.5 & 6.33 , -2.07 & 5.83 , -2.07 & 5.33 , -2.07 & 4.83 , -2.07 & 4.33 , -2.07 & 3.83 , -2.07 & 3.33 , -2.07 & 2.83 , -2.07 & 2.33 , -2.07 & 1.83 , -2.07 & 1.33 , -2.07 \\
5.0 & 6.33 , -1.57 & 5.83 , -1.57 & 5.33 , -1.57 & 4.83 , -1.57 & 4.33 , -1.57 & 3.83 , -1.57 & 3.33 , -1.57 & 2.83 , -1.57 & 2.33 , -1.57 & 1.83 , -1.57 & 1.33 , -1.57 \\
\bottomrule
\end{tabular}%
}
\end{table}

\begin{table}[h]
\centering
\caption{Bounds for $\tau_{uc}$ — Outcome: Taxes. $y_{min} = 0$, $y_{max} = 1203.185$}
\label{tab:partial_comb_tauuc_taxes}
\resizebox{\textwidth}{!}{%
\begin{tabular}{c|ccccccccccc}
\toprule
$c_{2}\backslash c_{1}$ & 0.0 & 0.5 & 1.0 & 1.5 & 2.0 & 2.5 & 3.0 & 3.5 & 4.0 & 4.5 & 5.0 \\
\midrule
0.0 & -1203.19 , 6.33 & -1203.19 , 6.83 & -1203.19 , 7.33 & -1203.19 , 7.83 & -1203.19 , 8.33 & -1203.19 , 8.83 & -1203.19 , 9.33 & -1203.19 , 9.83 & -1203.19 , 10.33 & -1203.19 , 10.83 & -1203.19 , 11.33 \\
0.5 & -1203.19 , 6.33 & -1203.19 , 6.83 & -1203.19 , 7.33 & -1203.19 , 7.83 & -1203.19 , 8.33 & -1203.19 , 8.83 & -1203.19 , 9.33 & -1203.19 , 9.83 & -1203.19 , 10.33 & -1203.19 , 10.83 & -1203.19 , 11.33 \\
1.0 & -1203.19 , 6.33 & -1203.19 , 6.83 & -1203.19 , 7.33 & -1203.19 , 7.83 & -1203.19 , 8.33 & -1203.19 , 8.83 & -1203.19 , 9.33 & -1203.19 , 9.83 & -1203.19 , 10.33 & -1203.19 , 10.83 & -1203.19 , 11.33 \\
1.5 & -1203.19 , 6.33 & -1203.19 , 6.83 & -1203.19 , 7.33 & -1203.19 , 7.83 & -1203.19 , 8.33 & -1203.19 , 8.83 & -1203.19 , 9.33 & -1203.19 , 9.83 & -1203.19 , 10.33 & -1203.19 , 10.83 & -1203.19 , 11.33 \\
2.0 & -1203.19 , 6.33 & -1203.19 , 6.83 & -1203.19 , 7.33 & -1203.19 , 7.83 & -1203.19 , 8.33 & -1203.19 , 8.83 & -1203.19 , 9.33 & -1203.19 , 9.83 & -1203.19 , 10.33 & -1203.19 , 10.83 & -1203.19 , 11.33 \\
2.5 & -1203.19 , 6.33 & -1203.19 , 6.83 & -1203.19 , 7.33 & -1203.19 , 7.83 & -1203.19 , 8.33 & -1203.19 , 8.83 & -1203.19 , 9.33 & -1203.19 , 9.83 & -1203.19 , 10.33 & -1203.19 , 10.83 & -1203.19 , 11.33 \\
3.0 & -1203.19 , 6.33 & -1203.19 , 6.83 & -1203.19 , 7.33 & -1203.19 , 7.83 & -1203.19 , 8.33 & -1203.19 , 8.83 & -1203.19 , 9.33 & -1203.19 , 9.83 & -1203.19 , 10.33 & -1203.19 , 10.83 & -1203.19 , 11.33 \\
3.5 & -1203.19 , 6.33 & -1203.19 , 6.83 & -1203.19 , 7.33 & -1203.19 , 7.83 & -1203.19 , 8.33 & -1203.19 , 8.83 & -1203.19 , 9.33 & -1203.19 , 9.83 & -1203.19 , 10.33 & -1203.19 , 10.83 & -1203.19 , 11.33 \\
4.0 & -1203.19 , 6.33 & -1203.19 , 6.83 & -1203.19 , 7.33 & -1203.19 , 7.83 & -1203.19 , 8.33 & -1203.19 , 8.83 & -1203.19 , 9.33 & -1203.19 , 9.83 & -1203.19 , 10.33 & -1203.19 , 10.83 & -1203.19 , 11.33 \\
4.5 & -1203.19 , 6.33 & -1203.19 , 6.83 & -1203.19 , 7.33 & -1203.19 , 7.83 & -1203.19 , 8.33 & -1203.19 , 8.83 & -1203.19 , 9.33 & -1203.19 , 9.83 & -1203.19 , 10.33 & -1203.19 , 10.83 & -1203.19 , 11.33 \\
5.0 & -1203.19 , 6.33 & -1203.19 , 6.83 & -1203.19 , 7.33 & -1203.19 , 7.83 & -1203.19 , 8.33 & -1203.19 , 8.83 & -1203.19 , 9.33 & -1203.19 , 9.83 & -1203.19 , 10.33 & -1203.19 , 10.83 & -1203.19 , 11.33 \\
\bottomrule
\end{tabular}%
}
\end{table}

\begin{table}[h]
\centering
\caption{Bounds for $\tau_c$ — Outcome: Fiscal balance. $y_{min} = -591.196$, $y_{max} = 1766.979$}
\label{tab:partial_comb_fiscal}
\resizebox{\textwidth}{!}{%
\begin{tabular}{c|ccccccccccc}
\toprule
$c_{2}\backslash c_{1}$ & 0.0 & 0.5 & 1.0 & 1.5 & 2.0 & 2.5 & 3.0 & 3.5 & 4.0 & 4.5 & 5.0 \\
\midrule
0.0 & 19.15 , -10.24 & 18.65 , -10.24 & 18.15 , -10.24 & 17.65 , -10.24 & 17.15 , -10.24 & 16.65 , -10.24 & 16.15 , -10.24 & 15.65 , -10.24 & 15.15 , -10.24 & 14.65 , -10.24 & 14.15 , -10.24 \\
0.5 & 19.15 , -9.74  & 18.65 , -9.74  & 18.15 , -9.74  & 17.65 , -9.74  & 17.15 , -9.74  & 16.65 , -9.74  & 16.15 , -9.74  & 15.65 , -9.74  & 15.15 , -9.74  & 14.65 , -9.74  & 14.15 , -9.74  \\
1.0 & 19.15 , -9.24  & 18.65 , -9.24  & 18.15 , -9.24  & 17.65 , -9.24  & 17.15 , -9.24  & 16.65 , -9.24  & 16.15 , -9.24  & 15.65 , -9.24  & 15.15 , -9.24  & 14.65 , -9.24  & 14.15 , -9.24  \\
1.5 & 19.15 , -8.74  & 18.65 , -8.74  & 18.15 , -8.74  & 17.65 , -8.74  & 17.15 , -8.74  & 16.65 , -8.74  & 16.15 , -8.74  & 15.65 , -8.74  & 15.15 , -8.74  & 14.65 , -8.74  & 14.15 , -8.74  \\
2.0 & 19.15 , -8.24  & 18.65 , -8.24  & 18.15 , -8.24  & 17.65 , -8.24  & 17.15 , -8.24  & 16.65 , -8.24  & 16.15 , -8.24  & 15.65 , -8.24  & 15.15 , -8.24  & 14.65 , -8.24  & 14.15 , -8.24  \\
2.5 & 19.15 , -7.74  & 18.65 , -7.74  & 18.15 , -7.74  & 17.65 , -7.74  & 17.15 , -7.74  & 16.65 , -7.74  & 16.15 , -7.74  & 15.65 , -7.74  & 15.15 , -7.74  & 14.65 , -7.74  & 14.15 , -7.74  \\
3.0 & 19.15 , -7.24  & 18.65 , -7.24  & 18.15 , -7.24  & 17.65 , -7.24  & 17.15 , -7.24  & 16.65 , -7.24  & 16.15 , -7.24  & 15.65 , -7.24  & 15.15 , -7.24  & 14.65 , -7.24  & 14.15 , -7.24  \\
3.5 & 19.15 , -6.74  & 18.65 , -6.74  & 18.15 , -6.74  & 17.65 , -6.74  & 17.15 , -6.74  & 16.65 , -6.74  & 16.15 , -6.74  & 15.65 , -6.74  & 15.15 , -6.74  & 14.65 , -6.74  & 14.15 , -6.74  \\
4.0 & 19.15 , -6.24  & 18.65 , -6.24  & 18.15 , -6.24  & 17.65 , -6.24  & 17.15 , -6.24  & 16.65 , -6.24  & 16.15 , -6.24  & 15.65 , -6.24  & 15.15 , -6.24  & 14.65 , -6.24  & 14.15 , -6.24  \\
4.5 & 19.15 , -5.74  & 18.65 , -5.74  & 18.15 , -5.74  & 17.65 , -5.74  & 17.15 , -5.74  & 16.65 , -5.74  & 16.15 , -5.74  & 15.65 , -5.74  & 15.15 , -5.74  & 14.65 , -5.74  & 14.15 , -5.74  \\
5.0 & 19.15 , -5.24  & 18.65 , -5.24  & 18.15 , -5.24  & 17.65 , -5.24  & 17.15 , -5.24  & 16.65 , -5.24  & 16.15 , -5.24  & 15.65 , -5.24  & 15.15 , -5.24  & 14.65 , -5.24  & 14.15 , -5.24  \\
\bottomrule
\end{tabular}%
}
\end{table}

\begin{table}[h]
\centering
\caption{Bounds for $\tau_{uc}$ — Outcome: Fiscal balance. $y_{min} = -591.196$, $y_{max} = 1766.979$.}
\label{tab:partial_comb_tauuc_fiscal}
\resizebox{\textwidth}{!}{%
\begin{tabular}{c|ccccccccccc}
\toprule
$c_{2}\backslash c_{1}$ & 0.0 & 0.5 & 1.0 & 1.5 & 2.0 & 2.5 & 3.0 & 3.5 & 4.0 & 4.5 & 5.0 \\
\midrule
0.0 & -2358.18 , 19.15 & -2358.18 , 19.15 & -2358.18 , 19.15 & -2358.18 , 19.15 & -2358.18 , 19.15 & -2358.18 , 19.15 & -2358.18 , 19.15 & -2358.18 , 19.15 & -2358.18 , 19.15 & -2358.18 , 19.15 & -2358.18 , 19.15 \\
0.5 & -2358.18 , 19.65 & -2358.18 , 19.65 & -2358.18 , 19.65 & -2358.18 , 19.65 & -2358.18 , 19.65 & -2358.18 , 19.65 & -2358.18 , 19.65 & -2358.18 , 19.65 & -2358.18 , 19.65 & -2358.18 , 19.65 & -2358.18 , 19.65 \\
1.0 & -2358.18 , 20.15 & -2358.18 , 20.15 & -2358.18 , 20.15 & -2358.18 , 20.15 & -2358.18 , 20.15 & -2358.18 , 20.15 & -2358.18 , 20.15 & -2358.18 , 20.15 & -2358.18 , 20.15 & -2358.18 , 20.15 & -2358.18 , 20.15 \\
1.5 & -2358.18 , 20.65 & -2358.18 , 20.65 & -2358.18 , 20.65 & -2358.18 , 20.65 & -2358.18 , 20.65 & -2358.18 , 20.65 & -2358.18 , 20.65 & -2358.18 , 20.65 & -2358.18 , 20.65 & -2358.18 , 20.65 & -2358.18 , 20.65 \\
2.0 & -2358.18 , 21.15 & -2358.18 , 21.15 & -2358.18 , 21.15 & -2358.18 , 21.15 & -2358.18 , 21.15 & -2358.18 , 21.15 & -2358.18 , 21.15 & -2358.18 , 21.15 & -2358.18 , 21.15 & -2358.18 , 21.15 & -2358.18 , 21.15 \\
2.5 & -2358.18 , 21.65 & -2358.18 , 21.65 & -2358.18 , 21.65 & -2358.18 , 21.65 & -2358.18 , 21.65 & -2358.18 , 21.65 & -2358.18 , 21.65 & -2358.18 , 21.65 & -2358.18 , 21.65 & -2358.18 , 21.65 & -2358.18 , 21.65 \\
3.0 & -2358.18 , 22.15 & -2358.18 , 22.15 & -2358.18 , 22.15 & -2358.18 , 22.15 & -2358.18 , 22.15 & -2358.18 , 22.15 & -2358.18 , 22.15 & -2358.18 , 22.15 & -2358.18 , 22.15 & -2358.18 , 22.15 & -2358.18 , 22.15 \\
3.5 & -2358.18 , 22.65 & -2358.18 , 22.65 & -2358.18 , 22.65 & -2358.18 , 22.65 & -2358.18 , 22.65 & -2358.18 , 22.65 & -2358.18 , 22.65 & -2358.18 , 22.65 & -2358.18 , 22.65 & -2358.18 , 22.65 & -2358.18 , 22.65 \\
4.0 & -2358.18 , 23.15 & -2358.18 , 23.15 & -2358.18 , 23.15 & -2358.18 , 23.15 & -2358.18 , 23.15 & -2358.18 , 23.15 & -2358.18 , 23.15 & -2358.18 , 23.15 & -2358.18 , 23.15 & -2358.18 , 23.15 & -2358.18 , 23.15 \\
4.5 & -2358.18 , 23.65 & -2358.18 , 23.65 & -2358.18 , 23.65 & -2358.18 , 23.65 & -2358.18 , 23.65 & -2358.18 , 23.65 & -2358.18 , 23.65 & -2358.18 , 23.65 & -2358.18 , 23.65 & -2358.18 , 23.65 & -2358.18 , 23.65 \\
5.0 & -2358.18 , 24.15 & -2358.18 , 24.15 & -2358.18 , 24.15 & -2358.18 , 24.15 & -2358.18 , 24.15 & -2358.18 , 24.15 & -2358.18 , 24.15 & -2358.18 , 24.15 & -2358.18 , 24.15 & -2358.18 , 24.15 & -2358.18 , 24.15 \\
\bottomrule
\end{tabular}%
}
\end{table}

\clearpage
\section{Assumptions for the Bootstrap of the KS type Test Statistics}

In this section, we invoke the assumptions under which the tests presented in Section 4.2 are consistent. For the sake of brevity, we present the assumptions required for the consistency of the bootstrap distribution of $KS_{+}$, the assumptions for $KS_{-}$ are analogous. Under these assumptions, the set of Assumptions in Whang (2001) is satisfied, which allows us to extend its bootstrap procedure for the case of series estimators:

\begin{enumerate}
    \item $\left\{ Y_{i,-1},Y_{i,0},Z_{i}\right\}_{i:Z_{i}\geq z_{0}}$ is i.i.d as $\left(Y_{-1},Y_{0},Z\right)$ and $\mathbb{V}\left ( Y_{-1}|Z \right ),\mathbb{V}\left ( Y_{0}|Z \right )$ are bounded on $\mathcal{S}^{+}(Z)$.
    \item The smallest eigenvalue of $\mathbb{E}\left [ p^{K}(Z_{i})p^{K}(Z_{i})^{'} \right ]$ is bounded away from zero uniformly in $K$.
    \item There exists a sequence of constants $\zeta_{0}(K)$ that satisfies $\sup_{z\in\mathcal{S}^{+}(Z)}\left\| p^{K}(z)\right\|\leq \zeta_{0}(K)$, where $\zeta_{0}(K)^{2}/n^{+}\rightarrow 0$ as $n^{+}\rightarrow\infty$.
    \item There is $\alpha>0$ such that $\sup_{z\in\mathcal{S}^{+}(Z)}\left |\mu_{t}(z)-p^{K}(z)^{'}\gamma^{+}_{t} \right |=O(K^{-\alpha})$, for $t\in\left\{-1,0\right\}$.
    \item As $n^{+}\rightarrow\infty$, $K\rightarrow\infty$ and $K/n^{+}\rightarrow 0$.
    \item $\mathbb{E}\left [ \left ( Y_{-1}-\mu_{-1}(z) \right )^{4}|z \right ]$, $\mathbb{E}\left [ \left ( Y_{0}-\mu_{0}(z) \right )^{4}|z \right ]$ are bounded and $\mathbb{V}\left ( Y_{-1}|z \right )$, $\mathbb{V}\left ( Y_{0}|Z \right )$ are bounded away from zero.
    \item $K^{-\alpha}=o((K/n^{+})^{1/2}).$
\end{enumerate}

\end{appendices}

\end{document}